\documentclass[11pt]{article}
\usepackage[margin=.8in,left=.8in]{geometry}
\usepackage{amsmath}
\usepackage{amsfonts}
\usepackage{amssymb}
\usepackage{amsthm}
\usepackage{mathtools}
\usepackage{subcaption}
\usepackage{graphicx}
\usepackage{color}
\usepackage{xcolor}
\usepackage{xfrac}
\usepackage{stmaryrd}
\usepackage[toc, page]{appendix}

\usepackage[safe]{tipa} 
\usepackage[utf8]{inputenc}
\usepackage{rotating}
\usepackage{hyperref}
\usepackage[all]{xy}

\usepackage{mathtools}
\usepackage{rotating}

\usepackage{stmaryrd}
\usepackage{multirow}

\usepackage{lmodern} 

\usepackage{tikz}
\usetikzlibrary{decorations.pathmorphing}
\usepackage{tikz-cd}
\usetikzlibrary{calc}
\usetikzlibrary{matrix}
\usetikzlibrary{positioning,arrows,shapes}
\usetikzlibrary{intersections,shapes.arrows, calc}
\usetikzlibrary{patterns}
\usetikzlibrary{mindmap} 
\usetikzlibrary[mindmap] 

\usepackage{lipsum}
\usepackage{adjustbox}

\definecolor{redi}{RGB}{255,38,0}
\definecolor{redii}{RGB}{200,50,30}
\definecolor{yellowi}{RGB}{255,251,0}
\definecolor{bluei}{RGB}{0,150,255}
\definecolor{blueii}{RGB}{135,247,210}
\definecolor{blueiii}{RGB}{91,205,250}
\definecolor{blueiv}{RGB}{115,244,253}
\definecolor{bluev}{RGB}{1,58,215}
\definecolor{orangei}{RGB}{220,160, 20}
\definecolor{orangeii}{RGB}{240,90, 10}
\definecolor{yellowii}{RGB}{222,247,100}
\definecolor{greeni}{RGB}{85,102,0}
\definecolor{greenii}{RGB}{20,140,10}
\definecolor{navy}{RGB}{17, 10, 102}
\definecolor{brown}{RGB}{60, 40, 0}
\definecolor{oxford}{RGB}{0, 0, 100}

\definecolor{plum}{rgb}{0.36078, 0.20784, 0.4}
\definecolor{chameleon}{rgb}{0.30588, 0.60392, 0.023529}
\definecolor{cornflower}{rgb}{0.12549, 0.29020, 0.52941}
\definecolor{scarlet}{rgb}{0.8, 0, 0}
\definecolor{brick}{rgb}{0.64314, 0, 0}
\definecolor{sunrise}{rgb}{0.80784, 0.36078, 0}
\definecolor{lightblue}{rgb}{0.15,0.35,0.75}
\definecolor{carolina}{RGB}{153, 186, 221}
\definecolor{darkblue}{rgb}{0.05,0.25,0.65}

\usepackage{multirow}

\usepackage{stmaryrd}    

\usepackage{enumerate} 

\usepackage{helvet}   

\definecolor{greenii}{RGB}{20,140,10}
\definecolor{darkblue}{rgb}{0.05,0.25,0.65}
\definecolor{darkyellow}{rgb}{.934,.934,0}

\usepackage{scalerel} 

\usepackage{mathptmx}
\usepackage{amsmath}
\usepackage{graphicx}

\usepackage{floatflt}  
\usepackage{array}     
\newcolumntype{L}[1]{>{\raggedright\let\newline\\\arraybackslash\hspace{0pt}}m{#1}}
\newcolumntype{C}[1]{>{\centering\let\newline\\\arraybackslash\hspace{0pt}}m{#1}}
\newcolumntype{R}[1]{>{\raggedleft\let\newline\\\arraybackslash\hspace{0pt}}m{#1}}

\newcommand{\gt}{>}

\makeatletter
\newcommand{\raisemath}[1]{\mathpalette{\raisem@th{#1}}}
\newcommand{\raisem@th}[3]{\raisebox{#1}{$#2#3$}}
\makeatother

\usepackage{tabularx}   

\usepackage[new]{old-arrows}   

\newdir{> }{{}*!/10pt/@{>}}

\makeatletter
\newif\if@sup
\newtoks\@sups
\def\append@sup#1{\edef\act{\noexpand\@sups={\the\@sups #1}}\act}%
\def\reset@sup{\@supfalse\@sups={}}%
\def\mk@scripts#1#2{\if #2/ \if@sup ^{\the\@sups}\fi \else%
  \ifx #1_ \if@sup ^{\the\@sups}\reset@sup \fi {}_{#2}%
  \else \append@sup#2 \@suptrue \fi%
  \expandafter\mk@scripts\fi}
\def\tensor#1#2{\reset@sup#1\mk@scripts#2_/}
\def\multiscripts#1#2#3{\reset@sup{}\mk@scripts#1_/#2%
  \reset@sup\mk@scripts#3_/}
\makeatother

\makeatletter
\newbox\slashbox \setbox\slashbox=\hbox{$/$}
\def\itex@pslash#1{\setbox\@tempboxa=\hbox{$#1$}
  \@tempdima=0.5\wd\slashbox \advance\@tempdima 0.5\wd\@tempboxa
  \copy\slashbox \kern-\@tempdima \box\@tempboxa}
\def\slash{\protect\itex@pslash}
\makeatother

\def\clap#1{\hbox to 0pt{\hss#1\hss}}
\def\mathllap{\mathpalette\mathllapinternal}
\def\mathrlap{\mathpalette\mathrlapinternal}
\def\mathclap{\mathpalette\mathclapinternal}
\def\mathllapinternal#1#2{\llap{$\mathsurround=0pt#1{#2}$}}
\def\mathrlapinternal#1#2{\rlap{$\mathsurround=0pt#1{#2}$}}
\def\mathclapinternal#1#2{\clap{$\mathsurround=0pt#1{#2}$}}

\let\oldroot\root
\def\root#1#2{\oldroot #1 \of{#2}}
\renewcommand{\sqrt}[2][]{\oldroot #1 \of{#2}}

\DeclareSymbolFont{symbolsC}{U}{txsyc}{m}{n}
\SetSymbolFont{symbolsC}{bold}{U}{txsyc}{bx}{n}
\DeclareFontSubstitution{U}{txsyc}{m}{n}

\DeclareSymbolFont{stmry}{U}{stmry}{m}{n}
\SetSymbolFont{stmry}{bold}{U}{stmry}{b}{n}

\DeclareFontFamily{OMX}{MnSymbolE}{}
\DeclareSymbolFont{mnomx}{OMX}{MnSymbolE}{m}{n}
\SetSymbolFont{mnomx}{bold}{OMX}{MnSymbolE}{b}{n}
\DeclareFontShape{OMX}{MnSymbolE}{m}{n}{
    <-6>  MnSymbolE5
   <6-7>  MnSymbolE6
   <7-8>  MnSymbolE7
   <8-9>  MnSymbolE8
   <9-10> MnSymbolE9
  <10-12> MnSymbolE10
  <12->   MnSymbolE12}{}

\makeatletter
\def\Decl@Mn@Delim#1#2#3#4{%
  \if\relax\noexpand#1%
    \let#1\undefined
  \fi
  \DeclareMathDelimiter{#1}{#2}{#3}{#4}{#3}{#4}}
\def\Decl@Mn@Open#1#2#3{\Decl@Mn@Delim{#1}{\mathopen}{#2}{#3}}
\def\Decl@Mn@Close#1#2#3{\Decl@Mn@Delim{#1}{\mathclose}{#2}{#3}}
\Decl@Mn@Open{\llangle}{mnomx}{'164}
\Decl@Mn@Close{\rrangle}{mnomx}{'171}
\Decl@Mn@Open{\lmoustache}{mnomx}{'245}
\Decl@Mn@Close{\rmoustache}{mnomx}{'244}
\makeatother

\makeatletter
\DeclareRobustCommand\widecheck[1]{{\mathpalette\@widecheck{#1}}}
\def\@widecheck#1#2{%
    \setbox\z@\hbox{\m@th$#1#2$}%
    \setbox\tw@\hbox{\m@th$#1%
       \widehat{%
          \vrule\@width\z@\@height\ht\z@
          \vrule\@height\z@\@width\wd\z@}$}%
    \dp\tw@-\ht\z@
    \@tempdima\ht\z@ \advance\@tempdima2\ht\tw@ \divide\@tempdima\thr@@
    \setbox\tw@\hbox{%
       \raise\@tempdima\hbox{\scalebox{1}[-1]{\lower\@tempdima\box
\tw@}}}%
    {\ooalign{\box\tw@ \cr \box\z@}}}
\makeatother

\makeatletter
\def\udots{\mathinner{\mkern2mu\raise\p@\hbox{.}
\mkern2mu\raise4\p@\hbox{.}\mkern1mu
\raise7\p@\vbox{\kern7\p@\hbox{.}}\mkern1mu}}
\makeatother

\def\1{{\bf 1}}

\def\<{\langle}
\def\>{\rangle}

\newcommand{\Z}{\ensuremath{\mathbb Z}}

\renewcommand{\(}{\begin{equation}}
\renewcommand{\)}{\end{equation}}
\newcommand{\bea}{\begin{eqnarray*}}
\newcommand{\eea}{\end{eqnarray*}}

\usepackage{cleveref}

\crefformat{section}{\S#2#1#3} 
\crefformat{subsection}{\S#2#1#3}
\crefformat{subsubsection}{\S#2#1#3}

\theoremstyle{italics}
\newtheorem{theorem}{Theorem}
\newtheorem{lemma}[theorem]{Lemma}
\newtheorem{prop}[theorem]{Proposition}

\theoremstyle{definition}
\newtheorem{defn}[theorem]{Definition}
\newtheorem{notation}[theorem]{Notation}
\newtheorem{example}[theorem]{Example}

\newtheorem{remark}[theorem]{Remark}
\newtheorem{note[theorem]}{Note}

\usepackage{amsfonts}
\usepackage{colortbl}

\newcommand{\underoverset}[3]{\underset{#1}{\overset{#2}{#3}}}

\definecolor{darkblue}{rgb}{0.05,0.25,0.65}
\definecolor{darkgreen}{rgb}{0.00,0.85,0.1}
\definecolor{plum}{rgb}{0.36078, 0.20784, 0.4}

\usepackage{hyperref}

\usepackage{amsfonts}

\begin{document}

\title{Super-exceptional embedding construction of the heterotic M5:
\\ Emergence of $\mathrm{SU}(2)$-flavor sector }

\author{Domenico Fiorenza, \; Hisham Sati, \; Urs Schreiber}
%

\maketitle

\begin{abstract}
  A new super-exceptional embedding construction of the heterotic M5-brane's
  sigma-model was recently shown to produce, at leading order in the
  super-exceptional vielbein components, the super-Nambu-Goto
  (Green-Schwarz-type) Lagrangian for the embedding fields plus the 
  Perry-Schwarz Lagrangian for the free abelian self-dual higher gauge field.
  Beyond that, further fields and interactions emerge in the model, 
  arising from
  probe M2- and probe M5-brane wrapping modes.
  Here we classify the full super-exceptional field content
  and work out some of its characteristic interactions from the
  rich super-exceptional Lagrangian of the model.
  We show that $\mathrm{SU}(2) \times U(1)$-valued scalar and
  vector fields emerge from probe M2- and M5-branes
  wrapping the vanishing cycle in the $\mathrm{A}_1$-type singularity;
  together with a pair of spinor fields
  of $U(1)$-hypercharge $\pm 1$ and
  each transforming as $\mathrm{SU}(2)$ iso-doublets.
  Then we highlight the appearance of a WZW-type term in the
  super-exceptional PS-Lagrangian and find that on the
  electromagnetic field it gives the first-order non-linear DBI-correction, while
  on the iso-vector scalar field it has the form characteristic of
  the coupling of vector mesons to pions via the Skyrme baryon current.
  We discuss how this is suggestive of a form of
  $\mathrm{SU}(2)$-flavor chiral hadrodynamics emerging on the
  single ($N=1$) M5 brane, different from, but akin to, holographic large-$N$ QCD.

\end{abstract}

\medskip

\tableofcontents

\medskip

\newpage

\section{Introduction}
\label{Introduction}

\noindent {\bf The general open problem of coincident M5-branes.}
It is widely appreciated that the problem of identifying/formulating
the expected non-abelian higher gauge theory
(``non-abelian gerbe theory'' \cite[p. 6, 15]{Witten02})
on coincident M5-branes remains open
(e.g., \cite[p. 77]{Moore12}\cite[p. 49]{Lambert12}\cite[6.3]{Lambert19}).
 This is a key aspect of the wider open problem
(e.g., \cite[6]{Duff96}\cite[p. 2]{HoweLambertWest97}\cite[p. 6]{Duff98}\cite[p. 2]{NicolaiHelling98}\cite[p. 330]{Duff99}\cite[12]{Moore14}\cite[p. 2]{ChesterPerlmutter18}\cite{Witten19}\cite{Duff19}) of formulating M-theory \cite{Duff98}\cite{Duff99} itself, the non-perturbative completion of
the perturbation theory formerly known as string theory \cite{Duff96}.

M5-branes compactified to four dimensions
with probe flavor D8-branes plausibly yield, once made rigorous,
an analytic first-principles working theory of hadrons
(e.g., \cite{Me88}\cite{MachleidtEntem11}): the Witten-Sakai-Sugimoto model of holographic QCD
\cite{Witten98}\cite{KarchKatz02}\cite{SakaiSugimoto04}\cite{SakaiSugimoto04}
(reviewed in \cite{RhoEtAl16}) or its variant with flavor D4-branes
\cite{VanRaamsdonkWhyte10}\cite{Seki10}.
Hence the solution of non-abelian M5-branes in M-theory
could solve the {\it Confinement Problem} (see \cite{Greensite11})
of quantum chromodynamics (QCD), famously one of the Millennium Problems
\cite{CMI}\cite{JaffeWitten}.

\vspace{.1cm}

\noindent {\bf The case of single but heterotic M5-branes.}
Remarkably, even the case of a single M5-brane,
while commonly thought to have been understood long ago
\cite{PerrySchwarz97}\cite{Schwarz97}\cite{PastiSorokinTonin97}\cite{BLNPST97},
is riddled with subtleties, indicating foundational issues not fully understood yet.
One of these puzzlements is (or was) that the modern ``super-embedding construction''
of $\kappa$-symmetric Lagrangians defining single super $p$-branes
\cite{BPSTV95}\cite{BSV95}\cite{Sorokin99} \emph{fails} for M5-branes
(works only for EOMs).
This is noteworthy because there is a boundary case
of non-abelian gauge theories on M5-branes
which nominally has to do with \emph{single} branes:
the single \emph{heterotic} M5-brane. This is the lift to
heterotic M-theory (Ho{\v r}ava-Witten theory, see \cite{HoravaWitten95}\cite{Witten96}\cite{HoravaWitten96}\cite{Ovrut02}\cite{Ovrut18})
of the NS5-brane of heterotic string theory \cite{Lechner10},
and dually of the D4-brane of Type I' string theory. These branes are expected
\cite{Witten95}\cite{GimonPolchinski96}\cite{AspinwallGross96}\cite{AspinwallMorrison97}
to carry a non-abelian gauge field, specifically
with gauge group $\mathrm{Sp}(1) \simeq \mathrm{SU}(2)$,
understood as being the special case of $N$ coincident
M5-branes for $N=2$, but with the two branes identified as mirror pairs of
an orientifold $\mathbb{Z}_2$-action. The orbi-orientifold singularity of this action,
regarded as the far-horizon geometry of a solitonic brane \cite[3]{AFFHS99},
is known as the $\tfrac{1}{2}\mathrm{NS}5$ in string theory
\cite[6]{GKST01}\cite[6]{DHTV14}\cite[p. 18]{ApruzziFazzi17},
or dually as the
$\tfrac{1}{2}${M}5 in M-theory \cite[Ex. 2.2.7]{HSS18}\cite[4]{FSS19d}\cite[4.1]{SS19a}:

\vspace{-.45cm}

\begin{equation}
 \label{TheBranes}
\raisebox{-275pt}{
 \scalebox{.9}{
 \small
 \begin{tikzpicture}

  \begin{scope}[shift={(0,1.3)}]

  \draw (-2.1,1)
    node
    {
      $
      \underset{
        \phantom{
        \overbrace{
          {\phantom{-----------}}
        }
        }
      }{
        \xymatrix{
          \;\;\;\;\;
          \mathllap{
            \mbox{
              \tiny
              \color{darkblue} \bf
              \begin{tabular}{c}
                bosonic part of
                \\
                super-Minkowski
                \\
               spacetime
              \end{tabular}
            }
            \;
          }
          \big(
            \mathbb{R}^{10,1\vert\mathbf{32}}_{\mathrlap{\phantom{\mathrm{ex}_s}}}
          \big)_{\mathrm{bos}}
          \ar@(ul,ur)^-{
            \mathllap{
              \mbox{
                \tiny
                \color{darkblue} \bf
                \begin{tabular}{c}
                  $\tfrac{1}{2}\mathrm{M5}$--
                  \\
                  orientifolding
                \end{tabular}
              }
              \;
            }
            \mathbb{Z}_2^{\mathrm{HW}}
              \times
            \mathbb{Z}_2^{A}
          }
        }
        =
      }
      $
    };

  \draw (.5,1)
    node
    {
      $
      \underset{
        \overbrace{
          {\phantom{-----}}
        }
      }{
        \xymatrix{
          \mathbb{R}^{1,1}
          \ar@[white]@(ul,ur)^{\phantom{G}}
        }
      }
      $
    };

  \draw (1.5,.85)
    node
    { $\times$ };

  \draw (5.25,.85)
    node
    { $\times$ };

  \draw (2.5,1)
    node
    {
      $
      \underset{
        \overbrace{
          {\phantom{-----}}
        }
      }{
        \xymatrix{
          \mathbb{C}^{\mathrlap{\phantom{1}}}
          \ar@[white]@(ul,ur)^-{
            \phantom{
            1_{\phantom{\mathrlap{ADE}}}
            \mathrlap{\phantom{\langle \mathbf{\Gamma}_{1234}\rangle}}
            }
          }
        }
      }
      $
    };

  \draw (4.5,1)
    node
    {
      $
      \underset{
        \overbrace{
          {\phantom{-----}}
        }
      }{
        \xymatrix{
          \mathbb{C}^{\mathrlap{\phantom{1}}}
          \mathrlap{\mathrm{j}}
          \ar@(ul,ur)^-{
            \mathrm{U}(1)_V
            \phantom{
            \mathrlap{\phantom{\langle \mathbf{\Gamma}_{1234}\rangle}}
            }
          }
        }
      }
      $
    };

  \draw (3.5,.85)
    node
    { $\times$ };

  \draw (8.5,1)
    node
    {
      $
      \underset{
        \overbrace{
          {\phantom{-----------}}
        }
      }{
        \xymatrix{
          \mathbb{H}^{\mathrlap{\phantom{1}}}
          \mathrlap{\ell}
          \ar@(ul,ur)^-{
            \mathbb{Z}_2^A
            \mathrlap{\;\subset \mathrm{SU}(2)_L  \subset \mathrm{SU}(2)_L \times \mathrm{SU}(2)_R}
          }
        }
      }
      $
    };

  \draw (8.5,2.3)
    node
    {
      \mbox{
        \tiny
        \color{darkblue} \bf
        \begin{tabular}{c}
          Atiyah--Witten
          \\
          orbi-folding
        \end{tabular}
      }
    };

  \draw (6+1+.5,.85)
    node
    { $\times$ };

  \draw (6,1)
    node
    {
      $
      \underset{
        \overbrace{
          {\phantom{---}}
        }
      }{
        \xymatrix{
          \mathbb{R}^{\mathrlap{\phantom{1}}}
          \ar@(ul,ur)^-{
            \mathbb{Z}_2^{\mathrm{HW}}
          }
        }
      }
      $
    };

  \draw (6,2.3)
    node
    {
      \mbox{
        \tiny
        \color{darkblue} \bf
        \begin{tabular}{c}
          Ho{\v r}ava--Witten
          \\
          orienti-folding
        \end{tabular}
      }
    };

  \draw (-1,0)
    node
    {
      \footnotesize
      $
        \mathllap{
          \mbox{
            \tiny
            \color{darkblue} \bf
            M-spacetime indices
          }
          \;\;
        }
        a =
      $
    };

  \draw (-1,-.55)
    node
    {
      \footnotesize
      $
        \mathllap{
          \mbox{
            \tiny
            \color{darkblue} \bf
            4d spacetime indices
          }
          \;\;
        }
        \!\mu =
      $
    };

  \draw (-1,-.55-.55-.275)
    node
    {
      \footnotesize
      $
        \mathllap{
          \mbox{
            \tiny
            \color{darkblue} \bf
            internal indices
          }
          \;\;
          \phantom{
            I\, =
          }
        }
      $
    };

  \draw (-1.15,-.55-.55)
    node
    {
      \footnotesize
      \;\;
      $
        z\, =
      $
    };

  \draw (-1.15,-.55-.55-.55)
    node
    {
      \footnotesize
      \;\;
      $
        I\, =
      $
    };

  \draw (0,0)
    node
    {
      \footnotesize
      $0$
    };

  \draw (1,0)
    node
    {
      \footnotesize
      $1$
    };

  \draw (2,0)
    node
    {
      \footnotesize
      $2$
    };

  \draw (3,0)
    node
    {
      \footnotesize
      $3$
    };

  \draw (4,0)
    node
    {
      \footnotesize
      $4$
    };

  \draw (5,0)
    node
    {
      \footnotesize
      $5$
    };

  \draw (6,0)
    node
    {
      \footnotesize
      $5'$
    };

  \draw (7,0)
    node
    {
      \footnotesize
      $6$
    };

  \draw (8,0)
    node
    {
      \footnotesize
      $7$
    };

  \draw (9,0)
    node
    {
      \footnotesize
      $8$
    };

  \draw (10,0)
    node
    {
      \footnotesize
      $9$
    };

  \draw (0,-.55)
    node
    {
      \footnotesize
      $0$
    };

  \draw (1,-.55)
    node
    {
      \footnotesize
      $1$
    };

  \draw (2,-.55)
    node
    {
      \footnotesize
      $2$
    };

  \draw (3,-.55)
    node
    {
      \footnotesize
      $3$
    };

  \draw (4,-.55-.55)
    node
    {
      \footnotesize
      $4$
    };

  \draw (5,-.55-.55)
    node
    {
      \footnotesize
      $5$
    };

  \draw (8,-.55-.55-.55)
    node
    {
      \footnotesize
      $7$
    };

  \draw (9,-.55-.55-.55)
    node
    {
      \footnotesize
      $8$
    };

  \draw (10,-.55-.55-.55)
    node
    {
      \footnotesize
      $9$
    };

  \end{scope}

  \draw (-3,-1.7)
    node
    {
      \mbox{
        \tiny
        \color{darkblue} \bf
        \begin{tabular}{c}
          fixed-point strata {\color{black}/}
          \\
          black branes
        \end{tabular}
      }
    };

  \draw (-3,-3.4)
    node
    {
      \mbox{
        \tiny
        \color{darkblue} \bf
        \begin{tabular}{c}
          sigma-model
          \\
          flavor brane
        \end{tabular}
      }
    };

  \draw (-1,-1)
    node
    {
      $\mathrm{MO}9$
    };

  \draw[darkblue, line width=10pt] (0-.45,-1) to (5+.45,-1);
  \draw[darkblue, line width=10pt] (7-.45,-1) to (10+.45,-1);

  \draw (-1,-1-.55)
    node
    {
      $\mathrm{MK}6$
    };

  \draw[darkyellow, line width=10pt] (0-.45,-1-.55) to (6+.45,-1-.55);

  \draw (-1,-1-.55-.55-.2)
    node
    {
      $\tfrac{1}{2}\mathrm{M}5$
    };

  \draw[gray, line width=10pt, opacity=.4]
    (0-.45,-1-.55-.55-.2) to (4+.45,-1-.55-.55-.2);
  \draw[gray, line width=10pt, opacity=.4]
    (0-.45,-1-.55-.55-.2) to (3+.45,-1-.55-.55-.2);
  \draw[greenii, line width=10pt,opacity=.5]
    (0-.45,-1-.55-.55-.2) to (5+.45,-1-.55-.55-.2);

  \draw (-1,-1-.55-.55-.2-1)
    node
    {
      $\mathrm{M}5_{\mathrm{flav}}$
    };

  \draw[gray, line width=10pt]
    (0-.45,-1-.55-.55-.2-1) to (4+.45,-1-.55-.55-.2-1);
  \draw[gray, line width=10pt]
    (4-.45+2,-1-.55-.55-.2-1) to (4+.45+2,-1-.55-.55-.2-1);

  \draw (7,-5.3)
    node
    {
\scalebox{.9}{

\begin{tikzpicture}

  \draw (-1.3,.5) node
    {
      $\mathbb{R}/\mathbb{Z}_2^{\mathrm{HW}}$
    };

  \draw (4.5,.8) node
    {
      $\mathbb{H}/\mathbb{Z}_2^A$
    };

  \shadedraw[draw opacity=0, top color=darkyellow, bottom color=yellow]
    (-2-1,.2) rectangle (0,-.2);
  \shadedraw[draw opacity=.0, fill opacity=.4, top color=gray, bottom color=gray]
    (-.4,.2) rectangle (0,-.2);

  \shadedraw[draw opacity=0, top color=darkyellow, bottom color=yellow]
    (-2.3-1,.2) rectangle (-2.15-1,-.2);
  \shadedraw[draw opacity=0, top color=darkyellow, bottom color=yellow]
    (-2.45-1,.2) rectangle (-2.6-1,-.2);
  \shadedraw[draw opacity=0, top color=darkyellow, bottom color=yellow]
    (-2.75-1,.2) rectangle (-2.9-1,-.2);

 \shadedraw[draw opacity=0, greenii, top color=greenii, bottom color=green] (0,0) circle (.2);

 \draw (-.1,-.55)
   node
   {
     \scalebox{.9}{
     \color{darkblue} \bf
     \begin{tabular}{c}
       $\tfrac{1}{2}\mathrm{M5}$
     \end{tabular}
     }
   };

 \draw (-4,-.7)
   node
   {
     \begin{rotate}{0}
       \scalebox{.9}{
       \color{darkblue} \bf
       \begin{tabular}{c}
         Ho{\v r}ava-Witten bulk
       \end{tabular}
       }
     \end{rotate}
   };

 \draw (-.1+.3,-.55-.3)
   node
   {
     \begin{rotate}{-45}
       \scalebox{.9}{
       \color{darkblue} \bf
       \begin{tabular}{c}
         transversal cone
       \end{tabular}
       }
     \end{rotate}
   };

 \shadedraw[draw opacity=0, top color=darkblue, bottom color=cyan]
   (0,0) -- (3,3) .. controls (2,2) and (2,-2) ..  (3,-3) -- (0,0);

 \draw[draw opacity=0, top color=white, bottom color=darkblue]
   (3,3)
     .. controls (2,2) and (2,-2) ..  (3,-3)
     .. controls (4,-3.9) and (4,+3.9) ..  (3,3);

\end{tikzpicture}

}
    };

 \end{tikzpicture}
 }
 }
\end{equation}

\newpage

\noindent {\bf Expected $\mathrm{SU}(2_f)$ flavor theory on heterotic M5.}
In fact, this $\mathrm{SU}(2)$ gauge symmetry expected on single heterotic M5-branes
came to be understood as being a \emph{flavor symmetry}
(e.g. \cite[4.2]{GHKKLV18}\cite[2.3.1]{Ohmori18}) just like the ``hidden local
symmetry'' \cite{Sakurai60}\cite{BKY88} (reviewed in \cite[6]{MH16}) of chiral
perturbation theory for confined hadrodynamics (reviewed in \cite{Me88}\cite{Scherer03}\cite{BM07}\cite{MachleidtEntem11}),
instead of a color symmetry as in the quark model of quantum chromodynamics.
This is informally argued in the literature by
\begin{enumerate}[{\bf (i)}]
\vspace{-2mm}
\item  invoking
(see \cite[2.3]{CabreraHananySperling19}) M/IIA duality along the
$S^1 \subset \mathrm{Sp}(1)_L$-action in \eqref{TheBranes} for identifying the
$\tfrac{1}{2}\mathrm{M}5$-brane configuration with
a $\mathrm{NS}5 \!\parallel\! \mathrm{D}6 \!\perp\! \mathrm{O}8$-brane
configuration in Type I' string theory, and then

\vspace{-3mm}
\item appealing (\cite{HananyZaffaroni98}\cite{BrunnerKarch98}, reviewed in
\cite[2.1]{ApruzziFazzi17}) to open string theory for identifying the
$\mathrm{D}6$-branes emanating from the $\mathrm{NS}5$,
corresponding to the $\mathrm{MK}6$-brane in \eqref{TheBranes},
with flavor branes, due to their semi-infinite transverse extension.
\end{enumerate}

\vspace{-2mm}
\noindent This suggests that single but heterotic M5-branes are not just a toy example for
the more general open problem of non-abelian gauge enhancement in M-theory,
but, when viewed through the lens of holographic hadrodynamics,
possibly the core example for making contact with phenomenology.
Hence the first open problem to address is:

\medskip

\noindent {\bf The specific open problem.}
{\it  A derivation/formulation of the
(higher) $\mathrm{SU}(2)$-flavor gauge theory
emerging on single heterotic M5-branes.}

\medskip
In \cite{FSS20a} we had discussed this problem in the ``topological sector'', i.e.,
focusing on gauge- and gravitational instanton sectors and their topological
global anomaly cancellation conditions, while temporarily putting to the side local
field/differential form data. There we had proven that, under the hypothesis that
the M-theory C-field is charge-quantized in J-twisted Cohomotopy cohomology theory
(``Hypothesis H'' \cite{Sati13}\cite{FSS19b}\cite{FSS19c}, review in
\cite{Sc20}), a topological
$\mathrm{Sp}(1) \simeq \mathrm{SU}(2)$ gauge field sector indeed
emerges in the sigma-model for single M5-branes.

\vspace{1mm}
Here we set out to discover the local differential form data
to complete this topological picture of the
non-abelian heterotic M5-brane theory,
by identifying its  {\it local field content} and its {\it couplings}.

\medskip

\noindent {\bf Solution via super-exceptional geometry.}
In \cite{FSS19d} we had demonstrated that the failure of the
super-embedding approach
to produce the M5-brane Lagrangian is resolved, for heterotic M5-branes,
by enhancing to \emph{super-exceptional} embeddings. This means
enhancing the $\mathcal{N}=1$, $D=11$-dimensional target super-spacetime
$\mathbb{R}^{10,1\vert \mathbf{32}}$ to  the \emph{super-exceptional spacetime}
$\mathbb{R}^{10,1\vert \mathbf{32}}_{\mathrm{ex}_s}$
\cite{Bandos17}\cite{FSS18}\cite{SS18}\cite[3]{FSS19d}
(recalled as Def. \ref{SuperExceptionalSpacetimeAsLieAlgebra} below).
This has the virtue of unifying graviton and gravitino modes with M2- and M5-brane wrapping modes and
with a pre-gaugino field \cite[Rem. 5.4]{FSS19d}, or rather to its heterotic/type-I' version
$\mathbb{R}^{9,1\vert \mathbf{16}}_{\mathrm{ex}_s}$ \cite[4]{FSS19d}:

\vspace{-.6cm}

\begin{equation}
\label{SuperExceptionalSpacetimesAsRepresentation}
\small
\hspace{1.8cm}
\raisebox{-80pt}{
\begin{tikzpicture}

 \draw node (0,0)
 {
  \xymatrix@R=-5pt@C=1pt{
  &
    {\phantom{AAA}}
  &
  \overset{
    \mbox{
      \tiny
      \color{darkblue} \bf
      \begin{tabular}{c}
        graviton
        \\
        \phantom{a}
      \end{tabular}
    }
  }{
    \mathbb{R}^{10,1}
  }
  &
    \times
  &
  \overset{
    \mathclap{
    \mbox{
      \tiny
      \color{darkblue} \bf
      \begin{tabular}{c}
      gravitino
      \\
      \phantom{a}
      \end{tabular}
    }
    }
  }{
    \mathbb{R}^{0\vert \mathbf{32}}
  }
  &
    \times
  &
  \overset{
    \mathclap{
    \mbox{
      \tiny
      \color{darkblue} \bf
      \begin{tabular}{c}
        probe M2-brane
        \\
        wrapping modes
        \\
        \phantom{a}
      \end{tabular}
    }
    }
  }{
    \wedge^2(\mathbb{R}^{10,1})^\ast
  }
  &
    \times
  &
  \overset{
    \mathclap{
    \mbox{
      \tiny
      \color{darkblue} \bf
      \begin{tabular}{c}
        probe M5-brane
        \\
        wrapping modes
        \\
        \phantom{a}
      \end{tabular}
    }
    }
  }{
    \wedge^5(\mathbb{R}^{10,1})^\ast
  }
  &
    \times
  &
  \overset{
    \mbox{
      \tiny
      \color{darkblue} \bf
      \begin{tabular}{c}
      \end{tabular}
    }
  }{
    \mathbb{R}^{0\vert \mathbf{32}}
  }
  \\
  \mathllap{
    \mbox{
      \tiny
      \color{darkblue} \bf
      \begin{tabular}{c}
        Super-exceptional
        \\
        M-theory
        \\
        Minkowski spacetime
      \end{tabular}
    }
    \;
  }
  \mathbb{R}^{10,1\vert \mathbf{32}}_{\mathrm{ex}_s}
  \ar@{}[drr]|-{
    \begin{rotate}{-30}
      $
        \mathclap{
        \underset{
          \!
          \mbox{
            \tiny
            \color{darkblue}
            \bf
            \begin{tabular}{c}
              \phantom{a}
              \\
              as a $\mathrm{Pin}^+(10,1)$-
              \\
              representation
            \end{tabular}
          }
        }{
          \mathclap{\simeq}
        }
        }
      $
    \end{rotate}
  }
  \ar@{}[urr]|-{
    \begin{rotate}{30}
      $
        \mathclap{
        \overset{
          \mbox{
            \tiny
            \color{darkblue} \bf
            \begin{tabular}{c}
              as a super-
              \\
              manifold
              \\
              $\phantom{a}$
            \end{tabular}
          }
        }{
          \mathclap{\simeq}
        }
        }
      $
    \end{rotate}
  }
  &
  &
  \\
  &
    {\phantom{AAA}}
  &
  \underset{
    \mbox{
      \tiny
      \color{darkblue} \bf
    }
  }{
    \mathbf{11}
  }
  &
    \oplus
  &
  \underset{
    \mbox{
      \tiny
      \color{darkblue} \bf
    }
  }{
    \mathbf{32}
  }
  &
    \oplus
  &
  \underset{
    \mbox{
      \tiny
      \color{darkblue}
      \bf
      \begin{tabular}{c}
        \phantom{a}
        \\
        probe M9-brane
        \\
        wrapping modes
      \end{tabular}
    }
  }{
    \wedge^9 \mathbf{11}
  }
  &
    \oplus
  &
  \underset{
    \mbox{
      \tiny
      \color{darkblue} \bf
    }
  }{
    \wedge^5 \mathbf{11}
  }
  &
    \oplus
  &
  \underset{
    \mbox{
      \tiny
      \color{darkblue} \bf
    }
  }{
    \mathbf{32}
  }
  \\
  &
  &
  \\
  {\phantom{A}}
  \\
  \mathllap{
    \mbox{
      \tiny
      \color{darkblue} \bf
      \begin{tabular}{c}
        Super-exceptional
        \\
        heterotic/type-I'
        \\
        Minkowski spacetime
      \end{tabular}
    }
    \;
  }
  \mathbb{R}^{9,1\vert \mathbf{16} }_{\mathrm{ex}_s}
  \ar@{^{(}->}[uuuu]
  &
  \underset{
    \mathclap{
    \mbox{
      \tiny
      \color{darkblue}
      \bf
      \begin{tabular}{c}
        $\phantom{a}$
        \\
        as a $\mathrm{Pin}^+(9,1)$-
        \\
        representation
      \end{tabular}
    }
    }
  }{
    \simeq
  }
  &
  \;
  \mathbf{10}
  \;\oplus\;
  \underset{
    \mathclap{
    \mbox{
      \tiny
      \color{darkblue}
      \bf
      \begin{tabular}{c}
        \phantom{a}
        \\
        dilaton
      \end{tabular}
    }
    }
  }{
    \mathbf{1}
  }
  &
  \oplus
  &
  \mathbf{16}
  &
  \oplus
  &
  \;\;
  \underset{
    \mathclap{
    \mbox{
      \tiny
      \color{darkblue}
      \bf
      \begin{tabular}{c}
        \phantom{a}
        \\
        probe $\mathrm{F}1_{\Omega}$-brane
        \\
        wrapping modes
      \end{tabular}
    }
    \;\;\;
    }
  }{
    \mathbf{10}_{\scalebox{.5}{$\mathrm{sgn}$}}
  }
  \;\;\;
  \oplus
  \;\;\;
  \underset{
     \mathclap{
     \mbox{
       \tiny
       \color{darkblue}
       \bf
       \begin{tabular}{c}
         \phantom{a}
         \\
         probe $\mathrm{D}8$-brane
         \\
         wrapping modes
       \end{tabular}
     }
     }
  }{
    \wedge^8 \mathbf{10}_{\scalebox{.5}{$\phantom{\mathrm{sgn}}$}}
  }
  &
  \;\oplus\;
  &
  \underset{
    \mathclap{
    \mbox{
      \tiny
      \color{darkblue}
      \bf
      \begin{tabular}{c}
        \phantom{a}
        \\
        probe $\mathrm{NS}5$-brane
        \\
        wrapping modes
      \end{tabular}
    }
    }
  }{
    \wedge^5 \mathbf{10}_{\scalebox{.5}{$\phantom{\mathrm{sgn}}$}}
  }
  \;\;\;
  \oplus
  \;\;\;
  \underset{
    \mathclap{
    \;\;
    \mbox{
      \tiny
      \color{darkblue}
      \bf
      \begin{tabular}{c}
        \phantom{a}
        \\
        probe $\mathrm{D}4$-brane
        \\
        wrapping modes
      \end{tabular}
    }
    }
  }{
    \wedge^4 \mathbf{10}_{\scalebox{.5}{$\phantom{\mathrm{sgn}}$}}
  }
  &
  \;\;
  \oplus
  &
  \underset{
    \mbox{
      \tiny
      \color{darkblue}
      \bf
      \begin{tabular}{c}
        \phantom{s}
        \\
        gaugino
      \end{tabular}
    }
  }{
    \mathbf{16}_{\scalebox{.5}{$\phantom{\mathrm{sgn}}$}}
  }
  }
 };

 \draw[draw=green, fill=green, opacity=.05] (-1.3,2.1) rectangle (6.9,-2.1);

 \draw[draw=green, fill=green, opacity=.05] (-2.15,2.1) rectangle (-4.8,-2.1);

\end{tikzpicture}
}
\end{equation}

\vspace{-2mm}

The dual $\mathrm{M9}$-brane wrapping modes in the second line,
anticipated in \cite[p. 8-9]{Hull98}, follow by
rigorous analysis \cite[4.26]{FSS18}
(see Prop. \ref{SpinRepresentationsOfSuperExceptional} and Remark \ref{Rem-hull} below)
and lead to
$\mathrm{D}4$/$\mathrm{D}8$-brane modes in 10d, as shown in the last line.
The underlying super-symmetry algebra of
this super-exceptional spacetime $\mathbb{R}^{10,1\vert \mathbf{32}}_{\mathrm{ex}_s}$
is the ``hidden supergroup of $D=11$ supergravity''
\cite{DAuriaFre82}\cite{BAIPV04}\cite{ADR16},
whose role or purpose had previously remained elusive.
We may understand it \cite[Rem. 3.9]{FSS19d}\cite[4.6]{FSS18}\cite{SS18}
as that supermanifold whose real cohomology
accommodates that of the classifying 2-stack of the M5-brane
sigma-model \cite{FSS15} under Hypothesis H \cite{FSS19c}.

\medskip

\noindent {\bf Result -- Emergent chiral $\mathrm{SU}(2_f)$-theory on the heterotic M5. }
As a consequence of the super-exceptional enhancement
\eqref{SuperExceptionalSpacetimesAsRepresentation} of target spacetime,
additional worldvolume fields emerge also on the
single heterotic M5-brane locus \eqref{TheBranes},
originating in M2- and M5-brane wrapping modes.
Here we classify this emergent super-exceptional field content
by computing the representations of the residual group actions (see \cref{ElementsOfRepresentation Theory})
on the super-exceptional vielbein fields after passage to the
super-exceptional $\mathrm{MK}6$-locus (Def. \ref{TheSuperExceptionalMK6} below),
reduced to 4d:
\begin{equation}
  \label{RepresentationTheoreticQuestion}
  \hspace{-.8cm}
  \footnotesize
  \raisebox{50pt}{
  \xymatrix@R=-4pt@C=15pt{
    &
    \mathrm{Spin}(10,1)
    \ar@{<-^{)}}[rr]^-{ \iota }
    &&
     \;\;\;\;\;\;
     \mathrm{Spin}(6,1)
     \,
     \times
     \;\;\;\;\;\;\;\;\;
       \mathrm{Spin}(4)
     \;\;\;\;\;\;\;\;\;
    \ar@{}@<30pt>[dd]|-{
      \begin{rotate}{-90}
        $\mathclap{\simeq}$
      \end{rotate}
    }
    \ar@{<-^{)}}[rr]
    &
    &
    \;\;\;\;\;\;\;\;\;\;
      \mathrm{Spin}(5,1)
    \;\;\;\;\;\;\;\;
    \times
    \mathrm{SU}(2)_R
    \\
    {\phantom{a}}
    \\
    \mbox{
      \tiny
      \color{darkblue} \bf
      \begin{tabular}{c}
        Symmetry
        \\
        groups
      \end{tabular}
    }
    &&&
    \;\;\;\;\;\;\;\;\;\;
    \mathrm{Spin}(6,1)
      \times
    \mathrm{SU}(2)_L
      \times
    \mathrm{SU}(2)_R
    \ar@{<-^{)}}[rr]
    &&
    \;
    \mathrm{Spin}(3,1) \times \mathrm{Spin}(2)\times \mathrm{SU}(2)_R
    \ar@{}@<0pt>[dd]|-{
      \begin{rotate}{-90}
        $\mathclap{\simeq}$
      \end{rotate}
    }
    \ar@<+14pt>@{^{(}->}[uu]
    \\
    {\phantom{A}}
    \\
    & & & & &
    \mathrm{Spin}(3,1) \times \;\mathrm{U}(1)_V\;\;\; \times \mathrm{SU}(2)_R
    \\
    {\phantom{a}}
    \\
    \mbox{
      \tiny
      \color{darkblue} \bf
      \begin{tabular}{c}
        Representation
        \\
        rings
      \end{tabular}
    }
    &
        \!\!\!\!    \!\!\!\!     \!\!\!\!     \!\!\!\!
    \mathrm{Rep}_{\mathbb{R}}\big(\mathrm{Spin}(10,1)\big)
    \ar[rr]^-{
      \mbox{\color{darkblue} \bf \tiny pass to}
    }_-{
      \mbox{
      \tiny
      \color{darkblue} \bf
      \begin{tabular}{c}
        fixed representation
        \\
        of $\bf \mathbb{Z}_2^A \subset \mathrm{SU}(2)_L$
      \end{tabular}
    }}
    &&
    \mathrm{Rep}_{\mathbb{R}}
    \big(
      \mathrm{Spin}(6,1) \times \mathrm{SU}(2)_R
    \big)
    \ar[rr]^-{
      \mbox{\color{darkblue} \bf \tiny pass to}
    }_-{
      \mbox{
      \tiny
      \color{darkblue} \bf
      \begin{tabular}{c}
        restricted
        \\
        representation
      \end{tabular}
    }
    }
    &&
    \mathrm{Rep}_{\mathbb{R}}
    \big(
      \mathrm{Spin}(3,1)
        \times
      \mathrm{U}(1)_V
        \times
      \mathrm{SU}(2)_R
    \big)
    \\
     \mbox{
      \tiny
      \color{darkblue} \bf
      \begin{tabular}{c}
        Super-exceptional
        \\
        spacetimes
      \end{tabular}
    }
    &
    \!\!\!\!
   \underset{
      \mathclap{
      \mbox{
        \tiny
        \color{darkblue} \bf
        \begin{tabular}{c}
          super-exceptional
          \\
          M-theory spacetime
        \end{tabular}
      }
      }
    }{
      \mathbb{R}^{10,1\vert \mathbf{32}}_{\mathrm{ex}_s}
    }
    \ar@{|->}[rr]_{
      \mbox{
        \tiny
        \color{darkblue} \bf
        \begin{tabular}{c}
          localize onto brane
        \end{tabular}
      }
    }
    &&
    \underset{
      \mathclap{
      \mbox{
        \tiny
        \color{darkblue} \bf
        \begin{tabular}{c}
          super-exceptional
          \\
          $\mathrm{MK}6$-locus
        \end{tabular}
      }
      }
    }{
      \big(
        \mathbb{R}^{10,1\vert \mathbf{32}}_{\mathrm{ex}_s}
      \big)^{\mathbb{Z}_2^A}
    }
    \ar@{|->}[rr]_{
      \mbox{
        \tiny
        \color{darkblue} \bf
        \begin{tabular}{c}
          dimensionally reduce to 4d
        \end{tabular}
      }
    }
    &&
    \underset{
      \mbox{
        \tiny
        \color{darkblue} \bf
        \begin{tabular}{c}
          super-exceptional
          \\
          MK6-locus
          \\
          reduced to 4d
        \end{tabular}
      }
    }{
      \big(
        \mathbb{R}^{10,1\vert \mathbf{32}}_{\mathrm{ex}_s}
      \big)^{\mathbb{Z}_2^A}_{4d}
    }
    }
  }
\end{equation}

\vspace{-.3cm}

The result is shown in the following diagram,
discussed in detail in \cref{FieldContent} below:

\vspace{-.7cm}

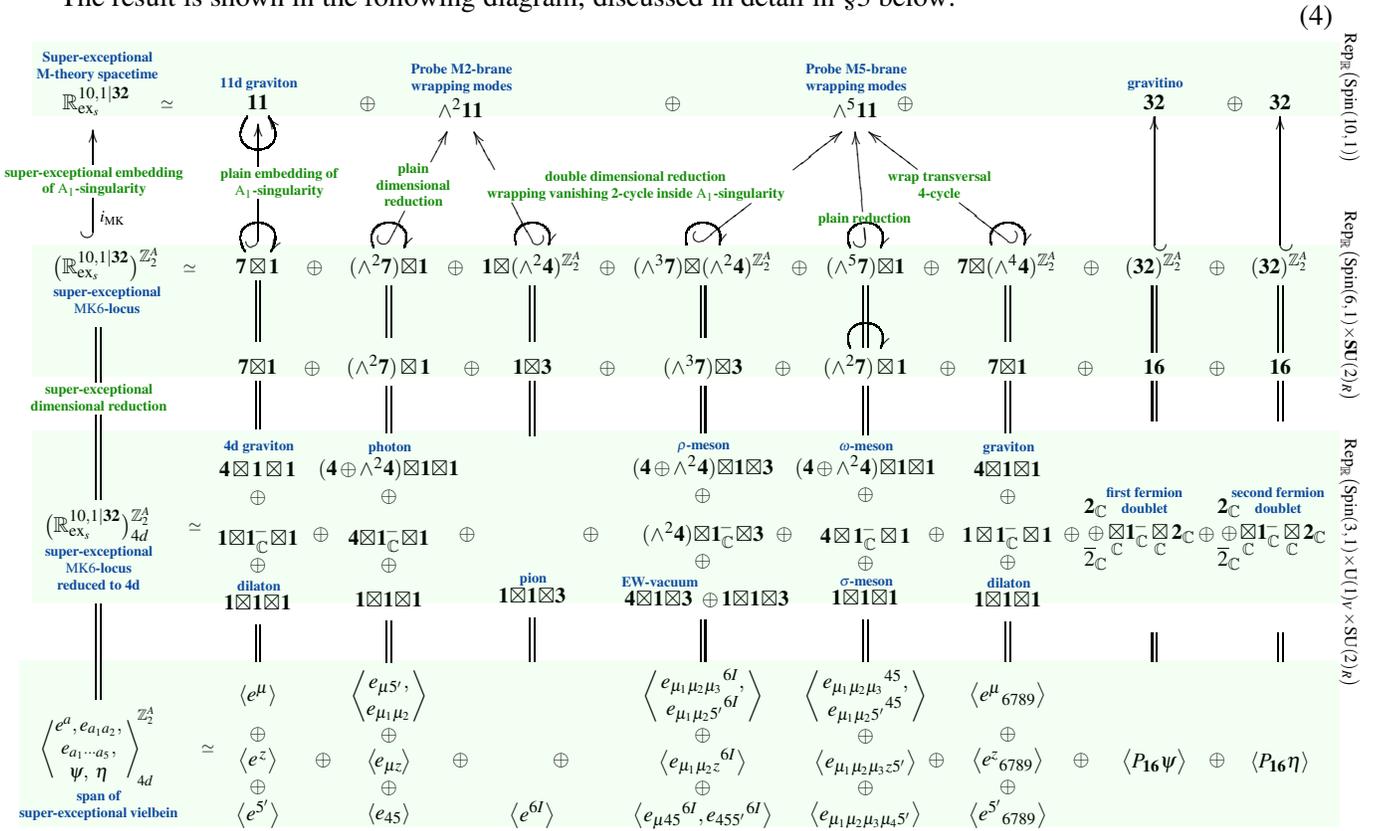
\begin{equation}
  \label{TheFieldContent}
  \;\;
\end{equation}

\vspace{-.5cm}

{
\scalebox{.9}{
  \hspace{-1.4cm}
  \footnotesize
\begin{tikzpicture}

\draw (0,0) node {
  \xymatrix@R=12pt@C=15pt{
    {\phantom{\Big(}}
    \overset{
      \mathclap{
      \mbox{
        \tiny
        \color{darkblue} \bf
        \begin{tabular}{c}
          Super-exceptional
          \\
          M-theory spacetime
        \end{tabular}
      }
      }
    }{
      \mathbb{R}^{10,1\vert\mathbf{32}}_{\mathrm{ex}_s}
    }
    {\phantom{
      \Big)^{
      \overset{
        \mathclap{
        \mbox{
          \tiny
          \color{darkblue} \bf
          \begin{tabular}{c}
          \end{tabular}
        }
        }
      }{
        \mathbb{Z}_2
      }
      }
    }
    }
    \!\!\!\!\!\!\!
    \ar@{}[r]|<<<{\simeq}
    &
    \overset{
      \mathclap{
      \mbox{
        \tiny
        \color{darkblue} \bf
        \begin{tabular}{c}
          11d
          graviton
        \end{tabular}
      }
      }
    }{
      \mathbf{11}
    }
    \ar@[white]@(dl,dr)_{\ }="s"
    \ar@{}[r]|>>>{\mbox{$\oplus$}}
    &
    \ar@{}[r]|-{\mbox{$
      \underset{
        \phantom{
        \tiny
        \begin{tabular}{c}
          a
          \\
          a
          \\
          a
          \\
          a
        \end{tabular}
        }
      }{
      \overset{
        \mathclap{
        \mbox{
          \tiny
          \color{darkblue} \bf
          \begin{tabular}{c}
            $\phantom{\vert^{\vert^{\vert^{\vert^{\vert^{\vert^\vert}}}}}}$
            \\
            Probe M2-brane
            \\
            wrapping modes
          \end{tabular}
        }
        }
      }{
        \wedge^2 \mathbf{11}
      }
      }
    $}}_{\ }="s1"
    &
    \ar@{}[rr]|>>>>>>>>>>>>>>>>>>>>>>>>>>>{\mbox{$\oplus$}}
    &
    \ar@{}[rr]|-{\mbox{$
      \underset{
        \phantom{
          \tiny
          \begin{tabular}{c}
            a
            \\
            a
            \\
            a
            \\
            a
          \end{tabular}
        }
      }{
      \overset{
        \mathclap{
        \mbox{
          \tiny
          \color{darkblue} \bf
          \begin{tabular}{c}
            $\phantom{\vert^{\vert^{\vert^{\vert^{\vert^{\vert^{\vert}}}}}}}$
            \\
            Probe M5-brane
            \\
            wrapping modes
          \end{tabular}
        }
        }
      }{
        \wedge^5 \mathbf{11}
      }
      }
    $}}_{\ }="s2"
    \ar@{}[rrr]|-{\mbox{$\oplus$}}
    &
    &
    &
    \overset{
      \mbox{
        \tiny
        \color{darkblue}
        \bf
        \begin{tabular}{c}
          gravitino
        \end{tabular}
      }
    }{
      \mathbf{32}
    }
    \ar@{<-^{)}}[dd]
    \ar@{}[r]|-{\mbox{$\oplus$}}
    &
    \mathbf{32}
    \ar@{<-^{)}}[dd]
    &
    \!\!\!\!\!\!\!\!\!\!\!\!
    \scalebox{.8}{
    \begin{rotate}{-90}
       $\mathclap{
         \!\!\!\!\!\!\!\!\!\!\!\!
         \mathrm{Rep}_{\mathbb{R}}
         \big(
           \mathrm{Spin}(10,1)
         \big)
       }$
    \end{rotate}
    }
    \\
    \\
      \underset{
        \mathclap{
        \mbox{
          \tiny
          \color{darkblue} \bf
          \begin{tabular}{c}
            super-exceptional
            \\
            $\mathrm{MK}6$-locus
          \end{tabular}
        }
        }
      }{
      {
        \big(
          \mathbb{R}^{10,1\vert\mathbf{32}}_{\mathrm{ex}_s}
        \big)^{
          \overset{
            \mathclap{
            \mbox{
              \tiny
              \color{greenii} \bf
              \begin{tabular}{c}
              \end{tabular}
            }
            }
          }{
            \mathbb{Z}_2^A
          }
        }
      }
      }
    \ar@<+2pt>@{^{(}->}[uu]|-{
      \mbox{
        \tiny
        \color{greenii} \bf
        \begin{tabular}{c}
          super-exceptional embedding
          \\
          of $\mathrm{A}_1$-singularity
        \end{tabular}
      }
    }_<<<{ i_{\mathrm{MK}} }
    \!\!\!\!\!\!\!
    \ar@{}[r]|-{\simeq}
    &
    \mathbf{7} \boxtimes \mathbf{1}
    \ar@[white]@(ul,ur)^{\ }="t"
    \ar@{=}[d]
    \ar@{}[r]|-{\mbox{$\oplus$}}
    &
    (\wedge^2 \mathbf{7})
    \!\boxtimes\!
    \mathbf{1}
    \ar@[white]@(ul,ur)^{\ }="t1"
    \ar@{=}[d]
    \ar@{}[r]|-{\mbox{$\oplus$}}
    &
    \mathbf{1}
      \!\boxtimes\!
    (\wedge^2 \mathbf{4})^{\mathbb{Z}_2^A}
    \ar@[white]@(ul,ur)^{\ }="t2"
    \ar@{=}[d]
    \ar@{}[r]|-{\mbox{$\oplus$}}
    &
    (\wedge^3 \mathbf{7})
    \!\boxtimes\!
    (\wedge^2 \mathbf{4})^{\mathbb{Z}_2^A}
    \ar@[white]@(ul,ur)^{\ }="t4"
    \ar@{=}[d]
    \ar@{}[r]|-{\mbox{$\oplus$}}
    &
    (\wedge^5 \mathbf{7})
    \!\boxtimes\!
    \mathbf{1}
    \ar@[white]@(ul,ur)^{\ }="t3"
    \ar@{=}[d]
    \ar@{}[r]|-{\mbox{$\oplus$}}
    &
    \mathbf{7}
    \!\boxtimes\!
    (\wedge^4 \mathbf{4})^{\mathbb{Z}_2^A}
    \ar@{=}[d]
    \ar@[white]@(ul,ur)^{\ }="t5"
    \ar@{}[r]|-{\mbox{$\oplus$}}
    &
    (\mathbf{32})^{\mathbb{Z}_2^A}
    \ar@{}[r]|-{\mbox{$\oplus$}}
    \ar@{=}[d]
    &
    (\mathbf{32})^{\mathbb{Z}_2^A}
    \ar@{=}[d]
    &
    \!\!\!\!\!\!\!\!\!\!\!\!
    \scalebox{.8}{
    \begin{rotate}{-90}
       $\mathclap{
         \;\;\;\;\;\;\;\;\;\;\;\;\;\;\;\;\;
         \mathrm{Rep}_{\mathbb{R}}
         \big(
           \mathrm{Spin}(6,1)
           \!\times\!
           \mathbf{SU}(2)_R
         \big)
       }$
    \end{rotate}
    }
    \\
    &
    \overset{
      \mbox{
        \tiny
        \color{darkblue}
        \bf
        \begin{tabular}{c}
        \end{tabular}
      }
    }{
      \mathbf{7} \!\boxtimes\! \mathbf{1}
    }
    \ar@{=}[d]
    \ar@{}[r]|-{\mbox{$\oplus$}}
    &
    \overset{
      \mbox{
        \tiny
        \color{darkblue} \bf
        \begin{tabular}{c}
        \end{tabular}
      }
    }{
      (\wedge^2 \mathbf{7})
      \boxtimes
      \mathbf{1}
    }
    \ar@{=}[d]
    \ar@{}[r]|-{\mbox{$\oplus$}}
    &
    \overset{
      \mbox{
        \tiny
        \color{darkblue} \bf
        \begin{tabular}{c}
        \end{tabular}
      }
    }{
      \mathbf{1} \!\boxtimes\! \mathbf{3}
    }
    \ar@{=}[d]
    \ar@{}[r]|-{\mbox{$\oplus$}}
    &
    \overset{
      \mbox{
        \tiny
        \color{darkblue} \bf
        \begin{tabular}{c}
        \end{tabular}
      }
    }{
      (\wedge^3 \mathbf{7})
        \!\boxtimes\!
      \mathbf{3}
    }
    \ar@{=}[d]
    \ar@{}[r]|-{\mbox{$\oplus$}}
    &
    (\wedge^2 \mathbf{7})
    \boxtimes
    \mathbf{1}
    \ar@[white]@(ul,ur)
    \ar@{=}[d]
    \ar@{}[r]|-{\mbox{$\oplus$}}
    &
    \overset{
      \mbox{
        \tiny
        \color{darkblue} \bf
        \begin{tabular}{c}
        \end{tabular}
      }
    }{
      \mathbf{7}
      \!\boxtimes\!
      \mathbf{1}
    }
    \ar@{=}[d]
    \ar@{}[r]|-{\mbox{$\oplus$}}
    &
    \mathbf{16}
    \ar@{=}[d]
    \ar@{}[r]|-{\mbox{$\oplus$}}
    &
    \mathbf{16}
    \ar@{=}[d]
    \\
    \underset{
      \mbox{
        \tiny
        \color{darkblue}
        \bf
        \begin{tabular}{c}
          super-exceptional
          \\
          $\mathrm{MK}6$-locus
          \\
          reduced to 4d
        \end{tabular}
      }
    }{
    \big(
      \mathbb{R}^{10,1\vert \mathbf{32}}_{\mathrm{ex}_s}
    \big)^{\mathbb{Z}_2^A}_{4d}
    }
    \ar@{=}[d]
    \ar@{}[r]|-{\simeq}
    \ar@{=}[uu]|>>>>>>>>>>>{
      \mbox{
        \tiny
        \color{greenii}
        \bf
        \begin{tabular}{c}
          super-exceptional
          \\
          dimensional reduction
        \end{tabular}
      }
    }
    &
    {\begin{array}{c}
      \overset{
        \mathclap{
        \mbox{
          \tiny
          \color{darkblue}
          \bf
          \begin{tabular}{c}
          4d
          graviton
          \end{tabular}
        }
        }
      }{
        \mathclap{
          \mathbf{4} \boxtimes \mathbf{1} \boxtimes \mathbf{1}
        }
        \mathrlap{
          \phantom{\wedge^2}
        }
      }
      \\
      \oplus
      \\
      \overset{
        \mbox{
          \tiny
          \color{darkblue}
          \phantom{A}
        }
      }{
        \mathclap{
        \mathbf{1}
          \!\boxtimes\!
        \mathbf{1}^{\!-}_{\mathbb{C}}
          \!\boxtimes\!
        \mathbf{1}
        }
      }
      \\
      \oplus
      \\
      \overset{
        \mbox{
          \tiny
          \color{darkblue}
          \bf
          dilaton
        }
      }{
        \mathclap{
        \mathbf{1}
          \!\boxtimes\!
        \mathbf{1}
          \!\boxtimes\!
        \mathbf{1}
        }
      }
    \end{array}}
    \ar@{}@<-3pt>[r]|-{\mbox{$\oplus$}}
    \ar@{=}[d]
    &
    {\begin{array}{c}
      \overset{
        \mbox{
          \tiny
          \color{darkblue}
          \bf
          photon
        }
      }{
        \mathclap{
          (\mathbf{4} \oplus \wedge^2 \mathbf{4})
            \!\boxtimes\!
          \mathbf{1}
            \!\boxtimes\!
          \mathbf{1}
        }
      }
      \\
      \oplus
      \\
      \overset{
        \mbox{
          \tiny
          \color{darkblue}
          \bf
          \phantom{A}
        }
      }{
        \mathclap{
          \mathbf{4}
            \!\boxtimes\!
          \mathbf{1}^{\!-}_{\mathbb{C}}
            \!\boxtimes\!
          \mathbf{1}
        }
      }
      \\
      \oplus
      \\
      \overset{
        \mbox{
          \tiny
          \color{darkblue}
          \bf
          \phantom{$\sigma$-meson}
        }
      }{
        \mathclap{
          \mathbf{1}
            \!\boxtimes\!
          \mathbf{1}
            \!\boxtimes\!
          \mathbf{1}
        }
      }
    \end{array}}
    \ar@{=}[d]
    \ar@<-3pt>@{}[r]|-{\mbox{$\oplus$}}
    &
    {\begin{array}{c}
      {\phantom{
        \overset{
          \mbox{
            \tiny
            \bf
            t
          }
        }{
          A
        }
      }}
      \\
      {\phantom{\oplus}}
      \\
      {\phantom{
        \overset{
          \mbox{
            \tiny
            \bf
            t
          }
        }{
          A
        }
      }}
      \\
      {\phantom{\oplus}}
      \\
      \overset{
        \mbox{
          \tiny
          \color{darkblue}
          \bf
          pion
        }
      }{
        \mathclap{
        \mathbf{1}
          \!\boxtimes\!
        \mathbf{1}
          \!\boxtimes\!
        \mathbf{3}
        }
      }
    \end{array}}
    \ar@{=}[d]
    \ar@<-3pt>@{}[r]|<<<<{\mbox{$\oplus$}}
    &
    \mathclap
    {\begin{array}{c}
      \overset{
        \mbox{
          \tiny
          \color{darkblue}
          \bf
          ${\rho}$-meson
        }
      }{
        \mathclap{
          (\mathbf{4} \oplus \wedge^2 \mathbf{4})
            \!\boxtimes\!
          \mathbf{1}
            \!\boxtimes\!
          \mathbf{3}
        }
      }
      \\
      \oplus
      \\
      \overset{
        \mbox{
          \tiny
          \color{darkblue}
          \bf
          \phantom{$\sigma$-meson}
        }
      }{
        \mathclap{
        (\wedge^2 \mathbf{4})
          \!\boxtimes\!
        \mathbf{1}^{\!-}_{\mathbb{C}}
          \!\boxtimes\!
        \mathbf{3}
        }
      }
      \\
      \oplus
      \\
      \overset{
        \mbox{
          \tiny
          \color{darkblue}
          \bf
          EW-vacuum
        }
      }{
        \mathclap{
        \mathbf{4}
          \!\boxtimes\!
        \mathbf{1}
          \!\boxtimes\!
        \mathbf{3}
        }
      }
      \!
      \oplus
      \overset{
        \mbox{
          \tiny
          \color{darkblue}
          \bf
          \phantom{A}
        }
      }{
      \mathbf{1}
        \!\boxtimes\!
      \mathbf{1}
        \!\boxtimes\!
      \mathbf{3}
      }
    \end{array}}
    \ar@{=}[d]
    \ar@<-3pt>@{}[r]|-{\mbox{$\oplus$}}
    &
    \mathclap
    {\begin{array}{c}
      \overset{
        \mbox{
          \tiny
          \color{darkblue}
          \bf
          ${\omega}$-meson
        }
      }{
        \mathclap{
          (\mathbf{4} \oplus \wedge^2 \mathbf{4})
            \!\boxtimes\!
          \mathbf{1}
            \!\boxtimes\!
          \mathbf{1}
        }
      }
      \\
      {\oplus}
      \\
      \overset{
        \phantom{
        \mbox{
          \tiny
          \bf
          A
        }
        }
      }{
      \mathbf{4}
        \boxtimes
      \mathbf{1}^{\!-}_{\mathbb{C}}
        \boxtimes
      \mathbf{1}
      }
      \\
      \oplus
      \\
      \overset{
        \mbox{
          \tiny
          \color{darkblue}
          \bf
          $\sigma$-meson
        }
      }{
        \mathclap{
        \mathbf{1}
          \!\boxtimes\!
        \mathbf{1}
          \!\boxtimes\!
        \mathbf{1}
        }
      }
    \end{array}}
    \ar@{=}[d]
    \ar@<-3pt>@{}[r]|-{\mbox{$\oplus$}}
    &
    \mathclap
    {\begin{array}{c}
      \overset{
        \mbox{
          \tiny
          \color{darkblue}
          \bf
          graviton
        }
      }{
        \mathclap{
          \mathbf{4}
            \!\boxtimes\!
          \mathbf{1}
            \!\boxtimes\!
          \mathbf{1}
        }
        \mathrlap{
          \phantom{\wedge^2}
        }
      }
      \\
      \oplus
      \\
      \mathclap{
        \overset{
          \phantom{
          \mbox{
            \tiny
            \bf
            a
          }
          }
        }{
        \mathbf{1}
          \boxtimes
        \mathbf{1}^{\!-}_{\mathbb{C}}
          \boxtimes
        \mathbf{1}
        }
      }
      \\
      \oplus
      \\
      \overset{
        \mbox{
          \tiny
          \color{darkblue}
          \bf
          dilaton
        }
      }{
        \mathclap{
        \mathbf{1}
          \!\boxtimes\!
        \mathbf{1}
          \!\boxtimes\!
        \mathbf{1}
        }
      }
    \end{array}}
    \ar@{=}[d]
    \ar@<-3pt>@{}[r]|<<<<<<<<<{\mbox{$\oplus$}}
    &
    \mathclap
    {\begin{array}{c}
    {\phantom{
      \overset{
        \mbox{
          \tiny
          \bf
        }
      }{
        A
      }
    }}
    \\
    {\phantom{\oplus}}
    \\
    \!\!\!\!\!\!\!\!
    \!\!\!\!\!\!\!
    \begin{array}{c}
      \mathbf{2}_{\mathbb{C}}
      \\
      \oplus
      \\
      \overline{2}_{\mathbb{C}}
    \end{array}
    \!\!\!\!\!\!\!
    \overset{
      \!\!\!\!\!\!\!\!\!
      \mbox{
      \tiny
        \color{darkblue}
        \bf
        \begin{tabular}{c}
          first fermion
          \\
          doublet
        \end{tabular}
      }
    }{
      \underset{
        \mathbb{C}
      }{\boxtimes}
      \mathbf{1}^{\!-}_{\mathbb{C}}
      \underset{
        \mathbb{C}
      }{\boxtimes}
      \mathbf{2}_{\mathbb{C}}
    }
    \\
    {\phantom{\oplus}}
    \\
    {\phantom{a}}
    \end{array}}
    \ar@{=}[d]
    \ar@<-3pt>@{}[r]|<<<<<<<{\mbox{$\;\oplus$}}
    &
    \;\,
    \mathclap
    {\begin{array}{c}
    {\phantom{
      \overset{
        \mbox{
          \tiny
          \bf
        }
      }{
        A
      }
    }}
    \\
    {\phantom{\oplus}}
    \\
    \!\!\!\!\!\!\!\!
    \begin{array}{c}
      \mathbf{2}_{\mathbb{C}}
      \\
      \oplus
      \\
      \overline{2}_{\mathbb{C}}
    \end{array}
    \!\!\!\!\!\!\!\!\!\!\!\!
    \overset{
      \!\!\!\!\!\!\!
      \mbox{
      \tiny
        \color{darkblue}
        \bf
        \begin{tabular}{c}
          second fermion
          \\
          doublet
        \end{tabular}
      }
    }{
      \underset{
        \mathbb{C}
      }{\boxtimes}
      \mathbf{1}^{\!-}_{\mathbb{C}}
      \underset{
        \mathbb{C}
      }{\boxtimes}
      \mathbf{2}_{\mathbb{C}}
    }
    \\
    {\phantom{\oplus}}
    \\
    {\phantom{a}}
    \end{array}}
    \ar@{=}[d]
    &
    \!\!\!\!\!\!\!\!\!\!\!\!
    \scalebox{.8}{
    \begin{rotate}{-90}
       $\mathclap{
         \;\;\;\;\;\;\;\;\;\;\;\;\;\;\;\;
         \mathrm{Rep}_{\mathbb{R}}
         \big(
           \mathrm{Spin}(3,1)
           \!\times\!
           \mathrm{U}(1)_V
           \!\times\!
           \mathrm{SU}(2)_R
         \big)
       }$
    \end{rotate}
    }
    \\
    \mathclap{
      \underset{
    \mbox{
      \tiny
      \color{darkblue}
      \bf
      \begin{tabular}{c}
        span of
        \\
        super-exceptional
        vielbein
      \end{tabular}
    }
      }{
      \scalebox{.9}{$
      \left\langle
        \!\!\!\!\!
        \begin{array}{c}
          e^a,
          e_{a_1 a_2},
          \\
          e_{a_1 \cdots a_5},
          \\
          \psi, \; \eta
        \end{array}
\        \!\!\!\!\!\!
      \right\rangle^{\!\!\mathbb{Z}_2^A}_{\!\!4d}
      $}
      }
    }
    \ar@{}[r]|>>>>>>>{\simeq}
    &
    \mathclap{\begin{array}{c}
      \langle e^{\mu} \rangle
      \mathclap{
        {\phantom{
          \left\langle
            \begin{array}{c}
              e_{\mu}{}^{I}
              \\
              e_{\mu}{}^{I}
            \end{array}
          \right\rangle
        }}
      }
      \\
      \oplus
      \\
      \big\langle
        e^{\,z}
      \big\rangle
      \\
      \oplus
      \\
      \big\langle
        e^{5'}
      \big\rangle
    \end{array}}
    \ar@<-5pt>@{}[r]|-{\mbox{$\oplus$}}
    &
    \mathclap
    {\begin{array}{c}
      {{
        \left\langle
          \!\!\!\!\!
          \begin{array}{c}
            e_{\mu 5'},
            \\
            e_{\mu_1 \mu_2}
          \end{array}
          \!\!\!\!\!
        \right\rangle
      }}
      \\
      \oplus
      \\
      \big\langle
        \!
        e_{\mu z}
        \!
      \big\rangle
      \\
      \oplus
      \\
      \langle
        e_{4 5}
      \rangle
    \end{array}}
    \ar@<-5pt>@{}[r]|-{\mbox{$\oplus$}}
    &
    \mathclap
    {\begin{array}{c}
      {\phantom{
        \left\langle
          \!\!\!\!
          \begin{array}{c}
            e_{\mu}{}^{5'},
            \\
            e_{\mu_1 \mu_2}
          \end{array}
          \!\!\!\!
        \right\rangle
      }}
      \\
      {\phantom{\oplus}}
      \\
      \phantom{\langle e^{\mu} \rangle}
      \\
      \phantom{\oplus}
      \\
      \big\langle
        e^{6 I}
      \big\rangle
    \end{array}}
    \ar@<-5pt>@{}[r]|<<<<{\mbox{$\oplus$}}
    &
    \mathclap
    {\begin{array}{c}
      {{
        \left\langle
          \!\!\!\!
          \begin{array}{c}
            e_{\mu_1 \mu_2 \mu_3}{}^{6 I},
            \\
            e_{\mu_1 \mu_2 5'}{}^{6 I}
          \end{array}
          \!\!\!\!
        \right\rangle
      }}
      \\
      {\oplus}
      \\
      {\big\langle
        e_{ \mu_1 \mu_2 z }{}^{6 I}
      \big\rangle}
      \\
      \oplus
      \\
      {\big\langle
        e_{\mu 4 5 }{}^{6 I},
        e_{4 5 5'}{}^{6 I}
      \big\rangle}
    \end{array}}
    &
    \mathclap
    {\begin{array}{c}
      {
        \left\langle
          \!\!\!\!
          \begin{array}{c}
            e_{\mu_1 \mu_2 \mu_3}{}^{4 5},
            \\
            e_{\mu_1 \mu_2 5'}{}^{4 5}
          \end{array}
          \!\!\!\!
        \right\rangle
      }
      \\
      \oplus
      \\
      \big\langle
        e_{ \mu_1 \mu_2 \mu_3 z 5' }
      \big\rangle
      \\
      \oplus
      \\
      \big\langle
        e_{ \mu_1 \mu_2 \mu_3 \mu_4 5' }
      \big\rangle
    \end{array}}
    \ar@<-5pt>@{}[r]|-{\mbox{$\oplus$}}
    &
    \mathclap
    {\begin{array}{c}
      \big\langle
        e^\mu{}_{6 7 8 9 }
      \big\rangle
      \mathclap{
        \phantom{
          \left\langle
            \begin{array}{c}
              e_{\mu}{}^I
              \\
              e_{\mu}{}^I
            \end{array}
          \right\rangle
        }
      }
      \\
      \oplus
      \\
      \big\langle
        e^z{}_{6 7 8 9}
      \big\rangle
      \\
      {\oplus}
      \\
      \big\langle
        e^{5'}{}_{6 7 8 9}
      \big\rangle
    \end{array}}
    \ar@<-5pt>@{}[r]|-{\mbox{$\oplus$}}
    &
    \mathclap
    {\begin{array}{c}
      {\phantom{
        \left\langle
          \!\!\!\!
          \begin{array}{c}
            e_{\mu_1 \mu_2 \mu_3 5'}{}^{4 5},
            \\
            e_{\mu_1 \mu_2 5'}{}^{4 5}
          \end{array}
          \!\!\!\!
        \right\rangle
      }}
      \\
      \phantom{\oplus}
      \\
      \mathclap{
        \big\langle
          P_{\mathbf{16}}\psi
        \big\rangle
      }
      \\
      \phantom{\oplus}
      \\
      \phantom{
        \big\langle
          e^{4 5}
        \big\rangle
      }
    \end{array}}
    \ar@<-5pt>@{}[r]|-{\mbox{$\oplus$}}
    &
    \mathclap
    {\begin{array}{c}
      {\phantom{
        \left\langle
          \!\!\!\!
          \begin{array}{c}
            e_{\mu_1 \mu_2 \mu_3 5'}{}^{4 5},
            \\
            e_{\mu_1 \mu_2 5'}{}^{4 5}
          \end{array}
          \!\!\!\!
        \right\rangle
      }}
      \\
      \phantom{\oplus}
      \\
      \mathclap{\langle P_{\mathbf{16}} \eta \rangle}
      \\
      \phantom{\oplus}
      \\
      \phantom{
      \langle
        e^{4 5}
      \rangle}
    \end{array}}
    \ar@{^{(}->}|-{
      \;\;\;\;\;\;\;\;\;\;\;
      \mbox{
        \tiny
        \color{greenii} \bf
        \begin{tabular}{c}
          plain embedding of
          \\
          $\mathrm{A}_1$-singularity
        \end{tabular}
      }
    }
      "t"+(0,4); "s"-(0,3)
    \ar@{^{(}->}|-{
     \mbox{
       \tiny
       \color{greenii} \bf
       \begin{tabular}{c}
         plain
         \\
         dimensional
         \\
         reduction
       \end{tabular}
     }
     \;\;\;
    }
      "t1"+(0,4); "s1"-(+2,3)
    \ar@{_{(}->}|-{
      \;\;\;\;\;\;\;\;
      \;\;\;\;\;\;\;\;
      \;\;\;\;\;\;\;\;
      \;\;\;\;\;\;\;\;
      \;\;\;\;\;\;\;\;
      \;\;\;\;\;\;\;\;
      \;\;\;\;\;\;\;\;
      \;\;\;\;\;\;\;\;
      \;\;\;\;\;\;\;
      \mbox{
        \tiny
        \color{greenii} \bf
        \begin{tabular}{c}
          double dimensional reduction
          \\
          wrapping vanishing 2-cycle
          inside $\mathrm{A}_1$-singularity
        \end{tabular}
      }
    }
    "t2"+(0,4); "s1"-(-2,3)
    \ar@{^{(}->}|<<<<{
      \mbox{
        \tiny
        \color{greenii} \bf
        \begin{tabular}{c}
          plain reduction
        \end{tabular}
      }
    }
      "t3"+(0,4); "s2"-(0,3)
    \ar@{^{(}->}|-{
      \phantom{
      \mbox{
        \tiny
        \color{greenii} \bf
        \begin{tabular}{c}
          wrap vanishing 2-cycle
          \\
          inside $\mathrm{A}_1$-singularity
        \end{tabular}
      }
      }
      \;\;\;\;
    }
      "t4"+(0,4); "s2"-(+2,3)
    \ar@{_{(}->}|-{
      \mbox{
        \tiny
        \color{greenii} \bf
        \begin{tabular}{c}
          wrap transversal
          \\
           4-cycle
        \end{tabular}
      }
    }
      "t5"+(0,4); "s2"-(-2,3)
  }
    };

  \begin{scope}[shift={(.3,.4)}]

  \draw[fill=green,draw opacity=0, fill opacity=.05]
    (-9.8,5.3) rectangle (9.5,4.2);

  \draw[fill=green,draw opacity=0, fill opacity=.05]
    (-9.8,5.3-3) rectangle (9.5,4.2-3.85);

  \draw[fill=green,draw opacity=0, fill opacity=.05]
    (-9.8,5.3-2.9-2.85) rectangle (9.5,4.2-3.9-3.3);

  \draw[fill=green,draw opacity=0, fill opacity=.05]
    (-10,5.3-2.9-2.7-3.54) rectangle (9.5,4.2-3.9-2.6-4);

  \end{scope}

\end{tikzpicture}
}
}

\vspace{-.3cm}

\noindent Here the outer $\boxtimes$-tensor products \eqref{OuterTensorProductOfRepresentations}
indicate the
transformation properties under the three residual groups.

\vspace{-.7cm}

\hspace{-.8cm}
\begin{tabular}{ll}
\hspace{-.5cm}
  \begin{tabular}{l}
  \raisebox{-36pt}{\footnotesize
  \xymatrix@C=-5.5pt@R=4pt{
    \begin{rotate}{+30}
      \color{olive}
      \scalebox{.7}{
      $
        \mathrlap{
          \mathrm{Spin}(3,1)
        }
      $
      }
    \end{rotate}
    \!\!\!
    &&
    \begin{rotate}{+30}
      \color{olive}
      \scalebox{.7}{
      $
        \mathrlap{
          \mathrm{U}(1)
        }
      $
      }
    \end{rotate}
    \!\!\!
    &&
    \begin{rotate}{+30}
      \color{olive}
      \scalebox{.7}{
      $
        \mathrlap{
          \mathrm{SU}(2)_R
        }
      $
      }
    \end{rotate}
    \!\!\!
    \\
    \mathbf{1} & \boxtimes &\mathbf{1} &\boxtimes &\mathbf{3}
    & &
    \mathrlap{\mbox{is an iso-vector scalar field}}
    \\
    \mathbf{4} &\boxtimes& \mathbf{1} &\boxtimes& \mathbf{1}
    &&
    \mathrlap{\mbox{is an iso-scalar vector field}}
    \\
    \mathbf{4} &\boxtimes& \mathbf{1} &\boxtimes& \mathbf{3}
    &&
    \mathrlap{\mbox{is an iso-vector vector field}}
    \\
    \mathbf{2}_{\mathbb{C}}
    &
    \underset{
      \mathbb{C}
    }{\boxtimes}
    &
    \mathbf{1}_{\mathbb{C}}
    &
    \underset{
      \mathbb{C}
    }{\boxtimes}
    &
   \mathbf{2}_{\mathbb{C}}
   & &
   \mathrlap{\mbox{is an iso-doublet spinor field}}
   \\
   \mathbf{2}_{\mathbb{C}}
   &
     \boxtimes
   &
   \mathbf{1}^{\!-}_{\mathbb{C}}
   &
     \boxtimes
   &
   \mathbf{2}_{\mathbb{C}}
   &&
   \mathrlap{\mbox{carries -1 units of hypercharge }}
  }
  \phantom{\mbox{is an iso-doublet spinor field}}
  }
  \end{tabular}
  &
  \hspace{0cm}
  \begin{minipage}[l]{11.5cm}

    \phantom{AA}

    \phantom{AA}

    This is the kind of field content encountered in
    quantum hadrodynamics
    (e.g. \cite{Ecker95}\cite{MachleidtEntem11}) where
    the iso-vector scalar field would
    be the pion field $\vec \pi$, the iso-scalar vector field
    would be the omega-meson $\omega$, its iso-vector partner
    the rho-meson $\rho$, and the two hypercharged fermion iso-doublets
    would be baryon fields (e.g. \cite{SerotWalecka92}).

    Indeed, we show in \cref{VectorMesonCouplingToSkyrmeBaryonCurrent}
    that the form of the WZW-term in the
    super-exceptional PS-Lagrangian is that characteristic of the
    coupling of neutral vector mesons to triples of pions via
    the Skyrme baryon current \cite{RhoEtAl16}, see \cref{ConcludingRemarks}
    for further discussion.
  \end{minipage}

\end{tabular}

\newpage

\section{The super-exceptional M5-brane model}
\label{SuperExceptionalMGeometry}

Here we recall the relevant aspects of the super-exceptional
embedding construction of the $\mathrm{M}5$-brane \cite{FSS19b}
to introduce the precise setup that is analyzed in the following sections.

\medskip

\begin{remark}[Super-Lie algebras]
Throughout, we use the following basic fact, see \cite[3]{FSS19a} for review:
\begin{enumerate}[{\bf (i)}]
\vspace{-2mm}
\item Finite-dimensional Lie superalgebras\footnote{Our ground field is the real numbers.}
are equivalently encoded in terms of their
differential graded-commutative (dgc) Chevalley-Eilenberg super-algebras,
known as ``FDA''s in the supergravity literature:

\vspace{-4mm}
\begin{equation}
  \label{CEAlgebra}
  \hspace{-1cm}
  \xymatrix@R=1pt{
    \mathrm{LieSuperAlgebras}
    \;\; \ar@{^{(}->}[rr]^-{ \mathrm{CE} }_-{
      \mathclap{
      \mbox{
        \tiny   \bf
        \color{darkblue}
        \begin{tabular}{c}
          fully faithful
          \\
          embedding
        \end{tabular}
      }
      }
    }
    &&
    \;\; \mathrm{dgcSuperAlgebras}^{\mathrm{op}}
    \\
    \big(
      \overset{
        \mathclap{
        \mbox{
          \tiny \bf
          \color{darkblue}
          \begin{tabular}{c}
            super-
            \\
            vector space
          \end{tabular}
        }
        }
      }{
        \overbrace{
          \mathfrak{g}
        }
      },
      \underset{
        \mathclap{
        \mbox{
          \tiny \bf
          \color{darkblue}
          \begin{tabular}{c}
            super
            \\
            Lie bracket
          \end{tabular}
        }
        }
      }{
        \underbrace{
          [-,-]
        }
      }
    \big)
  \;\;  \ar@{|->}[rr]
    &&
    \;\;
    \big(\,
      \overset{
        \mathclap{
        \mbox{
          \tiny \bf
          \color{darkblue}
          \begin{tabular}{c}
          super-
          \\
          Grassmann algebra
          \end{tabular}
        }
        }
      }{
        \overbrace{
          \wedge^\bullet \mathfrak{g}^\ast
        }
      },
      \,
      \underset{
        \mathclap{
        \mbox{
          \tiny \bf
          \color{darkblue}
          \begin{tabular}{c}
            CE--differential
          \end{tabular}
        }
        }
      }{
        \underbrace{
          d := [-,-]^\ast
        }
      }
    \, \big).
  }
\end{equation}

\vspace{-5mm}
\item Here \emph{super-Grassman algebra} means that elements
$v \in \mathfrak{g}$ of super-grading $\sigma_v \in \mathbb{Z}_2$
are dual to dgc-algebra elements $v^\ast \in \wedge^1 \mathfrak{g}^\ast$
of bidegree
$(1,\sigma) \in \mathbb{N} \times \mathbb{Z}_2$.
We write ``$\wedge$'' for the product in these super-Grassman algebras.
The sign rule is such that for elements
$\alpha, \beta \in \wedge^\bullet \mathfrak{g}^\ast$
of homogeneous bi-degrees
$( n_\alpha, \sigma_\alpha ), (n_\beta, \sigma_\beta)
\in \mathbb{N} \times \mathbb{Z}_2$ we have
$$
  \alpha \wedge \beta \;=\;
  (-1)^{ n_\alpha n_\beta + \sigma_\alpha \sigma_\beta }
  \,
  \beta \wedge \alpha
  \,.
$$

\vspace{-3mm}
\item
Generally, the Chevalley-Eilenberg algebra of a Lie superalgebra $\mathfrak{g}$
may be identified with the de Rham algebra of left-invariant
differential forms on the corresponding Lie supergroup $G$:

\vspace{-4mm}
\begin{equation}
  \label{CEToDeRham}
  \xymatrix{
    \mathllap{
      \mbox{
        \tiny \bf
        \color{darkblue}
        \begin{tabular}{c}
          Chevalley-Eilenberg algebra
          \\
          of Lie superalgebra $\mathfrak{g}$
        \end{tabular}
      }
      \;\;\;\;
    }
    \mathrm{CE}(\mathfrak{g})
    \ar[rrr]^-{ \simeq }
    &&&
    \Omega^\bullet_{\mathrm{LI}}(G)
    \mathrlap{
      \;\;\;\;
      \mbox{
        \tiny \bf
        \color{darkblue}
        \begin{tabular}{c}
          de Rham algebra of
          \\
          left-invariant differential forms
          \\
          on Lie supergroup $G$
        \end{tabular}
      }
    }
  }
\end{equation}
\end{enumerate}
\end{remark}

\begin{example}[Ordinary super-Minkowski spacetime]
 \label{OrdinarySuperMinkowskiSpacetime}
For $d \in \mathbb{N}$ and
$\mathbf{N} \in \mathrm{Rep}_{\mathbb{R}}(\mathrm{Spin}(d,1))$
a real $\mathrm{Spin}(d,1)$ representation of real dimension $N$,
the corresponding
{\it super-Poincar{\'e} super Lie algebra}
$\mathrm{Iso}(\mathbb{R}^{d,1\vert \mathbf{N}})$ (``supersymmetry algebra'')
is the algebra whose CE-algebra \eqref{CEAlgebra} is


\begin{equation}
  \label{SuperMinkowski}
  \;\;
\end{equation}

\vspace{-.6cm}

\hspace{-.9cm}
\begin{tabular}{cc|c}
$\mathrm{CE}
\big(
  \;
  \mathrm{Iso}(\mathbb{R}^{d,1\vert \mathbf{N}})
  \;
\big)$
&
...generated by
&
...with differential on generators given by
\\
&
$
  \Big\{
    \begin{array}{c}
    \overset{
      \mathrm{degree} = (1, \mathrm{even})
    }{
    \overbrace{
      e^{a}, \omega^{a_1 a_2}
    }
    },
    \\
    \underset{
      \mathrm{degree} = (1, \mathrm{odd})
    }{
    \underbrace{
      \psi^\alpha
    }
    }
    \end{array}
  \Big\}_{
    { ai \in \{0,1, \cdots, d\} }
    \atop
    { \alpha \in \{1,2, \cdots, N\} }
  }
$
&
$
  \begin{array}{rcl}
    d \!\!\!\!\!\! & e^a
      & =
      \;
    \big\langle \psi \wedge  {\pmb\Gamma}^a \! \cdot \psi \big\rangle
    \\
    d \!\!\!\!\!\! & \psi^\alpha
      & =
      \;
    \tfrac{1}{4} \omega^{a_1 a_2} {\pmb\Gamma}_{a_1 a_2} \!\cdot\! \psi
    \\
    d \!\!\!\!\!\! & \omega^{a_1}{}_{a_2}
    & =
    \;
    \omega^{a_1}{}_{a_3} \wedge \omega^{a_3}{}_{a_2}
    \,,
  \end{array}
$
\end{tabular}

\medskip
\noindent
with the $\omega^{a_1 a_2}$ skew-symmetric in their indices.
On the right we have the Clifford action and spinor pairing
$\langle -,-\rangle$
that comes with the real Spin representation
(the dot denotes matrix multiplication, hence
contraction of spinor indices).

\medskip
The underlying
{\bf translational super Lie algebra}
$\mathbb{R}^{d,1\vert \mathbf{N}}$ is obtained from this by discarding the
Lorentz generators $\omega^{a_1}{}_{a_2}$. The resulting CE-algebra
may be identified with the de Rham algebra of the canonical super-vielbein
on the $D = d+1$, $\mathcal{N} = \mathrm{dim}(\mathbf{N})$
super-Minkowski spacetime, which is the Cartesian super-manifold with
canonical super-coordinates
$\{\!\!\!\! \underset{\mathrm{deg} = (0,\mathrm{even})}{\underbrace{ x^a}}, \underset{\mathrm{deg} = (0,\mathrm{odd})}{\underbrace{\theta^\alpha}} \!\!\!\!\}$:

\vspace{-6mm}
\begin{equation}
  \label{LeftInvariantFormsOnSuperSpacetime}
\hspace{3cm}
\raisebox{-28pt}{
  $
  \mathllap{
  \mbox{
    \tiny
    \color{darkblue}
    \bf
    \begin{tabular}{c}
      super-vielbein basis on
      \\
      super-Minkowski spacetime
    \end{tabular}
  }
  }
  \phantom{AA}
  \raisebox{+27pt}{
  \xymatrix@R=-2pt{
    \mathrm{CE}
    \big(
      \;
      \mathbb{R}^{d,1\vert \mathbf{N}}
      \;
    \big)
    \ar[rr]^-{\simeq}
    &&
    \Omega^\bullet_{\mathrm{LI}}
    \big(
      \;
      \mathbb{R}^{d,1\vert \mathbf{N}}
      \;
    \big)
    \\
    e^a
    \ar@{|->}[rr]
    &&
    d x^a +
    \langle
      \theta
      ,
      {\pmb\Gamma}^a \!\!\cdot d \theta
    \rangle
    \\
    \psi^\alpha
    \ar@{|->}[rr]
    &&
    d \theta^\alpha
    \phantom{+ \overline{\theta} \!\cdot\!  {\pmb\Gamma}^a \!\!\cdot d \theta}
  }
  }
  \phantom{AA}
  \mathrlap{
  \mbox{
    \tiny
    \color{darkblue}
    \bf
    \begin{tabular}{c}
      left-invariant 1-form basis on
      \\
      super-Minkowski spacetime
    \end{tabular}
  }
  }
  $
  }
\end{equation}
\end{example}
We consider this here specifically for $d = 3, 6, 10$, with the
real $\mathrm{Spin}(d,1)$ representation
$\mathbf{N} = \mathbf{4}, \mathbf{16}, \mathbf{32}$
given by Dirac matrices
with coefficients in  $\mathbb{K} = \mathbb{C}, \mathbb{H}, \mathbb{O}$,
respectively; this is reviewed in \cref{Octonionic2ComponentSpinors} below.

\newpage

The following is the ``hidden supergroup of 11d supergravity''
due to \cite[6]{DAuriaFre82}\cite{BAIPV04},
interpreted as the translational supersymmetry algebra of
super-exceptional M-theory spacetime according to
\cite{Bandos17}\cite{FSS18}\cite{SS18} \cite[3]{FSS19d}:

\begin{defn}[Super-exceptional M-spacetime]
  \label{SuperExceptionalSpacetimeAsLieAlgebra}
  Regarded as a super-Lie algebra of super-translations along
  itself,
  in generalization of Example \ref{OrdinarySuperMinkowskiSpacetime},
  the {\it super-exceptional Minkowski spacetime}
  $\mathbb{R}^{10,1\vert \mathbf{32}}_{\mathrm{ex}_s}$
  \eqref{SuperExceptionalSpacetimesAsRepresentation}
  (for a parameter $s \in \mathbb{R} \setminus \{0\}$)
  has Chevalley-Eilenberg algebra \eqref{CEAlgebra}:
  \begin{equation}
  \label{SuperExceptionalBianchiIdentities}
  \mbox{
  \hspace{-.5cm}
  \begin{tabular}{lc|l}
    $
      \mathrm{CE}
      \big(
        \;
        \mathbb{R}^{10,1\vert \mathbf{32}}
        \;
      \big)
    $
    \!\!\!\!\!\!\!\!
    &
    ... generated from
    &
    ... with differential on generators given by
    \\
    &
    $
    \Big\{\!\!\!\!
      \begin{array}{c}
      \overset{
        \mathrm{degree} = (1,\mathrm{even})
      }{
      \overbrace{
      e^{a_1},
      \;
      e_{a_1, a_2},
      \;
      e_{a_1, \cdots a_5}
      }
      }
      ,
      \\
      \underset{
        \mathrm{deg} = (1,\mathrm{odd})
      }{
      \underbrace{
      \psi^\alpha
      \,,
      \eta^\alpha
      }
      }
      \end{array}
   \!\!\!\!
   \Big\}
  $
  &
  $
  \begin{array}{rcl}
    d \!\!\!\!\!\!\!\! & e^a
      & =
    \phantom{\tfrac{1}{2}}
    \big\langle
      \psi \wedge  {\pmb\Gamma}^a \!\cdot\! \psi
    \big\rangle
    \\
    d \!\!\!\!\!\!\!\! & e_{a_1 a_2} & =
      \tfrac{1}{2}
      \big\langle
        \psi \wedge {\pmb\Gamma}_{a_1 a_2} \!\cdot\! \psi
      \big\rangle
    \\
    d \!\!\!\!\!\!\!\! & e_{a_1 \cdots a_5} & =
      \tfrac{1}{5!}
      \big\langle
        \psi \wedge {\pmb\Gamma}_{a_1 \cdots a_5} \!\cdot\! \psi
      \big\rangle
    \\
    d \!\!\!\!\!\!\!\! & \psi & = 0
        \\
    d \!\!\!\!\!\!\!\! & \eta & =
      \big(
        (s+1) e^a {\pmb\Gamma}_a
        + e_{a_1 b_1} {\pmb\Gamma}^{a_1 a_2}
        +
        (1 + \tfrac{s}{6})
        e_{a_1 \cdots a_5} {\pmb\Gamma}^{a_1 \cdots a_5}
      \big)
      \!\cdot\!
      \psi
  \end{array}
  $
\end{tabular}
}
\end{equation}
\noindent for $
      {
        a_i \in \{0,1,2,3,4,5,5',6,7,8,9\}
      }
      ,
      {
        \alpha \in \{1,2,\cdots, 32\}
      }
  $,
  with the generators $e_{a_1 a_2}$ and
  $e_{a_1 a_2 a_3 a_4 a_5}$ all skew-symmetric in their indices.
  Here on the right we have the $\mathbb{O}$-Dirac matrix multiplication
  and spinor pairing of the 4-component octonionic realization of the
  $\mathbf{32}$ of $\mathrm{Spin}(10,1)$, reviewed in \cref{Octonionic2ComponentSpinors}
  below, and ${\pmb\Gamma}_{a_1 \cdots a_p}$ denotes the skew-symmetrization
  of the Clifford algebra products ${\pmb\Gamma}_{a_1} \cdots {\pmb\Gamma}_{a_p}$,
  as usual.
\end{defn}

In generalization of
\eqref{LeftInvariantFormsOnSuperSpacetime}, we may identify the
super-exceptional vielbein \eqref{SuperExceptionalBianchiIdentities}
as left-invariant differential forms on the underlying
Cartesian super-manifold
of $\mathbb{R}^{10,1\vert \mathbf{32}}_{\mathrm{ex}_s}$,
now given by exceptional coordinate functions
(as envisioned, in the bosonic sector, in \cite[4.3]{Hull07})
$$
  \big\{
    \;
    \underset{
      \mathrm{degree} = (0,\mathrm{even})
    }{
    \underbrace{
      \overset{
        \mathclap{
        \mbox{
          \tiny
          \color{darkblue}
          \bf
          \begin{tabular}{c}
            ordinary
            \\
            bosonic
          \end{tabular}
        }
        }
      }{
      \overbrace{
        x^a
      }},\;
       \overset{
         \mathclap{
         \mbox{
           \tiny
           \color{darkblue}
           \bf
           \begin{tabular}{c}
             exceptional
             \\
             bosonic
           \end{tabular}
         }
         }
       }{
       \overbrace{
       B_{a_1 a_2},
       \;
       B_{a_1 \cdots a_5}
       }
       },
    }
    }
    \,\;
    \underset{
      \mathrm{degree}
      =
      (0, \mathrm{odd})
    }{
    \underbrace{
      \overset{
        \mathclap{
        \mbox{
          \tiny
          \color{darkblue}
          \bf
          \begin{tabular}{c}
            ordinary
            \\
            fermionic
          \end{tabular}
        }
        \;\;\;\;\;
        }
      }{
      \overbrace{
        \theta
      }
      }, \;
      \overset{
        \mathclap{
        \;\;\;\;\;\;\;
        \mbox{
          \tiny
          \color{darkblue}
          \bf
          \begin{tabular}{c}
            exceptional
            \\
            fermionic
          \end{tabular}
        }
        }
      }{
      \overbrace{
        \theta'
      }
      }
    }
    }
    \;
  \big\}
  \mathrlap{
    \phantom{AA}
    \mbox{
      \tiny
      \color{darkblue}
      \bf
      \begin{tabular}{c}
        super-exceptional
        \\
        coordinate functions
      \end{tabular}
    }
  }
$$
as follows:

\vspace{-4mm}
\begin{equation}
  \label{VielbeinInCanonicalCoordinateBasis}
  \hspace{-7cm}
  \mathllap{
    \raisebox{-40pt}{
      \tiny
      \color{darkblue}
      \bf
      \begin{tabular}{c}
        super-exceptional
        \\
        vielbein basis
      \end{tabular}
    }
  }
  \xymatrix@R=-4pt{
    \mathrm{CE}
    \big(
      \mathbb{R}^{10,1\vert \mathbf{32}}_{\mathrm{ex}_s}
    \big)
    \ar[rr]^-{\simeq}
    &&
    \Omega^\bullet_{\mathrm{li}}
    \big(
      \mathbb{R}^{10,1\vert \mathbf{32}}_{\mathrm{ex}_s}
    \big)
    \\
    e^a
    &\longmapsto&
    \!\!\!\!\!\!\!\!\!
    \!\!\!\!\!\!\!\!\!
    \!\!\!\!\!\!\!\!\!
    \!\!\!\!\!\!\!\!\!
    \mathrlap{
      d x^a
      \;\;\;\;\;\;\;\;\;\;\;\;
      +
      \phantom{\tfrac{1}{2}}
      \langle
        \theta
        ,
        \Gamma^a  \!\cdot\! d\theta
      \rangle
    }
    \\
    e_{a_1 a_2}
    &\longmapsto&
    \!\!\!\!\!\!\!\!\!
    \!\!\!\!\!\!\!\!\!
    \!\!\!\!\!\!\!\!\!
    \!\!\!\!\!\!\!\!\!
    \mathrlap{
      d(B_{a_1 a_2})
      \;\;\;\,
      +
      \tfrac{1}{2}
      \langle
        \theta
        ,
        \Gamma_{a_1 a_2} \!\cdot\! d \theta
      \rangle
    }
    \\
    e_{a_1 \cdots a_5}
    &\longmapsto&
    \!\!\!\!\!\!\!\!\!
    \!\!\!\!\!\!\!\!\!
    \!\!\!\!\!\!\!\!\!
    \!\!\!\!\!\!\!\!\!
    \mathrlap{
      d(B_{a_1 \cdots a_5})
        \;
        +
      \tfrac{1}{5!}
      \langle
        \theta,
        \Gamma_{a_1 \cdots a_5} \!\cdot\! d\theta
      \rangle
    }
    \\
    \psi^\alpha
    &\longmapsto&
    \!\!\!\!\!\!\!\!\!
    \!\!\!\!\!\!\!\!\!
    \!\!\!\!\!\!\!\!\!
    \!\!\!\!\!\!\!\!\!
    \mathrlap{
      d \theta
    }
    \\
    \eta^\alpha
    &\longmapsto&
    \!\!\!\!\!\!\!\!\!
    \!\!\!\!\!\!\!\!\!
    \!\!\!\!\!\!\!\!\!
    \!\!\!\!\!\!\!\!\!
    \mathrlap{
      d \theta'
      \;\;\;\;\;\;\;\;\;\;\;\;
      -
      \big(
        (s+1) e^a {\pmb\Gamma}_a
        + e_{a_1 b_1} {\pmb\Gamma}^{a_1 a_2}
        +
        (1 + \tfrac{s}{6})
        e_{a_1 \cdots a_5} {\pmb\Gamma}^{a_1 \cdots a_5}
      \big)
      \!\cdot\!
      \theta \;.
    }
  }
  \phantom{AA}
  \mathrlap{
    {\phantom{AAAAAAAAAAAAA}}
    \raisebox{-40pt}{
      \tiny
      \color{darkblue}
      \bf
      \begin{tabular}{c}
        left-invariant 1-forms on
        \\
        super-exceptional Minkowski spacetime
      \end{tabular}
    }
  }
\end{equation}
Here the first four lines follow just
as in \eqref{LeftInvariantFormsOnSuperSpacetime}.
The point to notice is the last line,
which follows with the same Fierz identity
\cite[(6.3)-(6.4)]{DAuriaFre82}\cite[(20)-(23)]{BAIPV04}
that gives $d d \theta' = 0$
in \eqref{SuperExceptionalBianchiIdentities}.
    .

\begin{prop}[Spin representations on super-exceptional spacetime]
 \label{SpinRepresentationsOfSuperExceptional}
Regard the linear span of the generators in  the list \eqref{SuperExceptionalBianchiIdentities},
as acted on by $\mathrm{Spin}(10,1)$ in the evident way, as \footnote{
We recall representation-theoretic notation below in \cref{ElementsOfRepresentation Theory}.}
\begin{equation}
  \label{SpinActionOnSuperExceptional}
  \xymatrix@C=0pt@R=8pt{
    \mathbb{R}^{10,1\vert \mathbf{32}}
    &
      \simeq_{{}_{\mathbb{R}}}
    &
    \mathbf{11}
    \ar@{}[d]|-{
      \!
      \begin{rotate}{-90}
        \scalebox{.9}{
        $\mathclap{\!\simeq}$
        }
      \end{rotate}
    }
    &\oplus&
    \wedge^2 \mathbf{11}
    \ar@{}[d]|-{
      \!
      \begin{rotate}{-90}
        \scalebox{.9}{
        $\mathclap{\!\simeq}$
        }
      \end{rotate}
    }
    &\oplus&
    \wedge^5 \mathbf{11}
    \ar@{}[d]|-{
      \!
      \begin{rotate}{-90}
        \scalebox{.9}{
        $\mathclap{\!\simeq}$
        }
      \end{rotate}
    }
    &\oplus&
    \mathbf{32}
    \ar@{}[d]|-{
      \!
      \begin{rotate}{-90}
        \scalebox{.9}{
        $\mathclap{\!\simeq}$
        }
      \end{rotate}
    }
    &\oplus&
    \mathbf{32}
    \ar@{}[d]|-{
      \!
      \begin{rotate}{-90}
        \scalebox{.9}{
        $\mathclap{\!\simeq}$
        }
      \end{rotate}
    }
    &
    \;\;\;
    \in
    \mathrm{Rep}_{\mathbb{R}}
    \big(
      \mathrm{Spin}(10,1)
    \big).
    \\
    &&
    \big\langle
      e^a
    \big\rangle
    &&
    \big\langle
      e_{a_1 a_2}
    \big\rangle
    &&
    \big\langle
      e_{a_1 \cdots a_5}
    \big\rangle
    &&
    \big\langle
      \psi^\alpha
    \big\rangle
    &&
    \big\langle
      \eta^\alpha
    \big\rangle
  }
\end{equation}
\item {\bf (i)} This gives a $\mathrm{Spin}(10,1)$-action on
$\mathrm{CE}\big( \mathbb{R}^{10,1\vert \mathbf{32}}_{\mathrm{ex}_s} \big)$
by dgc-superalgebra automorphisms, hence a
$\mathrm{Spin}(10,1)$-action
on $\mathbb{R}^{10,1\vert \mathbf{32}}_{\mathrm{ex}_s}$ itself
by Lie superalgebra automorphisms.

\item {\bf (ii)} Moreover, this extends to a $\mathrm{Pin}^+(10,1)$-action by automorphisms,
if one lifts the
$\wedge^{2} \mathbf{11} \simeq_{{}_{\mathrm{Spin}(10,1)}} \wedge^9 \mathbf{11}$
of $\mathrm{Spin}(10,1)$ specifically to the
$\wedge^9 \mathbf{11}$ of $\mathrm{Pin}^+(10,1)$, hence as:
\begin{equation}
  \label{PinActionOnSuperExceptional}
  \xymatrix@C=0pt@R=8pt{
    \mathbb{R}^{10,1\vert \mathbf{32}}
    &
      \simeq_{{}_{\mathbb{R}}}
    &
    \mathbf{11}
    \ar@{}[d]|-{
      \!
      \begin{rotate}{-90}
        \scalebox{.9}{
        $\mathclap{\!\simeq}$
        }
      \end{rotate}
    }
    &\oplus&
    \wedge^9 \mathbf{11}
    \ar@{}[d]|-{
      \!
      \begin{rotate}{-90}
        \scalebox{.9}{
        $\mathclap{\!\simeq}$
        }
      \end{rotate}
    }
    &\oplus&
    \wedge^5 \mathbf{11}
    \ar@{}[d]|-{
      \!
      \begin{rotate}{-90}
        \scalebox{.9}{
        $\mathclap{\!\simeq}$
        }
      \end{rotate}
    }
    &\oplus&
    \mathbf{32}
    \ar@{}[d]|-{
      \!
      \begin{rotate}{-90}
        \scalebox{.9}{
        $\mathclap{\!\simeq}$
        }
      \end{rotate}
    }
    &\oplus&
    \mathbf{32}
    \ar@{}[d]|-{
      \!
      \begin{rotate}{-90}
        \scalebox{.9}{
        $\mathclap{\!\simeq}$
        }
      \end{rotate}
    }
    &
    \;\;\;
    \in
    \mathrm{Rep}_{\mathbb{R}}
    \big(
      \mathrm{Pin}^+(10,1)
    \big).
    \\
    &&
    \big\langle
      e^a
    \big\rangle
    &&
    \big\langle
      \epsilon^{a_1 a_2 \cdots e_{11}}
      e_{a_1 a_2}
    \big\rangle
    &&
    \big\langle
      e_{a_1 \cdots a_5}
    \big\rangle
    &&
    \big\langle
      \psi^\alpha
    \big\rangle
    &&
    \big\langle
      \eta^\alpha
    \big\rangle
  }
\end{equation}
\end{prop}
\begin{proof}
  The first statement follows immediately from the second.
  The second statement follows immediately from \cite[4.26]{FSS18}
  (reproduced as \cite[Lemma 3.10]{FSS19d})
  where the action of single reflection operators is given,
  which generate the action of $\mathrm{Pin}^+(10,1)$.
\hfill \end{proof}

\begin{remark}[Probe M9-branes]
\label{Rem-hull}
The result \eqref{PinActionOnSuperExceptional}
rigorously
supports
the proposal \cite[p. 8-9]{Hull98} (where Hull speaks of the ``most natural interpretation'')
that the summand $\wedge^2 (\mathbb{R}^{10,1})^\ast$
in the M-theory extended super Lie algebra should, at least in part,
be interpreted as the Hodge-dual incarnation of 9-brane charge.
\end{remark}

Thus we have the following super-exceptional enhancement of the
super $\mathrm{MK}6$-locus \cite[Thm. 4.3]{HSS18}:

\begin{defn}[{\cite[4]{FSS19d}}]
  \label{TheSuperExceptionalMK6}
  The \emph{super-exceptional $\mathrm{MK}6$-locus}
  is the sub-Lie superalgebra inside the super-exceptional M-spacetime
  of Def. \ref{SuperExceptionalSpacetimeAsLieAlgebra}
  which is fixed by the action
  via Prop. \ref{SpinRepresentationsOfSuperExceptional}
  of the subgroup
  $\mathbb{Z}_2^A \subset \mathrm{SU}(2)_L \subset \mathrm{Spin}(10,1)$
  in \eqref{RepresentationTheoreticQuestion}
  \begin{equation}
    \label{SuperExceptionalBraneLocus}
    \underset{
      \mathclap{
      \mbox{
        \tiny \bf
        \color{darkblue}
                \begin{tabular}{c}
          super-exceptional
          \\
          $\tfrac{1}{2}\mathrm{M5}$-locus
        \end{tabular}
      }
      }
    }{
    \big(
      \mathbb{R}^{10,1\vert \mathbf{16}}
    \big)^{\mathbb{Z}_2^A}
    }
    \;\;
    \subset
    \;\;
    \underset{
      \mathclap{
      \mbox{
        \tiny
        \color{darkblue}
        \bf
        \begin{tabular}{c}
          super-exceptional
          \\
          $\mathrm{MK}6$-locus
        \end{tabular}
      }
      }
    }{
    \big(
      \mathbb{R}^{10,1\vert \mathbf{32}}
    \big)^{\mathbb{Z}_2^A}
    }
    \;\;
    \subset
    \;\;
    \underset{
      \mathclap{
      \mbox{
        \tiny
        \color{darkblue}
        \bf
        \begin{tabular}{c}
          super-exceptional
          \\
          M-theory spacetime
        \end{tabular}
      }
      }
    }{
      \mathbb{R}^{10,1\vert \mathbf{32}}
    }
  \end{equation}
  The {\it super-exceptional $\tfrac{1}{2}\mathrm{M5}$-locus}
  is the further sub superalgebra with fermions fixed
  by $\mathbb{Z}_{2}^{\mathrm{HW}} \subset \mathrm{Pin}^+(10,1)$.
\end{defn}

\medskip

\noindent {\bf Super-exceptional 3-form flux.}
The {\it raison d'{\^e}tre} of the super-exceptional M-spacetime
\eqref{SuperExceptionalSpacetimeAsLieAlgebra}
is that it carries a map
$$
  \xymatrix{
    \mathbb{R}^{10,1\vert \mathbf{32}}_{\mathrm{ex}_s}
    \ar[rr]^-i
    &&
    \mathbb{R}^{10,1\vert \mathbf{32}}
  }
$$
to ordinary $D=11$ $\mathcal{N}-1$ super-spacetime \eqref{SuperMinkowski}
and a universal 3-flux form
\begin{equation}
  H_{\mathrm{ex}_s}
  \;\in\;
  \mathrm{CE}
  \big(
    \mathbb{R}^{10,1\vert \mathbf{32}}_{\mathrm{ex}_s}
  \big)
\end{equation}
that solves the twisted super-form Bianchi identity
\begin{equation}
  \label{H3exTrivializesPullbackOf4Flux}
  d H_{\mathrm{ex}_s}
  \;=\;
  i^\ast G_4
  \,,
\end{equation}
in the base case the M-theory C-field 4-flux $G_4$
has vanishing bosonic component
\begin{equation}
  \label{M2BraneCocycle}
  G_4
  \;=\;
  G_4^{(0)}
  \;:=\;
  \tfrac{1}{2}
  \big\langle
    \psi
      \wedge
      {\pmb\Gamma}_{a_1 a_2}
      \! \cdot
    \psi
  \big\rangle
  \wedge
  e^{a_1} \wedge e^{a_2}
  \,,
\end{equation}
hence in the case that it only has its bifermionic component,
fixed by the torsion constraint of 11d supergravity,
as shown on the right of \eqref{M2BraneCocycle}.
Explicitly, $H_{\mathrm{ex}_s}$ is a polynomial in
wedge products of the super-exceptional vielbein \eqref{SuperExceptionalBianchiIdentities}
of the following form \cite[6]{DAuriaFre82}\cite[3]{BAIPV04}\cite[3.5]{FSS19d},
where the $\oplus$-notation indicates that
we are showing only the monomials that appear, but
(for readability, and since this is all we need here)
not the real coefficients (which are functions of the parameter $s$) that
multiply them:
\begin{equation}
  \label{SuperExceptionalHFlux}
      \begin{aligned}
      \mathllap{
      \overset{
        \mathclap{
        \mbox{
          \tiny
          \color{darkblue} \bf
          \begin{tabular}{c}
            super-exceptional
            \\
            3-flux form
          \end{tabular}
        }
        }
      }{
        H_{\mathrm{ex}_s}
      }
      =
      \;\;\;
      }
      &
      \phantom{\oplus}\;
      e_{\color{darkblue} a_1 a_2}
      \wedge
      e^{\color{darkblue} a_1}
      \wedge
      e^{\color{darkblue} a_2}
      \\
      &
      \oplus
      e_{\color{orangeii} a_1 a_2}
      \wedge
      e^{\color{orangeii} a_2 a_3}
      \wedge
      e_{\color{orangeii} a_3}{}^{\color{orangeii} a_1}
      \\
      &
      \oplus
      e_{
        {\color{orangeii} a_1 a_2}
        \,
        {\color{greenii} b_1 b_2 b_3 }
      }
      \wedge
      e^{\color{orangeii} a_2 a_3}
      \wedge
      e_{
        \color{orangeii} a_3
      }
      {}^{
        \color{orangeii} a_1
        \,
        {\color{greenii} b_1 b_2 b_3 }
      }
      \\
      & \oplus
      \epsilon^{
        { \color{orangeii} a_0 a_1 a_2 a_3 a_4 }
        { \color{darkblue} a_{} }
        { \color{greenii} a_5 a_6 a_7 a_8 a_9 }
      }
      e_{ \color{orangeii} a_0 a_1 a_2 a_3 a_4 }
      \wedge
      e_{ \color{darkblue} a_{} }
      \wedge
      e_{ \color{greenii} a_5 a_6 a_7 a_8 a_9 }
      \\
      &
      \phantom{\mathclap{\vert^{\vert^{\vert^{\vert}}}}}
      \oplus
      \epsilon^{
        { \color{orangeii} a_1 a_2 a_3 a_4 a_5 a_6}
        { \color{greenii} b_1 b_2 b_3 b_4 b_5  }
      }
      \,
      e_{
        { \color{orangeii} a_1 a_2 a_3 }
        { \color{darkblue} c_1 c_2 }
      }
      \wedge
      e_{
        { \color{orangeii} a_4 a_5 a_6 }
      }
      {}^{
        { \color{darkblue} c_1 c_2 }
      }
      \wedge
      e_{ \color{greenii} b_1 b_2 b_3 b_4 b_5 }
      \\
      &
      \phantom{\mathclap{\vert^{\vert^{\vert^{\vert}}}}}
      \oplus
      \big\langle
        \eta
        \wedge
        {\pmb \Gamma}_{ \color{darkblue} a }
        \!\cdot\!
        \psi
      \big\rangle
      \wedge e^{\color{darkblue} a }
      \;\;\oplus\;\;
      \big\langle
        \eta
        \wedge
        {\pmb \Gamma}^{ \color{greenii} a_1 a_2}
        \!\cdot\!
        \psi
      \big\rangle
      \wedge e_{\color{greenii} a_1 a_2}
      \;\;\oplus\;\;
      \big\langle
        \eta
        \wedge
        {\pmb \Gamma}^{ \color{orangeii}  a_1 a_2 a_3 a_4 a_5}
        \!\cdot\!
        \psi
      \big\rangle
      \wedge e_{ \color{orangeii} a_1 a_2 a_3 a_4 a_5}
      \,.
    \end{aligned}
\end{equation}

\medskip

\noindent {\bf Super-exceptional PS-Lagrangian.}
From this, we found in \cite[Prop. 5.9]{FSS19d}
the super-exceptional solution to the next twisted Bianchi identity
\begin{equation}
  d \mathbf{L}_{\mathrm{ex}_s}
  \;=\;
  2
  \big(
    \tfrac{1}{2}
    H_{\mathrm{ex}_s} \wedge i^\ast G_4
    +
    i^\ast G_7
  \big)^{\mathrm{vert}}
\end{equation}
characterizing the gauge sector of the heterotic M5-brane Lagrangian.
Still in the base case that the dual 7-flux
has vanishing bosonic component
\begin{equation}
  \label{M5BraneCocycle}
  G_7
  \;=\;
  \tfrac{1}{5!}
  \big\langle
    \psi
    \wedge
    {\pmb\Gamma}_{a_1 \cdots a_5}
    \cdot
    \psi
  \big\rangle
  \wedge
  e^{a_1} \wedge \cdots \wedge e^{a_5}
  \,,
\end{equation}
hence in the case that it only has its bifermionic component
fixed by the torsion constraint of 11d supergravity,
as shown on the right of \eqref{M5BraneCocycle},
this is given by the
following {\it super-exceptional PS-Lagrangian} {\cite[(72)]{FSS19d}}:
\begin{equation}
\label{SuperExceptionalPSLagrangianSummands}
\hspace{-3mm}
  \raisebox{-50pt}{
  $
  \begin{aligned}
{
  \overset{
    \mathclap{
    \mbox{
      \tiny
      \color{darkblue} \bf
      \begin{tabular}{c}
        Super-exceptional
        \\
        PS-Lagrangian
        \\
        {\phantom{a}}
      \end{tabular}
    }
    }
  }{
    \mathbf{L}^{\!\mathrm{PS}}_{\mathrm{ex}_s}
  }
  }
  \; = \;\;\;
  &
  \raisebox{0pt}{$
  \begin{aligned}
  &
  \big(\,
  \underset{
    \mathclap{
    \mbox{
      \tiny
      \color{darkblue} \bf
      \begin{tabular}{c}
        contraction with
        super-
        \\exceptional
        lift
        of isometry $v_{5'}$
      \end{tabular}
    }
    }
  }{
  \underbrace{
  (
    \iota_{5'}
    \oplus
    \iota_{5' 6 7 8 9}
  )
  }
  }
  \overbrace{
    H_{\mathrm{ex}_s}
  }
  \big)
  \,\wedge\;
  \big(
  \overset{
    \mathclap{
    \!\!\!\!\!\!\!\!\!\!\!\!\!\!\!\!\!\!\!
    \!\!\!\!\!\!\!\!\!\!\!\!\!\!\!\!\!\!\!
    \!\!\!\!\!\!\!\!\!\!\!\!\!\!\!\!
    \mbox{
      \tiny
      \color{darkblue} \bf
      \begin{tabular}{c}
        super-exceptional
        3-flux form
      \end{tabular}
    }
    }
  }{\,
    \overbrace{
      H_{\mathrm{ex}_s}
    }
  }
  \underset{
    \mathclap{
    \mbox{
      \tiny
      \color{darkblue} \bf
      \begin{tabular}{c}
        dual form to
        \\
        isometry $v_{5'}$
      \end{tabular}
    }
    }
  }{\wedge
  \underbrace{
    e^{5'}
  }
  }
  \big)
  \end{aligned}
  $}
  \\
  \\
  & =
  \raisebox{-50pt}{$
  \begin{aligned}
  &
  \;\phantom{\wedge}
  \left(
    \begin{aligned}
    &\phantom{+}\;
    e_{ {\color{darkblue} a} 5'} \wedge e^{\color{darkblue} a}
    \\
    &
    \oplus
    \;
    \epsilon^{
      { \color{orangeii}
        a_0 a_1 a_2 a_3 a_4
      }
      \,
      5'
      \,
      {
        \color{greenii}
              a_{5} a_6 a_7 a_8 a_9
      }
    }
    e_{\color{orangeii} a_0 a_1 a_2 a_3 a_4}
    \wedge
    e_{\color{greenii} a_5 a_6 a_7 a_8 a_9}
    \\
    &
    \oplus
    e^{ {\color{darkblue} a } 5' }
      \wedge
    e_{ { \color{darkblue} a } 6 7 8 9 }
    \\
    &
    \oplus
    \epsilon^{
      { \color{orangeii}
        a_0 a_1 a_2 a_3 a_4
      }
      \,
      {
        \color{darkblue}
        a
      }
      \,
      {
        5' 6 7 8 9
      }
    }
    e_{\color{orangeii} a_0 a_1 a_2 a_3 a_4}
    \wedge
    e_{\color{darkblue} a }
    \\
    &
    \oplus
    \epsilon^{
      {\color{orangeii} a_1 a_2 a_3 a_4 a_5 a_6 }
      5' 6 7 8 9
    }
    e_{
      \color{orangeii} a_1 a_2 a_3
      \color{darkblue} c_1 c_2
    }
    \wedge
    e_{
      \color{orangeii} a_4 a_5 a_6
    }{}^{
      \color{darkblue} c_1 c_2
    }
    \\
    & \oplus
    \big\langle
      \eta \wedge {\pmb\Gamma}_{5'} \!\cdot\! \psi
    \big\rangle
    \end{aligned}
  \right)
  \\
  &
  \wedge
    \left(
      \begin{aligned}
      &
      \phantom{\oplus}\;
      e_{\color{darkblue} a_1 a_2}
      \wedge
      e^{\color{darkblue} a_1}
      \wedge
      e^{\color{darkblue} a_2}
      \\
      &
      \oplus
      e_{\color{orangeii} a_1 a_2}
      \wedge
      e^{\color{orangeii} a_2 a_3}
      \wedge
      e_{\color{orangeii} a_3}{}^{\color{orangeii} a_1}
      \\
      &
      \oplus
      e_{
        {\color{orangeii} a_1 a_2}
        \,
        {\color{greenii} b_1 b_2 b_3 }
      }
      \wedge
      e^{\color{orangeii} a_2 a_3}
      \wedge
      e_{
        \color{orangeii} a_3
      }
      {}^{
        \color{orangeii} a_1
        \,
        {\color{greenii} b_1 b_2 b_3 }
      }
      \\
      & \oplus
      \epsilon^{
        { \color{orangeii} a_0 a_1 a_2 a_3 a_4 }
        { \color{darkblue} a_{} }
        { \color{greenii} a_5 a_6 a_7 a_8 a_9 }
      }
      e_{ \color{orangeii} a_0 a_1 a_2 a_3 a_4 }
      \wedge
      e_{ \color{darkblue} a_{} }
      \wedge
      e_{ \color{greenii} a_5 a_6 a_7 a_8 a_9 }
      \\
      &
      \phantom{\mathclap{\vert^{\vert^{\vert^{\vert}}}}}
      \oplus
      \epsilon^{
        { \color{orangeii} a_1 a_2 a_3 a_4 a_5 a_6}
        { \color{greenii} b_1 b_2 b_3 b_4 b_5  }
      }
      \,
      e_{
        { \color{orangeii} a_1 a_2 a_3 }
        { \color{darkblue} c_1 c_2 }
      }
      \wedge
      e_{
        { \color{orangeii} a_4 a_5 a_6 }
      }
      {}^{
        { \color{darkblue} c_1 c_2 }
      }
      \wedge
      e_{ \color{greenii} b_1 b_2 b_3 b_4 b_5 }
      \\
      &
      \phantom{\mathclap{\vert^{\vert^{\vert^{\vert}}}}}
      \oplus
      \big\langle
        \eta
        \wedge
        {\pmb \Gamma}_{ \color{darkblue} a }
        \!\cdot\!
        \psi
      \big\rangle
      \wedge e^{\color{darkblue} a }
      \;\;\oplus\;\;
      \big\langle
        \eta
        \wedge
        {\pmb \Gamma}^{ \color{greenii} a_1 a_2}
        \!\cdot\!
        \psi
      \big\rangle
      \wedge e_{\color{greenii} a_1 a_2}
      \;\;\oplus\;\;
      \big\langle
        \eta
        \wedge
        {\pmb \Gamma}^{ \color{orangeii}  a_1 a_2 a_3 a_4 a_5}
        \!\cdot\!
        \psi
      \big\rangle
      \wedge e_{ \color{orangeii} a_1 a_2 a_3 a_4 a_5}
      \end{aligned}
    \right)
  \mathrlap{
    \wedge
    e^{5'}
  }
  \end{aligned}
  $}
  \\
  \\
  & =
  \raisebox{-7pt}{$
  \begin{aligned}
  &
  \phantom{\oplus}\;
  \underset{
    \mathclap{
    \mbox{
      \tiny
      \color{darkblue} \bf
      \begin{tabular}{c}
        free Perry-Schwarz Lagrangian
      \end{tabular}
    }
    }
  }{
  \underbrace{
  e_{ {\color{darkblue} \alpha_1 } 5' }
  \wedge
  e_{ \color{darkblue} \alpha_2 \alpha_3 }
  \wedge
  e^{ \color{darkblue} \alpha_1 }
  \wedge
  e^{ \color{darkblue} \alpha_2 }
  \wedge
  e^{ \color{darkblue} \alpha_3 }
  \wedge
  e^{5'}
  }
  }
  \;\;
    \oplus
  \;\;
  \underset{
    \mathclap{
    \mbox{
      \tiny
      \color{darkblue} \bf
      \begin{tabular}{c}
        super-exceptional
        \\
        correction terms
      \end{tabular}
    }
    }
  }{
    \underbrace{
      \overset{
        \mbox{
          \tiny
          \color{darkblue}
          \bf
          WZW-term
        }
      }{
      \overbrace{
      e_{
       { \color{darkblue} a }
       5'
      }
      \wedge
      e^{
        \color{darkblue} a
      }
      \;\wedge\;
      e_{ \color{orangeii} a_1 a_2}
      \wedge
      e^{ \color{orangeii} a_2 a_3}
      \wedge
      e_{ \color{orangeii} a_3}{}^{ \color{orangeii} a_1}
      \;\wedge\;
      e^{5'}
      }
      }
        \;\;\oplus\;\;
        \cdots
        }
      }
  \end{aligned}
  $}
  \end{aligned}
  $
  }
\end{equation}
  This expression is the definition
  {\cite[(72)]{FSS19d}} spelled out,
  using \cite[(44)]{FSS19d}
  and discarding some vielbein generators
  with an odd number of indices
  $a \in \{6, 7, 8, 9\}$
  in the fourth summand of the first wedge factor,
  since these vanish on the $\mathrm{MK}6$-locus,
  according to Theorem \ref{FieldsAtSingularity} below.

 \medskip

\noindent {\bf Super-exceptional sigma-model fields.}
A {\it field configuration} of the super-exceptional sigma-model is
a dgc-superalgebra homomorphism dual to a map from the
super-worldvolume of a heterotic M5-brane to the super-exceptional
$\tfrac{1}{2}\mathrm{M5}$-locus

\vspace{-5mm}
\begin{equation}
  \label{SuperExceptionalFieldConfiguration}
  \xymatrix{
    \Omega^\bullet
    \big(
      \overset{
        \mathclap{
        \mbox{
          \tiny
          \color{darkblue}
          \bf
          \begin{tabular}{c}
            $\mathcal{N}= (1,0)$ M5-brane
            \\
            super-worldvolume
          \end{tabular}
        }
        }
      }{
        \mathbb{R}^{5,1\vert \mathbf{8}}
      }
    \big)
    \ar@{<-}[rrr]^-{ \sigma^\ast }_-{
      \mathclap{
      \mbox{
        \tiny
        \color{darkblue}
        \bf
        \begin{tabular}{c}
          super-exceptional
          \\
          sigma-model field
        \end{tabular}
      }
      }
    }
    &&&
    \mathrm{CE}
    \Big(
      \overset{
        \mathclap{
        \mbox{
          \tiny
          \color{darkblue}
          \bf
          \begin{tabular}{c}
            super-exceptional
            \\
            $\tfrac{1}{2}\mathrm{M5}$-locus
          \end{tabular}
        }
        }
      }{
      \big(
        \mathbb{R}^{9,1\vert \mathbf{16}}_{\mathrm{ex}_s}
      \big)^{\mathbb{Z}_2^A}
      }
    \Big).
  }
\end{equation}
Here the field components
\begin{equation}
  \label{SuperExceptionalEmbeddingCondition}
  \sigma^\ast
  \big(
    e^\alpha
  \big)
  \;:=\;
  d x^\alpha \;+\; \mathrm{fermions},
  \phantom{AAAA}
  \alpha \in \{0,1,2,3,4,5'\}
\end{equation}
are fixed by the super-embedding condition
and the value of the super-exceptional PS-Lagrangian on such
a field configuration is the pullback of \eqref{SuperExceptionalPSLagrangianSummands}
to the M5-brane worldvolume along this super-exceptional
sigma-model field:
\begin{equation}
  \label{EvaluationOfTheLagrangian}
  \mathbf{L}^{\!\mathrm{PS}}_{\mathrm{ex}_s}(\sigma)
  \;:=\;
  \sigma^\ast
  \big(
    \mathbf{L}^{\!\mathrm{PS}}_{\mathrm{ex}_s}
  \big).
\end{equation}

\noindent {\bf Recovering the free PS-Lagrangian and 4d electromagnetism.}
The further super-exceptional vielbein field components
\begin{equation}
  \label{PhotonFieldComponent}
  \sigma^\ast\big( e_{\mu 5'} \big)
  \;=\;
  d (A_\mu) \;+\; \mathrm{fermions},
  \phantom{AAA}
  \mu \in \{0,1,2,3\},
\end{equation}
are interpreted as those of an electromagnetic vector potential $A$
with field strength $F := d A$
in 4-dimensional spacetime.
The condition of Hodge self-duality
\begin{equation}
  \star_{{}_6} H_3
  \;=\;
  H_3
\end{equation}
of the ordinary flux 3-form
\begin{equation}
  H_3(\sigma)
  \;:=\;
  \sigma^\ast
  \big(
    e_{\alpha_1 \alpha_2}
    \wedge
    e^{\alpha_1}
    \wedge
    e^{\alpha_2}
  \big)
\end{equation}
then enforces
\begin{equation}
  \label{DualPhotonField}
  \sigma^\ast
  \big(
    e_{\mu_1 \mu_2}
  \big)
  \;=\;
  \tfrac{1}{2}
  \big(
    \star_{{}_4} F
  \big)_{\mu_1 \mu_2}
  d x^4
\end{equation}
which is the corresponding Hodge dual field strength times $d x^4$.
With this, the first summand in 
the second line of \eqref{SuperExceptionalPSLagrangianSummands},
which gives \cite[Prop. 5.1]{FSS19d} the
Henneaux-Teitelboim-Perry-Schwarz Lagrangian for a free self-dual
3-form field \cite[(17)]{PerrySchwarz97}\cite{HenneauxTeitelboim88}
\begin{equation}
  \label{TheMaxwellLagrangian}
  \mathbf{L}^{\!\mathrm{PS}_{\mathrm{free}}}_{\mathrm{ex}_s}
  \;=\;
  e_{
    { \color{darkblue} \alpha }
    5'
    }
    \wedge
    e^{
      \color{darkblue} \alpha
    }
  \;\wedge\;
  e_{ \color{darkblue} \alpha_1 \alpha_2 }
  \wedge
  e^{
    \color{darkblue} \alpha_1
  }
  \wedge
  e^{
    \color{darkblue} \alpha_2
  }
  \;\wedge\;
  e^{5'}
  \,,
\end{equation}
evaluates to the Maxwell Lagrangian for source-free 
electromagnetism in 4d:
\begin{equation}
  \label{MaxwellLagrangianAppears}
  \sigma^\ast
  \big(
    \mathbf{L}^{\!\mathrm{PS}_{\mathrm{free}}}_{\mathrm{ex}_s}
  \big)
  \;=\;
  \big(
    d A \wedge \star_{{}_4} d A
  \big)
  \wedge d x^4 \wedge d x^{5'}
  \;+\;
  \mathrm{fermions}.
\end{equation}
This is, at its core, the super-exceptional incarnation
from \cite{FSS19d} of the Perry-Schwarz
mechanism \cite{PerrySchwarz97}\cite{Schwarz97}
in the construction of the M5-brane action functional.

\medskip

\noindent {\bf Beyond the free PS-Lagrangian.}
But the full super-exceptional
Lagrangian \eqref{SuperExceptionalPSLagrangianSummands}
evidently has many more terms than just \eqref{TheMaxwellLagrangian},
describing a rich interacting worldvolume theory on the
M5-brane.
In particular, next there is a term of Wess-Zumino form
\begin{equation}
  \label{SuperExceptionalWZWTerm}
  \mathbf{L}^{\!\mathrm{WZW}}_{\mathrm{ex}_s}
  \;:=\;
  \big(
    \iota_{v^{5'} \oplus v^{5' 6 7 8 9}}
    H_{\mathrm{ex}_s}
  \big)
  \;\wedge\;
  e_{ \color{orangeii} a_1 a_2}
  \wedge
  e^{ \color{orangeii} a_2 a_3}
  \wedge
  e_{ \color{orangeii} a_3}{}^{ \color{orangeii} a_1}
  \;\wedge\;
  e^{5'}
  \,.
\end{equation}
We analyze the value of this term on super-exceptional sigma-model fields
below in \cref{TheWZWTermsInTheLagrangian}.

\medskip
But first we turn now to the classification of the super-exceptional
field content.

\section{Classification of the field content}
\label{FieldContent}

Here we prove the classification, shown in the big diagram \eqref{TheFieldContent},
of the super-exceptional vielbein fields on super-exceptional M-theory spacetime
(Def. \ref{SuperExceptionalSpacetimeAsLieAlgebra})
restricted to the super-exceptional $\mathrm{MK}6$-locus
(Def. \ref{TheSuperExceptionalMK6}) regarded as a representation
of the residual symmetry group
$\mathrm{Spin}(3,1) \times \mathrm{U}(1)_V \times \mathrm{SU}(2)_R$,
according to decompositions \eqref{RepresentationTheoreticQuestion}.

\medskip

Much of the table \eqref{TheFieldContent}
follows from straightforward branching of
exterior power representations, recalled as
Example \ref{RestrictionOfWedgePowerOfVectorRep} ,
and from the familiar Spin representation branching

\vspace{-.3cm}

$$
  \xymatrix@R=-1pt{
    \mathrm{Rep}_{\mathbb{R}}
    \big(
      \mathrm{Spin}(10,1)
    \big)
    \ar[rr]^{  }
    &&
    \mathrm{Rep}_{\mathbb{R}}
    \big(
      \mathrm{Spin}(9,1)
    \big)
    \ar[rr]^-{ (-)^{\mathbb{Z}_2^A} }
    &&
    \mathrm{Rep}_{\mathbb{R}}
    \big(
      \mathrm{Spin}(5,1)
    \big)
    \\
     \mathbf{32}
     \ar@{|->}[rr]
     &&
     \mathbf{16}
       \oplus
     \overline{
       \mathbf{16}
     }
     \ar@{|->}[rr]
     &&
     \mathbf{8} \oplus \overline{\mathbf{8}}
     \\
     \mathllap{
       \mbox{
         recalled as:
       }
       \;\;\;\;\;\;\;\;\;\;\;\;\;\;\;\;\;\;\;\;\;\;\;
     }
     \ar@{}[rr]|-{
       \mbox{
         Remark \ref{BranchingOf11dSpinorsIn10d}
       }
     }
     &&
     \ar@{}[rr]|-{
       \mbox{
         Lemma \ref{FixedSubspaceOfGamma6789}
       }
     }
     &&
  }
$$

\medskip

The further statements we need needs are captured in the following result:

\newpage

\begin{theorem}[Fields at $\mathrm{A}_1$-singularity]
 \label{FieldsAtSingularity}
Under passage to the fixed locus \eqref{FormingFixedPoints}
of the
$\mathbb{Z}_2^A \subset \mathrm{SU}(2)_L$-action
in \eqref{RepresentationTheoreticQuestion}
we have the following representations under the residual
group actions:

\vspace{.3cm}

{\small
\begin{tabular}{|c|l|}
 \hline
  $
    \raisebox{38pt}{
    \xymatrix@R=1pt{
      \mathclap{\phantom{\vert^{\vert^{\vert^{\vert^{\vert}}}}}}
      \;
      \mathrm{Rep}_{\mathbb{R}}
      \big(
        \mathrm{SU}(2)_L \times \mathrm{SU}(2)_R
      \big)
      \;
      \ar@{->}[rr]^-{ (-)^{\mathbb{Z}_2^A} }
      &&
      \;
      \mathrm{Rep}_{\mathbb{R}}
      \big(
        \mathrm{SU}(2)_R
      \big)
      \;
      \\
      \wedge^2 \mathbf{4}
      \ar@{|->}[rr]
      \ar@{}[d]|-{
        \scalebox{.9}{
        \begin{rotate}{-90}
          $\mathclap{\simeq}$
        \end{rotate}
        }
      }
      &&
      \mathbf{3}
      \ar@{}[d]|-{
        \scalebox{.9}{
        \begin{rotate}{-90}
          $\mathclap{\simeq}$
        \end{rotate}
        }
      }
      \\
      \big\langle
        e^{a_1 a_2}
      \big\rangle_{\mathrlap{a_i \in \{6,7,8,9\}}}
      &&
      \big\langle
        \tfrac{1}{2} e^{ 6 I  }
        -
        \tfrac{1}{2}
        \epsilon^{I J K} e_{J K}
      \big\rangle_{\mathrlap{I \in \{7,8,9\}}}
     \mathclap{\phantom{\vert_{\vert_{\vert_{\vert{\vert}}}}}}
     \\
     \mathclap{\phantom{\vert^{\vert^{\vert^{\vert^{\vert}}}}}}
     \wedge^1 \mathbf{4},\,  \wedge^3 \mathbf{4}
     \ar@{|->}[rr]
     &&
     0
    }
    }
    \mathrlap{\phantom{\vert_{\vert_{\vert_{\vert{\vert}}}}}}
  $
  &
  \begin{tabular}{l}
    \rm
    discussed in \cref{TheA1TypeSingularity};
    \\
    \rm
    see Prop. \ref{3Inside4Wedge3}
  \end{tabular}
  \\
  \hline
  $
    \raisebox{43pt}{
    \xymatrix@R=1pt{
      \mathclap{\phantom{\vert^{\vert^{\vert^{\vert^{\vert}}}}}}
      \mathrm{Rep}_{\mathbb{R}}
      \big(
      \mathrm{Spin}(5,1)
      \times
      \mathrm{SU}(2)_R
      \big)
      \ar[rr]^-{  }
      &&
      \mathrm{Rep}_{\mathbb{R}}
      \big(
        \mathrm{Spin}(3,1)
        \times
        \mathrm{U}(1)_V
        \times
        \mathrm{SU}(2)_R
      \big)
      \\
      \mathbf{8}
      \ar@{|->}[rr]
      \ar@{}[d]|-{
        \scalebox{.9}{
        \begin{rotate}{-90}
          $\mathclap{
            \,
            \underset{
              \scalebox{.6}{$\mathbb{R}$}
            }{
              \simeq
            }
          }$
        \end{rotate}
        }
      }
      &&
      \mathbf{2}_{\mathbb{C}}
      \boxtimes
      \mathbf{1}^{\!-}_{\mathbb{C}}
      \boxtimes
      \mathbf{2}_{\mathbb{C}}
      \ar@{}[d]|-{
        \scalebox{.9}{
        \begin{rotate}{-90}
          $\mathclap{
            \,
            \underset{
              \scalebox{.6}{$\mathbb{R}$}
            }{\simeq}
          }$
        \end{rotate}
        }
      }
      \\
      \scalebox{.9}{
      $
      {\begin{array}{c}
        \mathbb{H}
        \\
        \oplus
        \\
        \mathbb{H}
      \end{array}}
      $
      }
      &&
      \scalebox{.9}{
      $
      {\begin{array}{ccc}
        \mathbb{C} &\oplus& \mathbb{C} \mathrm{j}
        \\
        \oplus && \oplus
        \\
        \mathbb{C} &\oplus& \mathbb{C} \mathrm{j}
      \end{array}}
      $
      }
      \mathclap{\phantom{\vert_{\vert_{\vert_{\vert_{\vert}}}}}}
      \\
      \overline{
        \mathbf{8}
      }
      \ar@{|->}[rr]
      \ar@{}[u]|-{
        \scalebox{.9}{
        \begin{rotate}{-90}
          $\mathclap{
            \underset{
              \scalebox{.6}{$\mathbb{R}$}
            }{
              \simeq
            }
            \,
          }$
        \end{rotate}
        }
      }
      &&
      \overline{
        \mathbf{2}
      }_{\mathbb{C}}
      \boxtimes
        \mathbf{1}^{\!-}_{\mathbb{C}}
      \boxtimes
        \mathbf{2}_{\mathbb{C}}
      \ar@{}[u]|-{
        \scalebox{.9}{
        \begin{rotate}{-90}
          $\mathclap{
            \,
            \underset{
              \scalebox{.6}{$\mathbb{R}$}
            }{
              \simeq
            }
          }$
        \end{rotate}
      }
      }
    }
    }
  $
  &
  \begin{tabular}{l}
    \rm
    discussed in \cref{11dSpinorsAtTheA1TypeSingularity};
    \\
    \rm
    see Prop. \ref{TheReducedSpinRep}
  \end{tabular}
  \\
  \hline
\end{tabular}
}
\end{theorem}

\vspace{2mm}
\begin{remark}[Real vs. complex representations]
  The classification in \eqref{TheFieldContent}
  is as \emph{real} representations, since all
  supersymmetry algebras \eqref{SuperMinkowski}
  are based on real Spin representations (e.g. \cite[3]{Freed99}).
  Now, the two complex Weyl spinor representations
  $\mathbf{2}_{\mathbb{C}}, \overline{\mathbf{2}}_{\mathbb{C}}$
 (see  \eqref{ComplexWeylRepresentations} below)
  of $\mathrm{Spin}(3,1)$, which are  distinct as complex
  representations, actually become isomorphic
  when regarded as real representations
  (recalled as Lemma \ref{RealReps4AndBar4OfSpin31AreIsomorphic} below).
  But this degeneracy is lifted by the further action of
  $\mathrm{U}(1)_V$: The
  real representations underlying their outer complex tensor product
  \eqref{RestrictionOfRepresentations} with the
  $\mathbf{1}^{\!-}_{\mathbb{C}}$ of $\mathrm{U}(1)$
  \eqref{RepresentationsOfU1}
  are \emph{not} isomorphic even as real representations (Lemma \ref{IsoclassesOfTheFullReps}):
  $
    \mathbf{2}_{\mathbb{C}}
      \underset{\mathbb{C}}{\boxtimes}
    \mathbf{1}^{\!-}_{\mathbb{C}}
      \underset{\mathbb{C}}{\boxtimes}
    \mathbf{2}_{\mathbb{C}}
    \;
    \simeq_{\mathbb{R}}
    \!\!\!\!\!\!\!\!\!\!\!\!\!/\,\,\,\,\,\,\,\,\,\,\,
    \;
    \overline{\mathbf{2}}_{\mathbb{C}}
      \underset{\mathbb{C}}{\boxtimes}
    \mathbf{1}^{\!-}_{\mathbb{C}}
      \underset{\mathbb{C}}{\boxtimes}
    \mathbf{2}_{\mathbb{C}}
    \,.
  $

\end{remark}

\begin{remark}[Gauge enhancement at ADE-singularities]
  A famous informal argument suggests that M2-branes on around
  vanishing 2-cycles inside an $\mathrm{ADE}$-singularity
  appear as $\mathrm{SU}(2)$ gauge fields on the
  transversal D-branes (e.g., \cite[3.1.2]{Acharya02}\cite[Ex. 2.2.5]{HSS18}).
  Now, if we interpret:

  \begin{enumerate}[{\bf (a)}]
  \vspace{-2mm}
  \item
  The exceptional vielbein components $e_{a_1 a_2}$ in \eqref{TheFieldContent}
  as being charges of M2-branes
  stretched along the directions $v_{a_1} \wedge v_{a_2}$;

  \vspace{-3mm}
  \item elements of
    $
      \wedge^2 \mathbf{4}
      =
      \wedge^2
      \mathbb{H}
    $
  as 2-cocycles on the transversal
  Euclidean conical orbifold geometry \eqref{TheBranes};

\vspace{-2mm}
  \item elements of the fixed locus
    $
      \big(
        \wedge^2 \mathbf{4}
      \big)^{\mathbb{Z}_2^A}
      =
      \big\langle
        e^{I 6}
      \big\rangle
    $
  as the restriction of these 2-cocyles to the
  singular point of the $\mathrm{SU}(2)_L$-action,
  hence to their evaluation on 2-cycles that are shrunken into the
  singularity;
  \end{enumerate}

  \vspace{-2mm}
\noindent  then this informal story becomes the statement
  of
  Theorem \ref{FieldsAtSingularity}.
  Explicitly, with the identification of Lemma
  \ref{IrrepDecompositionOf4Wedge4}, we have:
    $$
    \big\langle
    \mathrm{e}^{I 6}
    +
    \epsilon^{I J K}
    \mathrm{e}_{I J}
    \,,
    \;\;
    I,J \in \{7,8,9\}
    \big\rangle
    \;\longleftrightarrow\;
    \left\{ \!\!\!\!\!\!\!
      \mbox{ \small
        \begin{tabular}{c}
          Charges of M2-branes wrapped on
          \\
          vanishing 2-cycles in ADE-singularity
        \end{tabular}
      }
   \!\!\!\!\!\!\! \right\}.
  $$
\end{remark}

\medskip

\medskip

In the remainder we spell out in detail the proof of Theorem \ref{FieldsAtSingularity}.
This becomes nicely transparent in terms of octonionic 2-component spinor
representations.
Since this is not as widely known as it deserves to be,
we use the occasion to recall all the ingredients,
such as to make the proof fully self-contained.

\medskip

\subsection{Background in representation theory}
\label{ElementsOfRepresentation Theory}

For reference and to fix conventions,
we briefly recall some basic concepts of representation theory.

\vspace{-2mm}
\paragraph{Representation ring.} For $G$ a group, we write $\mathrm{Rep}_{\mathbb{R}}(G)$ for its
\emph{real representation ring}: elements are isomorphism
classes of real-linear finite-dimensional $G$-representations,
addition in the ring comes from direct sum
of representations, and the product in the ring from the
tensor product of representations.

\vspace{-2mm}
\paragraph{Irreducible representation.}
We denote irreducible representations of $G$ by
their dimension, typeset in boldface, equipped with decorations in case
there are inequivalent irreps of the same dimension.
Then we write
\begin{equation}
  k \cdot \mathbf{N}
  \;:=\;
  \underset{
    \mbox{
      \tiny
      $k$ summands
    }
  }{
  \underbrace{
    \mathbf{N} \oplus \cdots \oplus \mathbf{N}
  }
  }
  \;\in\;
  \mathrm{Rep}_{\mathbb{R}}(G)
\end{equation}
for the $k$-fold direct sum of an irrep with itself, of total dimension
$\mathrm{dim}_{{}_{\mathbb{R}}}(k \cdot N) = k \, N$.

\medskip

\noindent {\bf Exterior power representations.}
The representation ring is a ``lambda-ring'' in that for any representation
$V \in \mathrm{Rep}_{\mathbb{R}}(G)$ and $p \in \mathbb{N}$ we have the
$k$-fold exterior power representation
\begin{equation}
  \label{WedgeProductRepresentation}
  \wedge^p V \;\; \in \; \mathrm{Rep}_{\mathbb{R}}(G)\;.
\end{equation}

\vspace{-7mm}
\paragraph{Restricted representations.}
For $G_1 \overset{f}{\longrightarrow} G_2$ a group homomorphism
there is a representation-ring homomorphism
given by regarding a $G_2$ representation $V$ as a $G_1$-representation
by acting via $f$:
\begin{equation}
  \label{RestrictionOfRepresentations}
  \xymatrix@R=-4pt{
    \mathrm{Rep}_{\mathbb{R}}(G_2)
    \ar[rr]^-{f^\ast}
    &&
    \mathrm{Rep}_{\mathbb{R}}(G_1)
    \\
    \big[
      (g_2,v) \mapsto g_2 \cdot v
    \big]
    \ar@{|->}[rr]
    &&
    \big[
      (g_1,v) \mapsto f(g_1) \cdot v
    \big]
  }
\end{equation}
Specifically when $f : H \overset{\iota}{\hookrightarrow} G$ is an inclusion of subgroups, then
forming \emph{restricted representations} $f^\ast$ as in \eqref{RestrictionOfRepresentations}
is also called the ``branching of representations'' under
``breaking of symmetry'' from $G$ to $H$,
since irreps on the left will in general
``branch'' into direct sums of irreps on the right:
\begin{equation}
  \label{BranchingOfRepresentations}
  \xymatrix@R=-4pt{
    \mathrm{Rep}_{\mathbb{R}}(G)
    \ar[rr]^-{\iota^\ast}
    &&
    \mathrm{Rep}_{\mathbb{R}}(H)
    \\
    \underset{
      \mathclap{
      \mbox{
        \tiny
        \color{darkblue} \bf
        \begin{tabular}{c}
          irrep of $G$
        \end{tabular}
      }
      }
    }{
    {
      \mathbf{N}
    }
    }
    \ar@{|->}[rr]
    &&
    \underset{
      \mbox{
        \tiny
        \color{darkblue} \bf
        direct sum of irreps of $H$
      }
    }{
      \mathbf{n_1} \oplus \cdots \oplus \mathbf{n_k}
    }
  }
\end{equation}

\vspace{-3mm}
\paragraph{Outer tensor product.}
When $G = H_1 \times H_2$ is a direct product group
then the operation of
forming the tensor product as $G$-representations
of an $H_1$- and an $H_2$-representation,
both regarded as $G$-representations
under
\eqref{RestrictionOfRepresentations}
via the projection homomorphisms
$G \overset{ \mathrm{pr}_i }{ \longrightarrow } H_i$,
is also called the \emph{outer tensor product} of $H_i$-representations
and denoted by a square tensor product symbol:
\begin{equation}
  \label{OuterTensorProductOfRepresentations}
  \boxtimes
  \;:\;
  \xymatrix{
    \mathrm{Rep}_{\mathbb{R}}(H_1)
    \times
    \mathrm{Rep}_{\mathbb{R}}(H_2)
    \ar[rr]^-{ \mathrm{pr}_1^\ast \,\times \, \mathrm{pr}_2^\ast }
    &&
    \mathrm{Rep}_{\mathbb{R}}(H_1 \times H_2)
    \times
    \mathrm{Rep}_{\mathbb{R}}(H_1 \times H_2)
    \ar[r]^-{ \otimes }
    &
    \mathrm{Rep}_{\mathbb{R}}(H_1 \times H_2)\,.
  }
\end{equation}

\paragraph{Fixed points.}
Moreover, in this case when $G = H_1 \times H_2$ is a direct product group,
passage to
\emph{$H_2$-fixed points} in the representation
of an $H_1 \times H_2$-representation is an additive functor
(not though a monoidal one)
denoted
\begin{equation}
  \label{FormingFixedPoints}
  (-)^{H_2}
  \;:\;
  \xymatrix@R=-3pt{
    \mathrm{Rep}_{\mathbb{R}}(H_1 \times H_2)
    \ar[r]
    &
 \mathrm{Rep}_{\mathbb{R}}(H_1)
    \\
    V
    \ar@{|->}[r]
    &
    V^{H_2}
    \mathrlap{
      \;
      :=
      \{ v \in V \vert \, h_1 \cdot v = v \} .
    }
  }
\end{equation}
Notice that the $H_2$-fixed points of an outer tensor product \eqref{OuterTensorProductOfRepresentations}
with an irreducible $H_2$-representation
$\mathbf{N} \in \mathrm{Rep}_{\mathbb{R}}(H_2)$ are non-vanishing
precisely if $\mathbf{N}$ is the trivial representation:
\begin{equation}
  \label{H2FixedPointsOfOuterTensorProductWithH2Irrep}
  \big(
    \underset{
      \mathclap{
      \mbox{
        \tiny
        \color{darkblue} \bf
        \begin{tabular}{c}
          any rep
          \\
          of $H_1$
        \end{tabular}
      }
      }
    }{
    \underbrace{
      V
    }
    }
      \boxtimes
    \underset{
      \mathclap{
      \mbox{
        \tiny
        \color{darkblue} \bf
        \begin{tabular}{c}
          irrep
          \\
          of $H_2$
        \end{tabular}
      }
      }
    }{
    \underbrace{
      \mathbf{N}
    }
    }
  \big)^{H_2}
  \;=\;
  \left\{
    \begin{array}{lll}
      V && {\rm if} \; \mathbf{N} = \mathbf{1},
      \\
      0 &  & \mbox{otherwise}.
    \end{array}
  \right.
\end{equation}

\begin{example}[Pauli matrices]
\label{Ex-Pauli}
  The fundamental complex representation of $\mathrm{SU}(2)$
  \begin{equation}
    \label{FundamentalRepresentationofSU2}
    \mathbf{2}_{\mathbb{C}}
    \;\in\;
    \mathrm{Rep}_{\mathbb{C}}\big(
      \mathrm{SU}(2)
    \big)
    \,,
    \phantom{AAA}
    \mathbf{2}_{\mathbb{C}}
    \;\in\;
    \mathrm{Rep}_{\mathbb{C}}\big(
      \mathfrak{su}(2)
    \big)
  \end{equation}
  regarded as the underlying Lie algebra representation,
  indicated by the same symbol in the right, has as representation
  matrices the Pauli matrices, which we normalize as:
  \begin{equation}
    \label{PauliMatrices}
    \tau^1
    \;:=\;
    \tfrac{\mathrm{i}}{2}
    \left[
      \!\!\!
      \begin{array}{cc}
        1 & 0
        \\
        0 & -1
      \end{array}
      \!\!\!
    \right]
    \,,
    \phantom{AA}
    \tau^2
    \;:=\;
    \tfrac{\mathrm{i}}{2}
    \left[
      \!\!\!
      \begin{array}{cc}
        0 & -\mathrm{i}
        \\
        \mathrm{i} & 0
      \end{array}
      \!\!\!
    \right]
    \,,
    \phantom{AA}
    \tau^3
    \;:=\;
    \tfrac{\mathrm{i}}{2}
    \left[
      \!\!\!
      \begin{array}{cc}
        0 & 1
        \\
        1 & 0
      \end{array}
      \!\!\!
    \right]
  \end{equation}
  such that
  \begin{equation}
    \label{LieBracketInSU2}
    [\tau^i, \tau^j]
    \;=\;
    \epsilon^{i j}{}_k \tau^k
  \end{equation}
\end{example}

\begin{example}[Outer product with representations of $U(1)$]
    We denote  the irreducible complex representations of
  $\mathrm{U}(1)$, labeled by $n \in \mathbb{Z}$,
by
  \begin{equation}
    \label{RepresentationsOfU1}
    \mathbf{1}^{n}_{\mathbb{C}}
    \;\in\;
    \mathrm{Rep}_{\mathbb{C}}
    \big(
      \mathrm{U}(1)
    \big)
    \phantom{AA}
    \mbox{abbrieviated for $n = \pm 1$ as}:
    \phantom{A}
    \mathbf{1}^{\pm}_{\mathbb{C}}
    \;\in\;
    \mathrm{Rep}_{\mathbb{C}}
    \big(
      \mathrm{U}(1)
    \big)
    \,.
  \end{equation}
  Hence their outer tensor product \eqref{OuterTensorProductOfRepresentations}
  with the Pauli matrices from Example \ref{Ex-Pauli}
  $$
    \mathbf{1}^n_{\mathbb{C}}
      \underset{\mathbb{C}}{\boxtimes}
    \mathbf{2}_{\mathbb{C}}
    \;\;
    \in
    \;
    \mathrm{Rep}_{\mathbb{C}}
    \big(
      \mathrm{U}(1)
      \times
      \mathrm{SU}(2)
    \big)
  $$
  with the fundamental representation \eqref{FundamentalRepresentationofSU2} of $\mathrm{SU}(2)$
  has as Lie algebra representation matrices the Pauli matrices
  \eqref{PauliMatrices} and one more matrix given by
  \begin{equation}
    \tau^0
    \;:=\;
    \mathrm{i}
    \left[
      \!\!\!
      \begin{array}{cc}
        n & 0
        \\
        0 & n
      \end{array}
      \!\!\!
    \right]
    \,.
  \end{equation}
\end{example}

\begin{example}[Spinorial representations]
\label{SpinorialRepresentations}
For $d \in \mathbb{N}$ we write $\mathbb{R}^{d,1}$ for the
real inner product space with bilinear form
$\eta := \mathrm{diag}(-1,+1,+1, \cdots, +1)$.
We take the corresponding
Clifford algebra $\mathrm{Cl}_{\mathbb{R}}(d,1)$ to be the real associative
algebra generated from $\{{\pmb\Gamma}_a\}_{a = 0}^d$ subject to the relations
\begin{equation}
  \label{CliffordRelation}
  {\pmb\Gamma}_{a_1} {\pmb\Gamma}_{a_2}
  +
  {\pmb\Gamma}_{a_2} {\pmb\Gamma}_{a_1}
  \;=\;
  + 2 \eta_{a b}
  \,,
\end{equation}
and we write

\begin{equation}
  \label{EvenCliffordAlgebra}
  \;\;\;\;
  \mathrm{Cl}^{\mathrm{ev}}_{\mathbb{R}}(d,1)
  \subset
  \mathrm{Cl}_{\mathbb{R}}(d,1)
  \phantom{A}\;
  \mbox{
    for the subalgebra generated from the products
    $\big\{{\pmb \Gamma}_{a_1}{\pmb\Gamma}_{a_2}\big\}_{a_1, a_2}$.
  }
\end{equation}

\noindent Now if $\{R_{a_1 a_2}\}_{a_i = 0}^d$ denotes
the standard linear generators of  the Lie algebra $\mathfrak{so}(d,1)$,
with $R_{a_1 a_2} = - R_{a_2 a_1}$
and with Lie bracket given by
\begin{equation}
  \label{LieBracketInSO}
  [ R_{a_1 a_2}, R_{b_1 b_2}]
  \;=\;
  \eta_{a_2 b_1} R_{a_1 b_2}
  -
  \eta_{a_1 b_1} R_{a_2 b_2}
  +
  \eta_{a_2 b_2} R_{b_1 a_1}
  -
  \eta_{a_1 b_2} R_{b_1 a_2}
\end{equation}
then the assignment
\begin{equation}
  \label{LieSpinRepresentation}
  \xymatrix@R=-6pt{
    \mathfrak{so}(d,1)
    \; \ar@{^{(}->}[rr]
    &&
    \mathrm{Cl}^{\mathrm{ev}}_{\mathbb{R}}(d,1)
    \\
    R_{a_1 a_2}
    \ar@{|->}[rr]
    &&
    \tfrac{1}{4}\big[
      {\pmb\Gamma}_{a_1},
      {\pmb\Gamma}_{a_2}
    \big]
    \mathrlap{
      \;
      =
      \left\{
      \begin{array}{ccc}
        \tfrac{1}{2}{\pmb\Gamma}_{a_1}{\pmb\Gamma}_{a_2}
        &\vert& a_1 \neq a_2
        \\
        0 &\vert& \mbox{otherwise}
      \end{array}
      \right.
    }
  }
\end{equation}
constitutes a linear embedding
which is a Lie algebra morphism
with respect to the commutator bracket on the right
(e.g. \cite[Prop. 6.2]{LawsonMichaelson89}).
As a consequence, every associative algebra representation
of $\mathrm{Cl}^{\mathrm{ev}}_{\mathbb{R}}(d,1)$
(``Clifford module'') becomes a
Lie algebra representation of $\mathfrak{so}(d,1)$,
and thus a Lie group representation of the
corresponding simply connected Lie group $\mathrm{Spin}(d,1)$, via
the exponential map
\begin{equation}
  \label{ExponentiatingCliffordElements}
  \xymatrix@R=-2pt{
    \mathfrak{so}(d,1)
    \ar@{->>}[rr]
    &&
    \mathrm{Spin}(d,1)
    \ar@{}[r]|-{ \subset }
    &
    \mathrm{Cl}^{\mathrm{ev}}_{\mathbb{R}}
    \big(
      d,1
    \big).
    \\
    \alpha R_{a_1 a_2}
    \ar@{|->}[rr]
    &&
    \mathrm{exp}\big(
      \alpha \tfrac{1}{2}{\pmb\Gamma}_{a_1}{\pmb\Gamma}_{a_2}
    \big)
  }
\end{equation}
The representations obtained this way are the
\emph{spinorial} representations, in contrast to the vector representations
\eqref{TheVectorRepresentationOfSpinN} and their exterior powers
\eqref{WedgeProductRepresentation}.

\medskip
Notice that for $a_1, a_2 > 0$,
whence $\big({\pmb\Gamma}_{a_1}{\pmb\Gamma}_{a_2}\big)^2 = -1$,
Euler's formula (which applies in any Banach algebra)
gives
\begin{equation}
  \label{EulerFormulaInCliffordAlgebra}
  \exp
  \big(
    \alpha\, {\pmb\Gamma}_{a_1}{\pmb\Gamma}_{a_2}
  \big)
  \;=\;
  \mathrm{cos}(\alpha)
  +
  \mathrm{sin}(\alpha)
  \,
  {\pmb\Gamma}_{a_1}{\pmb\Gamma}_{a_2}
  \;\;\;\;
  \in
  \mathrm{Cl}^{\mathrm{ev}}_{\mathbb{R}}(d,1)
  \,.
\end{equation}
Therefore the exponent in \eqref{ExponentiatingCliffordElements}
with the prefactor of $1/2$ from \eqref{LieSpinRepresentation},
is such that rotations by an angle of $\alpha = 2 \pi$
are represented by
\begin{equation}
  \label{2PiRotationOnSpinors}
  \exp(2 \pi  \tfrac{1}{2}{\pmb\Gamma}_{a_1}{\pmb\Gamma}_{a_2}) = -1
\end{equation}
and it is only rotations by $\alpha = 4\pi$ that yield the identity
on spinors, reflecting the double covering
\begin{equation}
  \label{SpinDoubleCovering}
  \xymatrix{\mathrm{Spin}(d,1) \ar@{->>}[r] & \mathrm{SO}(d,1)}\!.
\end{equation}
\end{example}

\begin{example}
\label{VectorRepresentationofSpinN}
For $D \in \mathbb{N}$ the
{\it vector representation} of $\mathrm{Spin}(D-1,1)$
(or of $\mathrm{Spin}(D)$)
is the defining representation of $\mathrm{SO}(d,1)$
(or $\mathrm{SO}(D)$)
via \eqref{SpinDoubleCovering}.
This is a $D$-dimensional irrep, which we hence denote
\begin{equation}
  \label{TheVectorRepresentationOfSpinN}
  \mathbf{D}
    \;\in\;
  \mathrm{Rep}_{\mathbb{R}}
  \big(
    \mathrm{Spin}(D-1,1)
  \big)
    \phantom{AAA}
  \mbox{or}
  \phantom{AAA}
  \mathbf{D}
    \;\in\;
  \mathrm{Rep}_{\mathbb{R}}
  \big(
    \mathrm{Spin}(D)
  \big)
  \,.
\end{equation}
\end{example}


\begin{example}[Restriction of exterior power representation]
\label{RestrictionOfWedgePowerOfVectorRep}
For natural numbers $D_1 + D_2 = D$
the restriction \eqref{BranchingOfRepresentations} of an exterior power
\eqref{WedgeProductRepresentation}
of the vector representation $\mathbf{D}$ of $\mathrm{Spin}(D-1,1)$
along the canonical inclusion
$\mathrm{Spin}(D_1-1,1) \times \mathrm{Spin}(D_2)
\overset{\iota}{\hookrightarrow}
\mathrm{Spin}(D-1,1)$ is
\begin{equation}
  \label{BranchingOfWedgePowerOfVectorRep}
  \iota^\ast
  \big(
    \wedge^k \mathbf{D}
  \big)
  \;=\;
  \underset{p \in \{0,\cdots, k\}}{\bigoplus}
  (
    \wedge^p \mathbf{D}_1
  )
  \boxtimes
  (
    \wedge^{k-p} \mathbf{D}_2
  ).
\end{equation}
\end{example}

\subsection{Octonionic 2-component spinors}
\label{Octonionic2ComponentSpinors}

We discuss here real Spin representations
(see Example \ref{SpinorialRepresentations} for spinor conventions) in
spacetime dimensions 11, 10, 7, 6 and 4 in terms of
matrices with coefficients in the octonions
(following \cite{KugoTownsend82}, reviewed in \cite{Baez02}\cite{BH09} \cite{BH10})
which is well-adapted to the geometry of the $\tfrac{1}{2}\mathrm{M}5$-brane locus
\cref{TheA1TypeSingularity}, following \cite{HSS18}.

\medskip

We will find useful the presentation of the octonions
as generated from the quaternions and from one more
imaginary unit $\ell$. This \emph{Dickson double construction}
\cite[(6)]{Dickson1919}
is well-known in itself, but since the explicit minimal set of relations
\eqref{qlp} and \eqref{lqpl} below (highlighted in \cite[2.2]{Baez02})
is not as widely used
(but see \cite[Def. A.6]{HSS18}\cite[Def. 26]{HS18}), we recall it:

\medskip

\hspace{-.9cm}
\begin{tabular}{ll}

\begin{minipage}[l]{10cm}

\begin{lemma}[Octonions by generators-and-relations]
  \label{Octonions}
  The real star-algebra $\mathbb{O}$ of octonions,
  with its star-operation (conjugation) to be denoted $(-)^\ast$,
  is equivalently that generated from the
  algebra of quaternions $\mathbb{H}$ and from one more algebra
  element $\ell$, subject to these relations, for all $q, p \in \mathbb{H}$:
  \begin{equation}
    \ell^2 = -1
    \,,
    \phantom{AA}
    \ell^\ast = - \ell
  \end{equation}
  \begin{equation}
    \label{qlp}
    q (\ell p)
    \;=\;
    \ell (q^\ast p)
    \,,
    \phantom{AAAA}
    (q \ell) p
    \;=\;
    (q p^\ast) \ell
  \end{equation}
  \begin{equation}
    \label{lqpl}
    (\ell q) (p \ell)
    \;=\;
    - (q p)^\ast
    \,.
  \end{equation}
  Applied to
  an orthonormal basis of imaginary unit quaternions,
  $q \in \{{\rm i},{\rm j},{\rm k}\}$,
  this means that $\mathbb{O}$ is generated from the seven imaginary
  unit elements shown in the diagram on the right,
  subject to these relations:
  $
    \;\;\;\;
    a b = c,\;\;\; c a = b,\;\;\; b c = a, \;\;\;{\rm and}\;\;\; b a = - c
  $
  \\
  for every pair of consecutive arrows $a \to b \to c$ shown.
\end{lemma}

\end{minipage}

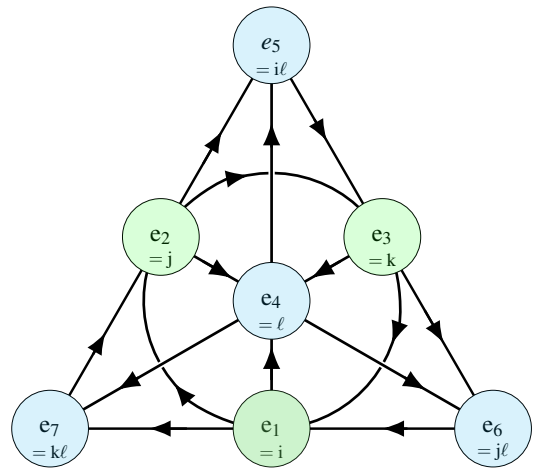
\scalebox{.85}{
\raisebox{-95pt}{
\begin{tikzpicture}

 \begin{scope}[rotate={(-20)}]
 \draw[->, -{>[scale=1.24]}, >=Latex, line width=1.2pt]
   (0+4:2) arc (0+4:-120:2);
 \draw[->, -{>[scale=1.24]}, >=Latex, line width=1.2pt]
   (-120+4:2) arc (-120+4:-240:2);
 \draw[->, -{>[scale=1.24]}, >=Latex, line width=1.2pt]
   (-240+4:2) arc (-240+4:-360:2);
 \end{scope}

 \begin{scope}[rotate={(+0)}]
 \draw[white, line width=2.8pt]
   (90+60:-4) to (90+60:-1);
 \draw[->, -{>[scale=1.24]}, >=Latex, line width=1.2pt]
   (90+60:2) to (90+60:.7);
 \draw[->, -{>[scale=1.24]}, >=Latex, line width=1.2pt]
   (90+60:1.24) to (90+60:-2.8);
 \draw[->, -{>[scale=1.24]}, >=Latex, line width=1.2pt]
   (90+60:-2.4) to (90+60:-4);
 \end{scope}

 \begin{scope}[rotate={(+120)}]
 \draw[white, line width=2.8pt]
   (90+60:-4) to (90+60:-1);
 \draw[->, -{>[scale=1.24]}, >=Latex, line width=1.2pt]
   (90+60:2) to (90+60:.7);
 \draw[->, -{>[scale=1.24]}, >=Latex, line width=1.2pt]
   (90+60:1.24) to (90+60:-2.8);
 \draw[->, -{>[scale=1.24]}, >=Latex, line width=1.2pt]
   (90+60:-2.4) to (90+60:-4);
 \end{scope}

 \begin{scope}[rotate={(-120)}]
 \draw[white, line width=2.8pt]
   (90+60:-4) to (90+60:-1);
 \draw[->, -{>[scale=1.24]}, >=Latex, line width=1.2pt]
   (90+60:2) to (90+60:.7);
 \draw[->, -{>[scale=1.24]}, >=Latex, line width=1.2pt]
   (90+60:1.24) to (90+60:-2.8);
 \draw[->, -{>[scale=1.24]}, >=Latex, line width=1.2pt]
   (90+60:-2.4) to (90+60:-4);
 \end{scope}

 \draw[->, -{>[scale=1.24]}, >=Latex, line width=1.1pt]
   (3.45,-2) to ({3.45/2},-2);
 \draw[->, -{>[scale=1.24]}, >=Latex, line width=1.1pt]
   ({3.45/2+.4},-2) to ({-3.45/2-.2},-2);
 \draw[line width=1.1pt]
   ({-3.45/2+.4},-2) to (-3.45,-2);

 \begin{scope}[rotate=(+120)]
 \draw[->, -{>[scale=1.24]}, >=Latex, line width=1.1pt]
   (3.45,-2) to ({3.45/2},-2);
 \draw[->, -{>[scale=1.24]}, >=Latex, line width=1.1pt]
   ({3.45/2+.4},-2) to ({-3.45/2-.2},-2);
 \draw[line width=1.1pt]
   ({-3.45/2+.4},-2) to (-3.45,-2);
 \end{scope}

 \begin{scope}[rotate=(-120)]
 \draw[->, -{>[scale=1.24]}, >=Latex, line width=1.1pt]
   (3.45,-2) to ({3.45/2},-2);
 \draw[->, -{>[scale=1.24]}, >=Latex, line width=1.1pt]
   ({3.45/2+.4},-2) to ({-3.45/2-.2},-2);
 \draw[line width=1.1pt]
   ({-3.45/2+.4},-2) to (-3.45,-2);
 \end{scope}

 \draw[fill=white] (90+60:2) circle (.6);
 \draw[fill=white] (90-60:2) circle (.6);
 \draw[fill=white] (90-60-120:2) circle (.6);

 \draw[fill=white] (90-60-120:0) circle (.6);

 \draw[fill=white] (90+60:-4) circle (.6);
 \draw[fill=white] (90-60:-4) circle (.6);
 \draw[fill=white] (90-60-120:-4) circle (.6);

 \draw (90+180:2) node
   {
     $
       {\rm e}_1
     $
   };
 \begin{scope}[shift={(0,-.36)}]
 \draw (90+180:2) node
   {
     \scalebox{.75}{
     $
       = {\rm i}
     $
     }
   };
 \end{scope}

 \draw (90+60:2) node
   {
     $
       {\rm e}_2
     $
   };
 \begin{scope}[shift={(0,-.36)}]
 \draw (90+60:2) node
   {
     \scalebox{.75}{
     $
       = {\rm j}
     $
     }
   };
 \end{scope}

 \draw (90-60:2) node
   {
     $
       {\rm e}_3
     $
   };
 \begin{scope}[shift={(0,-.36)}]
 \draw (90-60:2) node
   {
     \scalebox{.75}{
     $
       = {\rm k}
     $
     }
   };
 \end{scope}

 \draw (90-60:0) node
   {
     $
       {\rm e}_4
     $
   };
 \begin{scope}[shift={(0,-.36)}]
 \draw (90-60:0) node
   {
     \scalebox{.75}{
     $
       = \ell
     $
     }
   };
 \end{scope}

 \draw (90+60:-4) node
   {
     $
       {\rm e}_6
     $
   };
 \begin{scope}[shift={(0,-.36)}]
 \draw (90+60:-4) node
   {
     \scalebox{.75}{
     $
       = {\rm j} \ell
     $
     }
   };
 \end{scope}

 \draw (90-60:-4) node
   {
     $
       {\rm e}_7
     $
   };
 \begin{scope}[shift={(0,-.36)}]
 \draw (90-60:-4) node
   {
     \scalebox{.75}{
     $
       = {\rm k} \ell
     $
     }
   };
 \end{scope}

 \draw (90+180:-4) node
   {
     $
       e_5
     $
   };
 \begin{scope}[shift={(0,-.36)}]
 \draw (90+180:-4) node
   {
     \scalebox{.75}{
     $
       = {\rm i} \ell
     $
     }
   };
 \end{scope}

 \draw[fill=green, draw opacity=0, fill opacity=.15]
   (90+60:2) circle (.6);
 \draw[fill=green, draw opacity=0, fill opacity=.15]
   (90-60:2) circle (.6);
 \draw[fill=lightgray, draw opacity=0, fill opacity=.2]
   (90-60-120:2) circle (.6);
 \draw[fill=green, draw opacity=0, fill opacity=.15]
   (90-60-120:2) circle (.6);

 \draw[fill=cyan, draw opacity=0, fill opacity=.15]
   (90-60-120:0) circle (.6);

 \draw[fill=cyan, draw opacity=0, fill opacity=.15]
   (90+60:-4) circle (.6);
 \draw[fill=cyan, draw opacity=0, fill opacity=.15]
   (90-60:-4) circle (.6);
 \draw[fill=cyan, draw opacity=0, fill opacity=.15]
   (90-60-120:-4) circle (.6);

\end{tikzpicture}
}
}

\end{tabular}
\begin{proof}
  Notice that the relations \eqref{qlp} imply the following
  further relations
  \begin{equation}
    \label{InducedRelations}
    \begin{array}{rclccrcl}
    q \ell &=& \ell q^\ast
    &&
    \ell q &=& q^\ast \ell
    \\
    (\ell q) p &=& \ell (p q)
    &&
    q (p \ell) &=& (p q) \ell
    \\
    (\ell q) (\ell p)
    &=&
    - p q^\ast
    &&
    (q \ell) (p \ell)
    &=&
    - p^\ast  q
    \\
    \\
    \mathllap{
      \mbox{hence in particular:}
      \;\;\;\;\;\;\;\;\;\;\;\;\;\;\;\;\;\;
    }
    \ell (\ell p) &=& -p
    &&
    (q \ell) \ell &=& - q
    \\
    (\ell q) \ell & = & - q^\ast
    &&
    \ell (p \ell) & = & - p^\ast.
    \end{array}
  \end{equation}
  Consequently, one finds the general formula for the product of
  any pair of octonions $x_i$,
  parametrized as
  $x_i = q_i + p_i \ell$
  or as
  $x_i = q_i + \ell p_i$ (with $q_i, p_i \in \mathbb{H}$),
  to be,
  respectively:
  \begin{equation}
    \begin{aligned}
    (q_1 + p_1 \ell)
    (q_2 + p_2 \ell)
    & = \;
    \big(
      q_1 q_2
      -
      p_2^\ast p_1
    \big)
    \;+\;
    \phantom{\ell}
    \big(
      p_2 q_1
      +
      p_1 q_2^\ast
    \big)
    \ell,
    \\
    (q_1 + \ell p_1)
    (q_2 + \ell p_2)
    & =\;
    \big(
      q_1 q_2 - p_2 p_1^\ast
    \big)
    \;+\;
    \ell
    \big(
      q_1^\ast p_2
      +
      q_2 p_1
    \big)
    \phantom{\ell}.
    \end{aligned}
  \end{equation}
  This is the formula for the octonionic product according to
  \cite[(6)]{Dickson1919} (where the version in the first line appears),
  reviewed in \cite[2.2]{Baez02} (where the isomorphic version in the
  second line is given).
\hfill \end{proof}

\begin{remark}[Division algebra inclusions and Supersymmetry breaking patterns]
  Any choice of octonion generators $\mathrm{e}_1, \cdots \mathrm{e}_7$
  according to Lemma \ref{Octonions} induces algebra inclusions
  of, consecutively, the real numbers, the complex numbers
  and the quaternions into the octonions
  \begin{equation}
    \label{AlgebraSequence}
    \xymatrix@R=2pt{
      \mathbb{O}
    \;  \ar@{<-^{)}}[rr]
      \ar@{}[dd]|-{
        \scalebox{.9}{
          \begin{rotate}{-90}
            $\mathclap{\simeq_{{\scalebox{.6}{$\mathbb{R}$}}}}$
          \end{rotate}
        }
      }
      &&
      \;\; \mathbb{H}
     \; \ar@{<-^{)}}[rr]
      \ar@{}[dd]|-{
        \scalebox{.9}{
          \begin{rotate}{-90}
            $\mathclap{\simeq_{{\scalebox{.6}{$\mathbb{R}$}}}}$
          \end{rotate}
        }
      }
      &&
      \;\; \mathbb{C}
     \; \ar@{}[dd]|-{
        \scalebox{.9}{
          \begin{rotate}{-90}
            $\mathclap{\simeq_{{\scalebox{.6}{$\mathbb{R}$}}}}$
          \end{rotate}
        }
      }
      \ar@{<-^{)}}[rr]
      \ar@{}[dd]|-{
        \scalebox{.9}{
          \begin{rotate}{-90}
            $\mathclap{\simeq_{{\scalebox{.6}{$\mathbb{R}$}}}}$
          \end{rotate}
        }
      }
      &&
    \;\;  \mathbb{R}
     \; \ar@{}[dd]|-{
        \scalebox{.9}{
          \begin{rotate}{-90}
            $\mathclap{\simeq_{{\scalebox{.6}{$\mathbb{R}$}}}}$
          \end{rotate}
        }
      }
      \\
      \\
      \mathclap{
      \left\langle
        \!\!\!\!\!
        \begin{array}{l}
          1, \mathrm{e}_1, \mathrm{e}_2, \mathrm{e}_3
          \\
          \mathrm{e}_4, \mathrm{e}_5, \mathrm{e}_6, \mathrm{e}_7
        \end{array}
        \!\!\!\!\!
      \right\rangle
      }
      &&
      \mathclap{
      \langle
        1, \mathrm{e}_1, \mathrm{e}_2, \mathrm{e}_3
      \rangle
      }
      &&
      \mathclap{
      \langle
        1, \mathrm{e}_1
      \rangle
      }
      &&
      \mathclap{
      \langle
        1
      \rangle
      }
    }
  \end{equation}
  Moreover, multiplication of the linear sub-spaces corresponding to
  these sub-algebras with the remaining generators
  induces distinguished linear isomorphisms
  \begin{equation}
    \label{LinearSubspacesOfTheOctonions}
    \xymatrix{
      \mathbb{O}
      \;\simeq_{\mathbb{R}}\;
      \mathbb{H} \oplus \mathbb{H} \ell
      \;\simeq_{\mathbb{R}}\;
      \big(
        \mathbb{C} \oplus \mathbb{C} \mathrm{j}
      \big)
      \oplus
      \big(
        \mathbb{C} \oplus \mathbb{C} \mathrm{j}
      \big)
      \ell
      \;\simeq_{\mathbb{R}}\;
      \big(
        (\mathbb{R} \oplus \mathbb{R} \mathrm{i} )
        \oplus
        (\mathbb{R} \oplus \mathbb{R} \mathrm{i}) \mathrm{j}
      \big)
      \oplus
      \big(
        (\mathbb{R} \oplus \mathbb{R} \mathrm{i})
        \oplus
        (\mathbb{R} \oplus \mathbb{R} \mathrm{i})
        j
      \big)
      \ell
    }.
  \end{equation}
  This extra structure on $\mathbb{O}$, corresponding
  to the choice of an adapted linear basis according to
  Lemma \ref{Octonions}, turns out to reflect
  the supersymmetry breaking sequences
  \begin{equation}
  \xymatrix@R=3pt{
    \mathbb{R}^{10,1\vert\mathbf{32}}\;
    \ar@{|->}[dr]
    \ar@{|->}[rr]
    &&
    \;\mathbb{R}^{6,1\vert \mathbf{16}}\;
    \ar@{|->}[rr]
    \ar@{|->}[dr]
    &&
    \;\mathbb{R}^{4,1\vert \mathbf{8}}\;
    \ar@{|->}[dr]
    \\
    &
    \;
    \mathbb{R}^{9,1\vert \mathbf{16}}
    \;
    \ar@{|->}[rr]
    &&
    \;\mathbb{R}^{5,1\vert \mathbf{8}}\;
    \ar@{|->}[rr]
    &&
    \;\mathbb{R}^{3,1\vert \mathbf{4}}\;
    \ar@{|->}[rr]
    &&
    \;\mathbb{R}^{2,1\vert \mathbf{2}}
  }
  \end{equation}
  This is the statement of Prop. \ref{SpinRepsViaSL2O}
  and Prop. \ref{RealSpinRepsInDimPlus1ViaOconions} below.
\end{remark}

An illustrative example computation with the relations \eqref{InducedRelations}
is the following (used below in Prop. \ref{FixedSubspaceOfGamma6789}):

\begin{lemma}[Reversal of sign of $\ell$-component by left action]
\label{ConsecutiveLeftMultiplicationBye4e5e6e7}
The action of consecutive left multiplication by the
generators $\mathrm{e}_4$, $\mathrm{e}_5$, $\mathrm{e}_6$, $\mathrm{e}_7$
from Lemma \ref{Octonions} on any octonion
$x = q + p \ell$ ($q,p \in \mathbb{H}$) is by reversal of the
sign of the $\ell$-component:
\begin{equation}
  \label{SignReversalByLeftMultiplication}
  {\rm e}_4
  \bigg(
  {\rm e}_5
  \Big(
    {\rm e}_6
    \big(
      {\rm e}_7
      \;
      (q + p \ell)
      \;
    \big)
  \Big)
  \bigg)
  \;=\;
  q - p \ell
  \,.
\end{equation}
\end{lemma}
\begin{proof}
Using the relations \eqref{qlp} and \eqref{InducedRelations} and the
associativity of the multiplication on quaternions, we compute
as follows:

\vspace{-6mm}
$$
  \begin{aligned}
    {\rm e}_4
    \Big(
    {\rm e}_5
    \big(
      {\rm e}_6
      ({\rm e}_7 x)
    \big)
    \Big)
    & =
    \ell
    \bigg(
    ({\rm i} \ell)
    \Big(
      ({\rm j} \ell)
      \big(
        ({\rm k} \ell)
        x
      \big)
    \Big)
    \bigg)
    \;=\;
    \ell
    \bigg(
    ({\rm i} \ell)
    \Big(
      ({\rm j} \ell)
      \big(
        ({\rm k} x^\ast)
        \ell
      \big)
    \Big)
    \bigg)
    \;=\;
    \ell
    \Big(
      ({\rm i} \ell)
      \big(
        (x {\rm k}) {\rm j}
      \big)
    \Big)
    \;=\;
    \ell
    \bigg(
      \Big(
        {\rm i}
        \big(
          {\rm j}
          ({\rm k} x^\ast)
        \big)
      \Big)
    \ell
    \bigg)
    \\
    & =
    \big(
      ( x {\rm k})
      {\rm j}
    \big)
   {\rm i}
   \;=\;
    \big(
          ( (q + p \ell)
      {\rm k})
      {\rm j}
    \big)
    {\rm i}
    \;=\;
    q {\rm k j i}
    +
    \Big(
      \big(
        (p \ell) {\rm k}
      \big)
      {\rm j}
    \Big)
    {\rm i}
    \;=\;
    q
      {\rm k j i}
    -
    (p
          {\rm k j i}
    ) \ell
    \;=\;
    q - p \ell
    \,.
  \end{aligned}
$$

\vspace{-7mm}
\hfill \end{proof}

\begin{notation}[Conjugation]
  \label{Conjugation}
In what follows, we will denote conjugation by $(-) ^\ast$
in any of the real $\ast$-algebras $\mathbb{R},\mathbb{C},\mathbb{H}$, and $\mathbb{O}$.
For a matrix $A$ with coefficients in
$\mathbb{K}\in\{\mathbb{R},\mathbb{C},\mathbb{H},\mathbb{O}\}$
the component-wise conjugate matrix will be denoted by $A^\ast$,
while the Hermitian conjugate matrix will be denoted
by $A^\dagger :=  {}^tA^*$.
\end{notation}
We record the following immediate generalization of the
standard Pauli matrix construction:

\begin{lemma}[$\mathbb{K}$-Pauli matrices]
  \label{OctonionicPauliMatrices}
  Let $\mathbb{K}\in \{\mathbb{R},\mathbb{C},\mathbb{H},\mathbb{O}\}$.
  There is a linear isomorphism of real vector spaces
  equipped with quadratic forms (the color code follows the configurations in \eqref{TheBranes})

  \vspace{-.4cm}

  \begin{equation}
    \label{ExOctopnionicPauliMatices}
    \hspace{-2cm}
    \raisebox{-140pt}{
    \scalebox{.9}{\begin{tikzpicture}

    \draw (0,0) node
    {
    \xymatrix@R=-4pt{
      \big(
        \mathbb{R}^{\dim_{\mathbb{R}}\mathbb{K}+1,1}, \eta
      \big)
      \ar[rr]_-{\simeq}^-\sigma
      &&
      \big(
        \mathrm{Mat}(2,\mathbb{K})^{\mathrm{Herm}},
        - \mathrm{det}
      \big)
      \\
      v
      =
      \left[
        {\begin{array}{c}
          v^0
          \\
          v^1
          \\
          v^2
          \\
          v^3
          \\
          v^4
          \\
          v^5
          \\
          v^6
          \\
          v^7
          \\
          v^8
          \\
          v^{9}
        \end{array}}
      \right]
      \ar@{}[rr]|-{\longmapsto}
      &&
      \!\!\!\!\!\!\!
      \left[
        \!\!
        {\begin{array}{cc}
          v^0 + v^1
            &
          v^2
          \\
          v^2
          &
          v^0 - v^1
        \end{array}}
        \!\!
      \right]
      +
      \mathrlap{
      \big(
         v^3 \mathrm{e}_1
        \;+\;
        v^4 \mathrm{e}_2
        +
        v^5 \mathrm{e}_3
        \;+\;
          v^6{\rm e}_{4}
          +
          v^{7}{\rm e}_{5}
          +
           v^8{\rm e}_{6}
          +
          v^{9}{\rm e}_{7}
      \big)
      \left[
        \!\!
        {\begin{array}{cc}
          0 & 1
          \\
          -1
          &
          0
        \end{array}}
        \!\!
      \right]
      }
    }
    };

  \begin{scope}[shift={(4.55,.16)}]
    \draw[fill=red, draw opacity=0, fill opacity=.1]
      (-4.6,.15) rectangle (-1.0,-1.02);
    \draw[fill=lightgray, draw opacity=0, fill opacity=.2]
      (-4.6,.15) rectangle (0.34,-1.02);
    \draw[fill=green, draw opacity=0, fill opacity=.15]
      (-4.6,.15) rectangle (2.6,-1.02);
    \draw[fill=cyan, draw opacity=0, fill opacity=.15]
      (2.6,.15) rectangle (6.9,-1.02);
    \draw[fill=red, draw opacity=0, fill opacity=.1]
      (-7.4,.6) rectangle (-6.65,1.92);
    \draw[fill=lightgray, draw opacity=0, fill opacity=.2]
      (-7.4,0.1) rectangle (-6.65,1.92);
    \draw[fill=green, draw opacity=0, fill opacity=.15]
      (-7.4,-.87) rectangle (-6.65,1.92);
    \draw[fill=cyan, draw opacity=0, fill opacity=.15]
      (-7.4,-2.7) rectangle (-6.65,-.87);
  \end{scope}
    \end{tikzpicture}
    }
    }
  \end{equation}
  from $(\dim_{{}_{\mathbb{R}}}\mathbb{K}+2)$-dimensional Minkowski spacetime with
  quadratic form being its Minkowski metric \break $\eta = \mathrm{diag}(-1,1,1,\cdots, 1)$
  to the vector space of $2 \times 2$ $\mathbb{K}$-Hermitian matrices
  with quadratic form being minus the determinant operation.
\end{lemma}

We denote the $\mathbb{K}$-Pauli matrices \eqref{ExOctopnionicPauliMatices}
corresponding to the coordinate basis elements as follows:
\begin{equation}
  \label{CoordinatePauliMatrices}
    \sigma_a \;:=\; \sigma(v_a)
    \qquad
    \text{and}
    \qquad
    \sigma^a \;:=\; \eta^{a b} \sigma_b,
\end{equation}
where $v_a \in \mathbb{R}^{\dim_{\mathbb{R}\mathbb{K}+1,1}}$
denotes the vector with components $(v_a)^b := \delta^b_a$.

\medskip
The following observation is due to \cite{KugoTownsend82},
with a streamlined review in \cite{BH09}:

\begin{prop}[Real Spin representations in dimension 10, 6, 4, 3,
via $\mathbb{K}$-Pauli matrices]
  \label{SpinRepsViaSL2O}
   Let $\mathbb{K}\in \{\mathbb{R},\mathbb{C},\mathbb{H},\mathbb{O}\}$.
   \vspace{-1mm}
   \item {\bf (i)}   The assignments
   \begin{equation}
     \label{KPauliActions}
     {\pmb\Gamma}_{a_1}{\pmb\Gamma}_{a_2}
     \;\longmapsto\;
     \sigma^{a_1} \cdot \big( \sigma_{a_2} \cdot (-) \big)
     \phantom{AAAA}
     \mbox{and}
     \phantom{AAAA}
     {\pmb\Gamma}_{a_1}{\pmb\Gamma}_{a_2}
     \;\longmapsto\;
     \sigma_{a_1} \cdot \big( \sigma^{a_2} \cdot (-) \big)
   \end{equation}
   of left multiplication actions by the $\mathbb{K}$-Pauli matrices \eqref{CoordinatePauliMatrices}
   for alternating index positions
   constitute representations of the even Clifford algebras \eqref{EvenCliffordAlgebra} on the real vector space
   underlying $\mathbb{K}^2$, hence are real algebra homomorphisms
   $$
     \xymatrix{
       \mathrm{Cl}^{\mathrm{ev}}_{\mathbb{R}}
       \big(
         \mathrm{dim}_{{}_{\mathbb{R}}} \mathbb{K} + 1, 1
       \big)
       \ar[rr]
       &&
       \mathrm{End}_{{}_{\mathbb{R}}}\big( \mathbb{K}^2 \big)
     }
     \,.
   $$

   \vspace{-5mm}
\item {\bf (ii)}   As a consequence \eqref{LieSpinRepresentation}
   there are two real spinorial representations
   $\mathbf{2\, \mathbf{dim}_{{}_{\mathbb{R}}}\mathbb{K}}$ and $\overline{\mathbf{2}\, \mathbf{dim}_{{}_{\mathbb{R}}}\mathbb{K}}$ of
   $\mathrm{Spin}(\dim_{\mathbb{R}}\mathbb{K}+1,1)$, regarded as real modules of the
   Lie algebra $\mathfrak{so}(\dim_{\mathbb{R}}\mathbb{K}+1,1)$,
   each isomorphic to
   the real vector space underlying
   $\mathbb{K}^2$
   equipped, respectively, with the following action of the standard basis elements $R_{a_1 a_2}$ \eqref{LieBracketInSO}:
  \begin{equation}
    \label{HeteroticSpinRepsAsOctonionicModules}
    \hspace{-.4cm}
    \raisebox{15pt}{
    \scalebox{.9}{
    \xymatrix@R=5pt@C=3pt{
      \mathfrak{so}(\dim_{{}_{\mathbb{R}}}\mathbb{K}+1,1)
      \ar@{}[r]|-{\times}
      \ar@{}[dd]|-{
        \scalebox{.9}{
          \begin{rotate}{-90}
            $\mathclap{=}$
          \end{rotate}
        }
      }
      &
     \mathbf{2\, \mathbf{dim}_{{}_{\mathbb{R}}}\mathbb{K}}
           \ar@{}[dd]|-{
        \scalebox{.9}{
          \begin{rotate}{-90}
            $\mathclap{\simeq_{{\scalebox{.6}{$\mathbb{R}$}}}}$
          \end{rotate}
        }
      }
      \ar[rr]
      &{\phantom{AA}}&
      \mathbf{2\, \mathbf{dim}_{{}_{\mathbb{R}}}\mathbb{K}}
      \ar@{}[dd]|-{
        \scalebox{.9}{
          \begin{rotate}{-90}
            $\mathclap{\simeq_{{\scalebox{.6}{$\mathbb{R}$}}}}$
          \end{rotate}
        }
      }
      \\
      \\
      \mathfrak{so}(\dim_{{}_{\mathbb{R}}}\mathbb{K}+1,1)
      \ar@{}[r]|<<<{\times}
      &
      \mathbb{K}^{\mathrlap{2}}
      \;\ar[rr]
      &&
      \mathbb{K}^{\mathrlap{2}}
      \\
    {\phantom{mmmmmmmmm}}
     \mathllap{(}
        R_{a_1 a_2},
        &
        \psi
      \mathrlap{)}
      \ar@{}[rr]|-{\;\longmapsto}
      &&
     \tfrac{1}{2}  \sigma^{a_1} \cdot (\sigma_{a_2} \cdot \psi)
    }
    \;
    \xymatrix@R=5pt@C=3pt{
      \mathfrak{so}(\dim_{{}_{\mathbb{R}}}\mathbb{K}+1,1)
      \ar@{}[r]|-{\times}
      \ar@{}[dd]|-{
        \scalebox{.9}{
          \begin{rotate}{-90}
            $\mathclap{=}$
          \end{rotate}
        }
      }
      &
      \overline{\mathbf{2}\, \mathbf{dim}_{{}_{\mathbb{R}}}\mathbb{K}}
      \ar@{}[dd]|-{
        \scalebox{.9}{
          \begin{rotate}{-90}
            $\mathclap{\simeq_{{\scalebox{.6}{$\mathbb{R}$}}}}$
          \end{rotate}
        }
      }
      \ar[rr]
      &{\phantom{AA}}&
     \overline{\mathbf{2}\, \mathbf{dim}_{{}_{\mathbb{R}}}\mathbb{K}}
      \ar@{}[dd]|-{
        \scalebox{.9}{
          \begin{rotate}{-90}
            $\mathclap{\simeq_{{\scalebox{.6}{$\mathbb{R}$}}}}$
          \end{rotate}
        }
      }
      \\
      \\
      \mathfrak{so}(\dim_{{}_{\mathbb{R}}}\mathbb{K}+1,1)
      \ar@{}[r]|<<<{\times}
     &
      \mathbb{K}^{\mathrlap{2}}
      \;\ar[rr]
      &&
      \mathbb{K}^{\mathrlap{2}}
      \\
    {\phantom{mmmmmmmmm}}
    \mathllap{(}
        R_{a_1 a_2},
        &
        \psi
      \mathrlap{)}
      \ar@{}[rr]|-{\;\longmapsto}
      &&
      \tfrac{1}{2}\sigma_{a_1} \cdot (\sigma^{a_2} \cdot \psi)
    }
    }
    }
  \end{equation}
\end{prop}

\begin{remark}[Complex Weyl representations]
\label{ComplexWeylSpinorRepresentations}
For $\mathbb{K}=\mathbb{C}$, the action by $\mathbb{C}$-Pauli matrices in \eqref{HeteroticSpinRepsAsOctonionicModules} is clearly $\mathbb{C}$-linear, so the  $\mathrm{Spin}(3,1)$-representatons $\mathbf{4}$ and $\overline{\mathbf{4}}$ are the real representations underlying a pair of complex representations. It is manifest from \eqref{HeteroticSpinRepsAsOctonionicModules} for $\mathbb{K}=\mathbb{C}$ that these two complex representations are the
  standard complex 2-dimensional Weyl Spin representations,
  which we denote by
  \begin{equation}
    \label{ComplexWeylRepresentations}
    \overset{
      \mathclap{
      \mbox{
        \tiny
        \color{darkblue}
        \bf
        \begin{tabular}{c}
          left/right
          \\
          Weyl Spin representations
        \end{tabular}
      }
      }
    }{
      \mathbf{2}_{\mathbb{C}},
      \;\;
      \overline{\mathbf{2}}_{\mathbb{C}}
    }
    \;\;\;\;\;\;\;\;
    \in
    \;\;\;
    \overset{
      \mathclap{
      \mbox{
        \tiny
      \bf  \color{darkblue} \bf
        \begin{tabular}{c}
          complex
          \\
          representation ring
        \end{tabular}
      }
      }
    }{
      \mathrm{Rep}_{\mathbb{C}}
      \big(
        \mathrm{Spin}(3,1)
      \big)
    }
    \xymatrix{\ar[r]&}
    \overset{
      \mathclap{
      \mbox{
        \tiny
      \bf  \color{darkblue} \bf
        \begin{tabular}{c}
          real
          \\
          representation ring
        \end{tabular}
      }
      }
    }{
      \mathrm{Rep}_{\mathbb{R}}
      \big(
        \mathrm{Spin}(3,1)
      \big).
    }
  \end{equation}
\end{remark}

\begin{lemma}[Isomorphism of real spinor irreps of $\mathrm{Spin}(3,1)$]
  \label{RealReps4AndBar4OfSpin31AreIsomorphic}
  The representations $\mathbf{4}$ and $\overline{\mathbf{4}}$
   of $\mathrm{Spin}(3,1)$ obtained from \eqref{HeteroticSpinRepsAsOctonionicModules} for $\mathbb{K}=\mathbb{C}$
  are isomorphic as \emph{real} representations
  (not as complex representations):
  $$
    \mathbf{4}
    \;\simeq\;
    \overline{\mathbf{4}}
    \;\;
    \in
    \mathrm{Rep}_{\color{magenta}\mathbb{R}}
    \big(
      \mathrm{Spin}(3,1)
    \big).
  $$
\end{lemma}
\begin{proof}
  Write
  $
    \epsilon
    \;:=\;
    \footnotesize
    \begin{bmatrix}
      0 & 1
      \\
      -1 & 0
    \end{bmatrix}
    \!\!\!
    \;\;\in
    \mathrm{Mat}(2,\mathbb{C})\;.
  $
  We claim that the $\mathbb{R}$-linear isomorphism of $\mathbb{R}$-vector spaces

  \vspace{-2mm}
  \begin{equation}
    \label{IsomorphismBetween4AndBar4AsRealReps}
    \xymatrix@R=-2pt{
     \mathbb{C}^2  \simeq_{{\scalebox{.6}{$\mathbb{R}$}}} \mathbf{4}
           \ar[rr]_-{\simeq}^-{\phi}
      &&
      \overline{\mathbf{4}}
    \simeq_{{\scalebox{.6}{$\mathbb{R}$}}} \mathbb{C}^2
           \\
      \psi
      \ar@{|->}[rr]
      &&
      \epsilon \cdot \psi^\ast
    }
  \end{equation}
%
%
 \noindent  is an isomorphism of real representations of $\mathrm{Spin}(3,1)$.
  To see this, use that
  $
    \epsilon
    \cdot
     \footnotesize
       \begin{bmatrix}
      a & b
      \\
      c & d
    \end{bmatrix}
     \cdot
    \epsilon^{-1}
    \;=\;
        \begin{bmatrix}
      d & -c
      \\
      -b & a
    \end{bmatrix}
    $
  and observe that therefore, by \eqref{CoordinatePauliMatrices},
  we have
  $
    \epsilon \cdot (\sigma_\mu)^\ast \cdot \epsilon^{-1}
    =
    - \sigma^\mu.
  $
 As complex conjugation $(-)^\ast\colon \mathbb{C}\to \mathbb{C}$ is
 an $\mathbb{R}$-algebra homomorphism
 we have $(A \cdot B)^\ast = A^\ast \!\cdot B^\ast$ for complex matrices
 (following Notation \ref{Conjugation}).
 Therefore:
  \begin{equation}
    \label{TowardsProvingIsomorphismBetween4andBar4}
    \begin{aligned}
      \epsilon \cdot
      \big(
        \sigma^{\mu_1} \cdot \sigma_{\mu_2}
      \big)^\ast
      \cdot
      \epsilon^{-1}
      & =
      \big(
      \epsilon
      \cdot
      \big(
        \sigma^{\mu_1}
      \big)^\ast
      \cdot
      \epsilon^{-1}
      \big)
      \cdot
      \big(
      \epsilon
      \cdot
      \big(
        \sigma_{\mu_2}
      \big)^\ast
      \cdot
      \epsilon^{-1}
      \big)
      \\
      & =\;
      (-\sigma_{\mu_1}) \cdot (-\sigma^{\mu_2})
      \;=\;
      \sigma_{\mu_1} \cdot \sigma^{\mu_2}
      \,.
    \end{aligned}
  \end{equation}
  With this the claim follows:
  \begin{equation}
    \label{FurtherTowardsProvingIsomorphismBetween4andBar4}
    \begin{aligned}
    \phi
    \left(\tfrac{1}{2}
      \sigma^{\mu_1}
      \cdot
      \sigma_{\mu_2}
       \cdot
       \psi
    \right)
    & =
    \epsilon \cdot
    \left(\tfrac{1}{2}
      \sigma^{\mu_1}
      \cdot
      \sigma_{\mu_2}
      \cdot
      \psi
    \right)^\ast
    \;=\;
      \Big(
      \epsilon
      \cdot
      \big(\tfrac{1}{2}
      \sigma^{\mu_1}
      \cdot
      \sigma_{\mu_2}
      \big)^\ast
      \cdot
      \epsilon^{-1}
      \Big)
      \cdot
      \big(
      \epsilon
      \cdot
      \psi^\ast
      \big)
      \\
      & =\tfrac{1}{2}
      \sigma_{\mu_1}
      \cdot
      \sigma^{\mu_2}
      \cdot
      \phi(\psi)\;.
    \end{aligned}
  \end{equation}

  \vspace{-6mm}
\hfill
\end{proof}

\medskip
\begin{remark}[$\mathbb{K}$-Weyl spinors beyond the complex case]
We highlight the following subtle points:
  \vspace{-1mm}
\item {\bf (i)}  The reason that the proof of Prop. \ref{RealReps4AndBar4OfSpin31AreIsomorphic}
  does not identify the two Weyl Spin-representations $\mathbf{2}_{\mathbb{C}},
    \overline{\mathbf{2}}_{\mathbb{C}}$ \eqref{ComplexWeylRepresentations}
  when regarded as \emph{complex} representations
  is due to the complex conjugation on the right
  of \eqref{IsomorphismBetween4AndBar4AsRealReps},
  which makes $\phi$ a real-linear map, but not a complex-linear map.

  \vspace{-1mm}
\item {\bf(ii)}   The reason that
  the proof of Prop. \ref{RealReps4AndBar4OfSpin31AreIsomorphic}
  does not generalize to identify also the two
  real representations $\mathbf{8}$ and $\overline{\mathbf{8}}$
   of $\mathrm{Spin}(5,1)$,
  nor the two real representations $\mathbf{16}$
  and $\overline{\mathbf{16}}$
  of $\mathrm{Spin}(9,1)$ given by \eqref{HeteroticSpinRepsAsOctonionicModules}
  for $\mathbb{K}=\mathbb{H}$ and $\mathbb{K}=\mathbb{O}$, respectively, is
  that in these cases, due to the non-commutativity of the
  quaternions and of the octonions,
  equations \eqref{TowardsProvingIsomorphismBetween4andBar4}
  and \eqref{FurtherTowardsProvingIsomorphismBetween4andBar4}
  do not hold, as quaternionic and  octonion conjugation is not an $\mathbb{R}$-algebra
   homomorphism but an anti-homomorphism: $(x y)^\ast = y^\ast x^\ast$.
\end{remark}

\medskip

Next we record the following immediate generalization of the
Dirac-matrix construction:

\begin{lemma}[$\mathbb{K}$-Dirac matrices]
  \label{OctonionicDiracMatrices}
  Let $\mathbb{K} \in \{\mathbb{R}, \mathbb{C}, \mathbb{H}, \mathbb{O}\}$.
  There is a linear isomorphism of real vector spaces equipped with
  quadratic forms
  \begin{equation}
    \label{OctonionicDiracMatrices}
    \hspace{-1cm}
    \raisebox{-60pt}{
\begin{tikzpicture}

  \draw (0,0) node
  {
  \xymatrix@R=-2pt{
    \big(
      \mathbb{R}^{\dim_{{}_{\mathbb{R}}}\mathbb{K}+2,1},
      \eta
    \big)
    \ar[rr]_-\simeq
    &&
    \big(
      \mathrm{Mat}(2,\mathbb{K})^{\mathrm{Herm}}
      \times \mathbb{R},
      - \mathrm{det}(-) + (-)^2
    \big)
    \; \ar@{^{(}->}[r]
    &
    \mathrm{Mat}(4,\mathbb{K})^{\mathrm{Herm}},
    \\
    \left[
    \!\!\!
    {\begin{array}{c}
      v^0
      \\
      \vdots
      \\
      v^9
      \\
      v^{5'}
    \end{array}}
    \!\!\!
    \right]
    \ar@{|->}[rr]
    &&
    \mathrlap{
    \underoverset{a=0}{9}{\sum}
    v^a
    \,
    \left[
    \!\!\!
    {\begin{array}{cc}
      0 & \sigma_a
      \\
      \sigma^a & 0
    \end{array}}
    \!\!\!
    \right]
    \;+\;
    v^{5'}
    \left[
    \!\!\!
    {\begin{array}{cc}
      1 & 0
      \\
      0 & -1
    \end{array}}
    \!\!\!
    \right]
    }
  }
  };

  \begin{scope}[shift={(0,.14)}]

  \shadedraw[draw opacity=0, fill opacity=.4,
     top color=greenii, bottom color=darkblue]
    (-5.7-.25,.65) rectangle (-5.7+.25,-.96);

  \draw[draw opacity=0, fill opacity=.4,
     fill=darkyellow]
    (-5.7-.25,-.97) rectangle (-5.7+.25,-1.5);

  \begin{scope}[shift={(.9,0)}]

  \shadedraw[draw opacity=0, fill opacity=.4,
     left color=greenii, right color=darkblue]
    (-.3,-.45+.5) rectangle (.6,-.45-.5);

  \draw[draw opacity=0, fill opacity=.4,
     fill=darkyellow]
    (2.8,-.45+.5) rectangle (3.35,-.45-.5);

  \end{scope}

  \end{scope}

\end{tikzpicture}
    }
  \end{equation}
  from $(\dim_{{}_{\mathbb{R}}}\mathbb{K}+3)$-dimensional Minkowski spacetime with
  quadratic form being its Minkowski metric \break $\eta = \mathrm{diag}(-1,1,1,\cdots, 1)$
  to the vector space of $\mathbb{K}$-matrices, as shown,
  where $\sigma_a$ and $\sigma^a$ are the $\mathbb{K}$-Pauli matrices
  \eqref{CoordinatePauliMatrices}
  from Lemma \ref{OctonionicPauliMatrices}.
\end{lemma}

The following observation is due to \cite{KugoTownsend82},
with a streamlined review in \cite{BH10}:
\begin{prop}[Real Spin representations in dimension 11, 7, 5, 4, via
$\mathbb{K}$-Dirac matrices]
\label{RealSpinRepsInDimPlus1ViaOconions}
For $\mathbb{K} \in \{\mathbb{R}, \mathbb{C}, \mathbb{H}, \mathbb{O}\}$,
the assignment
\begin{equation}
  \label{RealSpinRepsInDimPlus1}
  \xymatrix@R=-4pt{
    \mathrm{Cl}_{\mathbb{R}}
    \big(
      \mathrm{dim}_{{}_{\mathbb{R}}}+2,1
    \big)
    \ar[rr]
    &&
    \mathrm{End}_{\mathbb{R}}\big( \mathbb{K}^2 \oplus \mathbb{K}^2 \big)
    \\
    {\pmb\Gamma}_a
    \ar@{|->}[rr]
    &&
\left\{
    {\footnotesize \begin{array}{ccc}
      \left[
      \!\!\!
      {\begin{array}{cc}
        1 & 0
        \\
        0 & -1
      \end{array}}
      \!\!\!
      \right]
            \cdot
      (-)
      && for\;\;
      a = 5'
            \\
            \\
                      {\footnotesize
      \left[
      \!\!\!
      {\begin{array}{cc}
        0 & \sigma^a
        \\
        \sigma_a & 0
      \end{array}}
      \!\!\!
      \right]
      }
      \cdot
      (-)
      &&
      \mbox{otherwise}
    \end{array}}
    \right.
  }
\end{equation}
of left multiplication action by the $\mathbb{K}$-Dirac matrices
\eqref{OctonionicDiracMatrices} constitutes a real representation
of the full Clifford algebra \eqref{CliffordRelation}
on the real vector space underlying $\mathbb{K}^4$.
\end{prop}
\begin{remark}[Branching of 11d spinors in 10d]
  \label{BranchingOf11dSpinorsIn10d}
The linear representation of
$\mathrm{Spin}(\mathrm{dim}_{{}_{\mathbb{R}}}\mathbb{K}+2,1 )$
corresponding via \eqref{LieSpinRepresentation} to the Clifford representation \eqref{RealSpinRepsInDimPlus1}
restricted to a representation of
$\mathrm{Spin}(\mathrm{dim}_{{}_{\mathbb{R}}}\mathbb{K}+1,1 )$
(omitting the index $5'$)
is manifestly the direct sum of the two representations in Lemma \ref{OctonionicPauliMatrices}:
\begin{equation}
  \xymatrix@R=-2pt{
    \mathrm{Cl}^{\mathrm{ev}}_{\mathbb{R}}
    \big(
      \mathrm{dim}_{{}_{\mathbb{R}}}\mathbb{K}  +1,1
    \big)
    \;\; \ar@{^{(}->}[r]
    & \;
    \mathrm{Cl}^{\mathrm{ev}}_{\mathbb{R}}
    \big(
      \mathrm{dim}_{{}_{\mathbb{R}}}\mathbb{K}  +2,1
    \big)
    \ar[r]
    &
    \mathrm{End}_{\mathbb{R}}\big( \mathbb{K}^2 \oplus \mathbb{K}^2 \big)
    \\
    {\pmb\Gamma}_{a_1}
    {\pmb\Gamma}_{a_2}
    \ar@{|->}[rr]
    &&
    \footnotesize
    \left[
    \!\!\!
    {\begin{array}{cc}
      \sigma^{a_1} \cdot (\sigma_{a_2}\cdot (-)) & 0
      \\
      0 & \sigma_{a_1} \cdot (\sigma^{a_2}\cdot (-))
    \end{array}}
    \!\!\!
    \right]
  }
\end{equation}
\end{remark}

In terms of $\mathbb{K}$-Dirac matrix calculus, the
Spin-invariant spinor pairing is given as follows:

\begin{prop}[Spinor pairing]
  \label{SpinorPairingViaKDiracMatrices}
  For $\mathbb{K} \in \{\mathbb{R}, \mathbb{C}, \mathbb{H}, \mathbb{O}\}$,
  let $\mathbf{N} := \mathbb{K}^2 \oplus \mathbb{K}^2$ be
  the $N := 4 \mathrm{dim}_{{}_{\mathbb{R}}} \mathbb{K}$-dimensional \newline
  $\mathrm{Spin}\big( \mathrm{dim}_{{}_{\mathbb{R}}}\mathbb{K} +2, 1 \big)$ representation
  from Prop. \ref{RealSpinRepsInDimPlus1ViaOconions}. Then
  the {\it spinor bilinear pairing} is
  \begin{equation}
    \label{SpinrorPairing}
    \xymatrix@R=-2pt{
      \mathbf{N}
      \;\times\;
      \mathbf{N}
      \ar[rr]
      &&
      \mathbb{R}
      \\
      \psi, \phi
      \ar@{|->}[rr]
      &&
      \langle\psi,\phi\rangle
      \mathrlap{
        :=
        \mathrm{Re}
        \big(
          \psi^\dagger
          \!\!\cdot
          {\pmb\Gamma}_0
          \!
          \cdot
          \phi
        \big)
      }
    }
  \end{equation}
  where on the right $\psi^\dagger := {}^{t}\psi^\ast$
  is the Hermitian conjugate, Notation \ref{Conjugation},
  and where ``$\cdot$'' denotes matrix multiplication over
  $\mathbb{K}$.\footnote{The triple matrix product
  in \eqref{SpinrorPairing}
  is associative even over $\mathbb{O}$, since the components of
  { ${\pmb\Gamma}_0 = \footnotesize \left[\!\!\begin{array}{cc} 0 & -1_{2 \times 2} \\ 1_{2 \times 2} & 0\end{array}\!\!\right]$}
  are real.}
  Furthermore, this is bilinear, skew-symmetric and
  $
    \mathrm{Spin}
    \big(
      \mathrm{dim}_{{}_{\mathbb{R}}}\mathbb{K} +2, 1
    \big)
  $-invariant.
\end{prop}

\subsection{The $\mathrm{A}_1$-type singularity}
\label{TheA1TypeSingularity}

Here we spell out basics of the
linear representations of $\mathrm{Sp}(1)_{L/R} \simeq \mathrm{SU}(2)_{L/R}$
by left/right quaternion multiplication on the quaternion space $\mathbb{H}$. The resulting orbifold quotient
$\mathbb{H}\sslash \mathbb{Z}_{n+1}$ for
$\mathbb{Z}_{n+1} \subset \mathrm{SU}(2)_L$ is the $\mathrm{A}_n$-type
singularity (e.g. \cite{SS19a}).

\medskip

Since some prefactors in the following crucially matter,
we begin by making fully explicit:

\vspace{.2cm}

\noindent {\bf The exceptional isomorphism $\mathrm{Spin}(4)\cong \mathrm{SU}(2)_L\times \mathrm{SU}(2)_R$.}
As recalled in \eqref{LieBracketInSO}, the 6-dimensional real Lie algebra $\mathfrak{so}(4)$ has a distinguished basis
$\{R_{a_1 a_2}\}_{1\leq a_1 < a_2\leq 4}$ with commutation relations
\begin{equation}
  \label{LieBracketso}
  [
    R_{ij},
    R_{kl}
  ]
    =
    \delta_{jk} R_{il}
    +
    \delta_{il} R_{jk}
    -
    \delta_{jl} R_{ik}
    -
    \delta_{ik} R_{jl}\;.
\end{equation}
Define elements $J^i_{L/R} \in \mathfrak{so}(4)$ by
\begin{equation}
  \label{su2GeneratorsInso4}
  J^{i-1}_{L}
       :=
    -\tfrac{1}{2}R^{1i}
    -
    \tfrac{1}{4}
    \epsilon^{1ijk} R_{jk}\;,
  \qquad
  J^{i-1}_{R}
        :=
    \phantom{-}
    \tfrac{1}{2}
    R^{1i}
    -
   \tfrac{1}{4}
  \epsilon^{1ijk} R_{jk}\;,
\qquad
 i,j,k \in \{2,3,4\}.
  \end{equation}
  More explicitly:

  \vspace{-6mm}
   \begin{equation}
  \mbox{
    \begin{tabular}{c}
    \end{tabular}
  }
  \phantom{A}
  \begin{array}{lllll}
  \mathclap{\phantom{\tfrac{1}{2}}}
  J^1_L
    =
  -\tfrac{1}{2}R_{12} - \tfrac{1}{2}R_{34}\,,
  &&
  J^2_L
    =
  -\tfrac{1}{2}R_{13} - \tfrac{1}{2}R_{42}\,,
  &&
  J^3_L
    =
  -\tfrac{1}{2}R_{14} - \tfrac{1}{2}R_{23}\,,
  \\
  \mathclap{\phantom{\tfrac{1}{2}}}
  J^1_R
    =
  \phantom{-}\, \tfrac{1}{2}R_{12} - \tfrac{1}{2}R_{34}\,,
  &&
  J^2_R
    =
    \phantom{-}\, \tfrac{1}{2}R_{13} - \tfrac{1}{2}R_{42}\,,
  &&
  J^3_R
    =
  \phantom{-}\, \tfrac{1}{2}R_{14} - \tfrac{1}{2}R_{23}
  \,.
  \\
  \mathclap{\phantom{\vert^{\vert^{\vert^{\vert^{\vert}}}}}}
  {\phantom{A}}
  \end{array}
\end{equation}

\vspace{-6mm}
\noindent It is immediate that $\{J^i_L, J^j_R\}_{i,j}$ is a linear basis of  $\mathfrak{so}(4)$.
Moreover, with \eqref{LieBracketso}, one finds the relations
$$
  [J^i_L, J^j_L]
  \;=\;
  \epsilon^{i j}{}_k J^k_{L}
  \,,
  \phantom{AAAA}
  [J^i_R, J^j_R]
  \;=\;
  \epsilon^{i j}{}_k J^k_{R}
  \,,
  \phantom{AAAA}
  [J^i_L, J^j_R]
  \;=\;
  0\;.
$$
Therefore, the subspaces
\begin{equation}
  \mathfrak{su}(2)_L
  \;:=\;
  \big\langle
    J^1_L, J^2_L, J^3_L
  \big\rangle
  \;\subset\;
  \mathfrak{so}(4)
  \qquad
  \textrm{and}
  \qquad
  \mathfrak{su}(2)_R
  \;:=\;
  \big\langle
    J^1_R, J^2_R, J^3_R
  \big\rangle
  \;\subset\;
  \mathfrak{so}(4)
\end{equation}
are two mutually commuting
Lie subalgebras of $\mathfrak{so}(4)$,
both canonically isomorphic to $\mathfrak{su}(2)$
\eqref{LieBracketInSU2},
whose joint embedding is a Lie algebra isomorphism.
This consequently induces an isomorphism between
the corresponding simply connected compact Lie groups:
\begin{equation}
  \label{ExceptionalIsomorphismSpin4}
   \mathfrak{su}(2)_L
     \oplus
   \mathfrak{su}(2)_R
   \overset{\cong}{\longrightarrow}
   \mathfrak{so}(4)
   \,,
   \phantom{AAAAAA}
  \mathrm{SU}(2)_L\times \mathrm{SU}(2)_R
  \overset{\cong}{\longrightarrow}
  \mathrm{Spin}(4)
  \,.
\end{equation}

\begin{remark}[Clifford representation of $\mathfrak{su}(2)_{L/R}$]
  \label{CliffordRepresentationOfsu2LR}
The images ${\pmb\tau}^i_{L/R}$ of the elements
$J^i_{L/R}\in \mathfrak{so}(4)$ \eqref{su2GeneratorsInso4}
under the embedding $\mathfrak{so}(4)\hookrightarrow \mathrm{Cl}_{\mathbb{R}}(4)$ \eqref{LieSpinRepresentation}
are given by
\begin{equation}
\label{su2GeneratorsInTermsOfGammaMatrices}
\begin{array}{ccccc}
  {\pmb\tau}^1_L
  \;=\;
  -\tfrac{1}{4}{{\pmb\Gamma}_1}{\pmb\Gamma}_2
  -\tfrac{1}{4}{{\pmb\Gamma}_3}{\pmb\Gamma}_4\,,
  &&
  {\pmb\tau}^2_L
    \;=\;
  -\tfrac{1}{4}{{\pmb\Gamma}_1}{\pmb\Gamma}_3
  -\tfrac{1}{4}{{\pmb\Gamma}_4}{\pmb\Gamma}_2\,,
  &&
  {\pmb\tau}^3_L
   \;=\;
  -\tfrac{1}{4}{{\pmb\Gamma}_1}{\pmb\Gamma}_4
  -\tfrac{1}{4}{{\pmb\Gamma}_2}{\pmb\Gamma}_3\,,
  \\
  \mathclap{\phantom{\vert^{\vert^{\vert^{\vert}}}}}
  {\pmb\tau}^1_R
    \;=\;
   \phantom{-}
   \tfrac{1}{4}{{\pmb\Gamma}_1}{\pmb\Gamma}_2
  -\tfrac{1}{4}{{\pmb\Gamma}_3}{\pmb\Gamma}_4\,,
  &&
  {\pmb\tau}^2_R
    \;=\;
   \phantom{-}
   \tfrac{1}{4}{{\pmb\Gamma}_1}{\pmb\Gamma}_3
  -\tfrac{1}{4}{{\pmb\Gamma}_4}{\pmb\Gamma}_2\,,
  &&
  \;
  {\pmb\tau}^3_R
    \;=\;
  \phantom{-}
   \tfrac{1}{4}{{\pmb\Gamma}_1}{\pmb\Gamma}_4
  -\tfrac{1}{4}{{\pmb\Gamma}_2}{\pmb\Gamma}_3
  \,.
\end{array}
\end{equation}
\end{remark}

\begin{lemma}[$\Z_2$ inside ${\rm SU}(2)$]
  \label{Z2inSU2}
The 1-parameter subgroup $\{\exp(2\pi t J_L^1)\}_{t\in \mathbb{R}/\mathbb{Z}}\subseteq \mathrm{SU}(2)_L$ \eqref{ExceptionalIsomorphismSpin4}
generated by the infinitesimal rotation $J_L^1$ \eqref{su2GeneratorsInso4}
is a copy of $\mathrm{U}(1)$.
Inside it we have a copy of
$\mathbb{Z}_2\cong \{1,-1\}$ given by $\{1,\exp(2\pi J_L^1)\}$
\eqref{2PiRotationOnSpinors}.
The image
\eqref{su2GeneratorsInTermsOfGammaMatrices}
of the generator of this copy of $\mathbb{Z}_2$
in the Clifford algebra \eqref{LieSpinRepresentation} is
\begin{equation}
  \label{HalfSpin3Rotation}
  \exp(2\pi {\pmb\tau}_L^1)
  \;=\;
  {\pmb\Gamma}_1{\pmb\Gamma}_2{\pmb\Gamma}_3{\pmb\Gamma}_4\;.
\end{equation}
\end{lemma}
\begin{proof}
 Observing that
 $({\pmb\Gamma}_1{\pmb\Gamma}_2)^2=({\pmb\Gamma}_3{\pmb\Gamma}_4)^2=-1$
 and
 that ${\pmb\Gamma}_1{\pmb\Gamma}_2$
 commutes with ${\pmb\Gamma}_3{\pmb\Gamma}_4$
 we have, with \eqref{su2GeneratorsInTermsOfGammaMatrices}:
\[
  \exp
  \big(
    2\pi {\pmb\tau}_L^1
  \big)
  \;=\;
  \exp
  \big(
    -\tfrac{\pi}{2}
    {{\pmb\Gamma}_1}
    {\pmb\Gamma}_2
    -
    \tfrac{\pi}{2}{{\pmb\Gamma}_3}{\pmb\Gamma}_4
  \big)
  =
  \exp\big(
    -\tfrac{\pi}{2}{{\pmb\Gamma}_1}{\pmb\Gamma}_2
  \big)
  \exp
  \big(
    -\tfrac{\pi}{2}{{\pmb\Gamma}_3}{\pmb\Gamma}_4
  \big)
  =
  (-{\pmb\Gamma}_1{\pmb\Gamma}_2)(-{\pmb\Gamma}_3{\pmb\Gamma}_4)
  =
  {\pmb\Gamma}_1{\pmb\Gamma}_2{\pmb\Gamma}_3{\pmb\Gamma}_4
  \,,
\]
where in the third step we used Euler's formula \eqref{EulerFormulaInCliffordAlgebra}
in the Clifford algebra.
\hfill \end{proof}

\begin{example}[Defining representation of the symplectic group]
  \label{The4OfSp1}
  The defining action of
  $\mathrm{Sp}(1) := \{q \in \mathbb{H} \vert q q^\ast = 1  \}$
  on the space $\mathbb{H}$ of quaternions is a real 4-dimensional
  irreducible representation, to be denoted
  \begin{equation}
    \label{4OfSp1}
    \mathbf{4} \;\in\; \mathrm{Rep}_{\mathbb{R}}(\mathrm{Sp}(1))
    \,.
  \end{equation}
  A priori there are two distinct such representations,
  given by left and quaternion multiplication, respectively
  \begin{equation}
    \label{LeftMultiplicationOfUnitQuaternions}
    \xymatrix@C=2pt@R=-4pt{
      \mathrm{Sp}(1) \ar@{}[r]|-{\times} & \mathbb{H}
      \ar[rrrr]
      &&&&
      \mathbb{H}
      \\
      (q,& v)
      \ar@{|->}[rrrr]
      &&&&
      q \cdot v
    }
    \,,
    \phantom{AAAAAA}
    \xymatrix@C=2pt@R=-4pt{
      \mathrm{Sp}(1) \ar@{}[r]|-{\times} & \mathbb{H}
      \ar[rrrr]
      &&&&
      \mathbb{H}\;\;\;.
      \\
      (q,& v)
      \ar@{|->}[rrrr]
      &&&&
      v \cdot q^\ast
    }
  \end{equation}

  \vspace{-3mm}
 \noindent  However, these are clearly isomorphic, via conjugation on
  all quaternions $\xymatrix{\mathbb{H} \ar[r]^{(-)^\ast}_-{\simeq_{\mathbb{R}}}& \mathbb{H}}$.
\end{example}

\begin{example}[${\rm Spin} (4)$ and two copies of ${\rm Sp}(1)$]
  Under the exceptional isomorphism \eqref{ExceptionalIsomorphismSpin4}
  \begin{equation}
    \label{ExceptionalIsomorphismFromSpSpToSpin4}
    \xymatrix{
      \mathrm{Sp}(1)_L
      \times
      \mathrm{Sp}(1)_R
      \;\cong\;
      \mathrm{SU}(2)_L\times \mathrm{SU}(2)_R
      \ar[r]^-{\simeq}
      &
      \mathrm{Spin}(4)
    }
  \end{equation}
  the vector representation $\mathbf{4}$ of $\mathrm{Spin}(4)$
  from \eqref{4OfSp1}
  is given by combined left and conjugate right multiplication
  \eqref{LeftMultiplicationOfUnitQuaternions}
  \begin{equation}
    \label{QuaternionIncarnationOf4dVectorRep}
    \xymatrix@C=2pt@R=-4pt{
      \mathrm{Sp}(1)_L
      \ar@{}[r]|-{\times}
        &
      \mathrm{Sp}(1)_R
       \ar@{}[r]|-{\times}
       &
      \mathbb{H}
      \ar[rrrr]
      &&&&
      \mathbb{H}
      \\
      (q_A, & q_V, &  v)
      \ar@{|->}[rrrr]
      &&&&
      q_A \cdot v \cdot \bar q_V
    }
  \end{equation}
  All three of these actions are irreducible in themselves,
  so that under restriction \eqref{RestrictionOfRepresentations}
  along any of the two inclusions \eqref{ExceptionalIsomorphismSpin4}
  $
    \mathrm{Sp}(1)_{L/R}
    \;
    \overset{\iota_{L/R}}{\longhookrightarrow}
    \;
    \mathrm{Spin}(4)
  $
  there is ``no branching'' in that
  $
    (\iota_{A/V})^\ast \mathbf{4}
    =
    \mathbf{4}
      $.
\end{example}

\begin{example}[Reduction to ${\rm Spin} (3)$]
  \label{TheExceptionalIsoSpin3ToSp1}
  Under the exceptional isomorphism
  \eqref{ExceptionalIsomorphismSpin4}
  and the further exceptional isomorphism
  $
       \mathrm{Sp}(1)
      \overset{\simeq}{\to}
  \mathrm{Spin}(3)
      $
  the canonical inclusion
  $\mathrm{Spin}(3) \overset{\iota}{\hookrightarrow} \mathrm{Spin}(4)$
  is identified with the diagonal map on $\mathrm{Sp}(1)$:
  $$
    \xymatrix@R=2pt{
      \mathrm{Spin}(3)
    \;  \ar@{}[d]|-{
        \begin{rotate}{-90}
          $\mathclap{\simeq}$
        \end{rotate}
      }
      \ar@{^{(}->}[rr]^-{ \iota }
      &&
      \mathrm{Spin}(4)
      \ar@{}[d]|-{
        \begin{rotate}{-90}
          $\mathclap{\simeq}$
        \end{rotate}
      }
      \\
      \mathrm{Sp}(1)
      \ar[rr]^-{ \mathrm{diag} }
      &&
      \mathrm{Sp}(1)_L \times \mathrm{Sp}(1)_R
      \\
      \mathbf{1} \oplus \mathbf{3}
      &
            \overset{\iota^\ast}
            {\longmapsfrom}
      & \mathbf{4}
    }
  $$
  and hence by restriction of the 4-dimensional
  vector representation
  in its quaternion form \eqref{QuaternionIncarnationOf4dVectorRep}
  it follows that the resulting
  $\mathbf{3} \in \mathrm{Rep}_{\mathbb{R}}( \mathrm{Sp}(1) )$ is given by
  the diagonal of the actions \eqref{LeftMultiplicationOfUnitQuaternions}
  \begin{equation}
    \label{QuaternionIncarnationOf3dVectorRep}
    \xymatrix@C=2pt@R=-4pt{
      \mathrm{Sp}(1)
       \ar@{}[r]|-{\times}
       &
      \mathbb{H}_{\mathrm{im}}
      \ar[rrrr]
      &&&&
      \mathbb{H}_{\mathrm{im}}
      \\
      (q, &  v)
      \ar@{|->}[rrrr]
      &&&&
      q \cdot v \cdot \bar q
    }
  \end{equation}
  Here $\mathbb{H}_{\mathrm{im}} \subset \mathbb{H}$
  is the real 3-dimensional space of imaginary quaternions.
\end{example}

\paragraph{The Lie algebra $\mathfrak{sp}_1$ via quaternions.}
As $\mathrm{Sp}(1)$ is the unit sphere of the skew-field $\mathbb{H}$ of quaternions, the
identification $\mathbb{H}\cong \mathbb{R}^4$ given by the
standard $\mathbb{R}$-basis
$\{\mathrm{e}_0,\mathrm{e}_1,\mathrm{e}_2,\mathrm{e}_2\}$ of $\mathbb{H}$ with
\begin{equation}
  \label{CanonicalQuaternionBasis}
    \begin{aligned}
    &
    \mathrm{e}_0 = 1\,,
    \qquad
       \mathrm{e}_i \mathrm{e}_i = -1 = - \mathrm{e}_0  \phantom{A} \mbox{for $i \in \{1,2,3\}$}\,,
 \qquad
    \mathrm{e}_{\sigma(1)} \mathrm{e}_{\sigma(3)} =  \mathrm{sgn}(\sigma)  \mathrm{e}_{\sigma(3)}
    \phantom{a}
    \mbox{for $\sigma \in \mathrm{Sym}(3)$},
    \end{aligned}
\end{equation}
 identifies the Lie algebra $\mathfrak{sp}_1$ of
$\mathrm{Sp}(1)$ with the tangent space at $S^3$ in $\mathbb{R}^4$ at the point $(1,0,0,0)$,
and so it is identified with the space $\mathbb{H}_{\mathrm{im}}$ of imaginary quaternions.
The Lie algebra structure on $\mathfrak{sp}_1$ is easily obtained by noticing that the 1-parameter
subgroups generated by the basis elements $\{\mathrm{e}_1,\mathrm{e}_2,\mathrm{e}_3\}$ of  $\mathfrak{sp}_1$ are given by
$
\mathrm{e}_i(t)= \cos(t)+ \mathrm{e}_i \sin(t)
$.
From this we get
$$
[\mathrm{e}_1,\mathrm{e}_2]= \frac{1}{2}\frac{d^2}{dsdt}\biggr\vert_{(s,t)=
(0,0)} \mathrm{e}_1(t)\mathrm{e}_2(s)\mathrm{e}_1(t)^{-1}\mathrm{e}_2(s)^{-1}=\mathrm{e}_3\,,
$$
and similarly
\[
[\mathrm{e}_2,\mathrm{e}_3]=\mathrm{e}_1 \qquad  {\rm and} \qquad [\mathrm{e}_3,\mathrm{e}_1]=\mathrm{e}_2\,.
\]
Therefore the basis vectors $\{\mathrm{e}_i\}$ are the standard Lie algebra basis for
$\mathfrak{so}_3\cong \mathfrak{sp}_1$.  Finally, by differentiating the action,
one sees that the Lie algebra representations corresponding to the representations
$\mathbf{4}_{l/r} \cong \mathbb{H}$ from \eqref{LeftMultiplicationOfUnitQuaternions},
and
$\mathbf{3} \cong \mathbb{H}_{\mathrm{im}}$ from \eqref{QuaternionIncarnationOf3dVectorRep}
are given, respectively, by
\begin{alignat}{4}
  \label{443}
  \mathfrak{sp}_1\otimes \mathbf{4}& \longrightarrow \mathbf{4}\,,
  \qquad  \qquad
  \mathfrak{sp}_1\otimes \mathbf{4}& \longrightarrow \mathbf{4}\,,
  \qquad  \qquad
  \mathfrak{sp}_1\otimes \mathbf{3} & \longrightarrow \mathbf{3}\;\;.
  \\
  \notag
  \mathrm{e}_i\otimes v&\longmapsto \mathrm{e}_i\,v\,
  \qquad  \qquad
  \mathrm{e}_i\otimes v \! & \longmapsto v\,\mathrm{e}_i\,
  \qquad  \qquad
  \mathrm{e}_i\otimes v & \longmapsto [\mathrm{e}_i,v]
\end{alignat}
 Here the multiplications and the commutators on the right are taken
 in the associative algebra $\mathbb{H}$ of quaternions.
As $\mathrm{Sp}(1)$ is a compact and simply connected Lie group, its Lie algebra
$\mathfrak{sp}_1$ knows everything about its representation theory. An example
of application of this principle are the proofs of the following lemmas.

\begin{lemma}[Decomposition of irreps of $\mathbf{4}\wedge \mathbf{4}$]
  \label{IrrepDecompositionOf4Wedge4}
  The second exterior power $\mathbf{4}\wedge \mathbf{4}$ (see \eqref{WedgeProductRepresentation})
  of the defining real 4-dimensional irrep
  $\mathbf{4} \in \mathrm{Rep}_{\mathbb{R}}( \mathrm{Sp}(1) )$
 (see \eqref{4OfSp1})
  is the direct sum of the real 3-dimensional irrep
  $\mathbf{3} \in \mathrm{Rep}_{\mathbb{R}}( \mathrm{Sp}(1) )$
  (see \eqref{QuaternionIncarnationOf3dVectorRep})
  with the 3-dimensional trivial rep:
  \begin{equation}
    \label{IrrepDecompositionOf4Wedge4Equ}
    \mathbf{4}\wedge \mathbf{4}
    \;\simeq\;
    3 \cdot \mathbf{1}
    \;\oplus\;
    1 \cdot \mathbf{3}
    \;\;\;\in\;
    \mathrm{Rep}_{\mathbb{R}}(\mathrm{Sp}(1))\;.
  \end{equation}
\end{lemma}
\begin{proof}
By \eqref{443},
the $\mathfrak{sp}_1$-action on $\mathbf{4}$ is given on the canonical linear basis \eqref{CanonicalQuaternionBasis}
  \begin{equation}
    \label{BasisFor4}
    \mathbf{4}
    \;\simeq\;
    \big\langle
      \mathrm{e}_0, \mathrm{e}_1, \mathrm{e}_2, \mathrm{e}_3
    \big\rangle_{\mathbb{R}}
  \end{equation}
by $\mathrm{e}_i\otimes \mathrm{e}_j\mapsto \mathrm{e}_i\mathrm{e}_j$, with $i\in \{1,2,3\}$ and $j\in \{0,1,2,3\}$.
Consider then the following linear basis of $\mathbf{4}\wedge \mathbf{4}$:
  \begin{equation}
    \label{LinearBasisFor4Wedge4}
    \mathbf{4}\wedge \mathbf{4}
    \;\simeq_{\mathbb{R}}\;
    \left\langle
      \begin{array}{c}
        a_1^L
        :=
        \mathrm{e}_0 \wedge \mathrm{e}_1 + \mathrm{e}_2 \wedge \mathrm{e}_3,
        \\
        a_2^L
        :=
        \mathrm{e}_0 \wedge  \mathrm{e}_2 + \mathrm{e}_3 \wedge  \mathrm{e}_1,
        \\
        a_3^L
        :=
        \mathrm{e}_0 \wedge \mathrm{e}_3 + \mathrm{e}_1 \wedge  \mathrm{e}_2,
      \end{array}
      \begin{array}{c}
        a_1^R
        :=
        \mathrm{e}_0 \wedge \mathrm{e}_1 - \mathrm{e}_2 \wedge \mathrm{e}_3,
        \\
        a_2^R
        :=
        \mathrm{e}_0 \wedge \mathrm{e}_2 - \mathrm{e}_3 \wedge \mathrm{e}_1,
        \\
        a_3^R
        :=
        \mathrm{e}_0 \wedge \mathrm{e}_3 - \mathrm{e}_1 \wedge \mathrm{e}_2
      \end{array}
    \right\rangle.
  \end{equation}
 The induced $\mathfrak{sp}(1)$ Lie algebra action is given by
  $$
    \mathrm{e}_i \cdot ( \mathrm{e}_j \wedge \mathrm{e}_k )
    \;=\;
    (\mathrm{e}_i \mathrm{e}_j) \wedge \mathrm{e}_k
    +
    \mathrm{e}_j \wedge (\mathrm{e}_i \mathrm{e}_k)
    \,.
  $$
  From this we find for $e_1$:
  $$
    \begin{aligned}
    \mathrm{e}_1 \cdot a_1^{L/R}
    & = \;
    \mathrm{e}_1
      \cdot
    \big(
      \mathrm{e}_0 \wedge \mathrm{e}_1
      \pm
      \mathrm{e}_2 \wedge \mathrm{e}_3
    \big)
    \\
    & =\;
    \big(
    \underset{
      = 0
    }{
      \underbrace{
        \mathrm{e}_1 \wedge \mathrm{e}_1
      }
    }
    +
    \underset{
      = 0
    }{
      \underbrace{
        \mathrm{e}_0 \wedge (-\mathrm{e}_0)
      }
    }
    \big)
    \pm
    \big(
    \underset{
      = 0
    }{
      \underbrace{
        \mathrm{e}_3 \wedge \mathrm{e}_3
      }
    }
    +
    \underset{
      = 0
    }{
      \underbrace{
        \mathrm{e}_2 \wedge (- \mathrm{e}_2)
      }
    }
    \big)
     =\;
    0
\\
    \mathrm{e}_1 \cdot a_2^{L/R}
    & =\;
    \mathrm{e}_1
      \cdot
    \big(
      \mathrm{e}_0 \wedge \mathrm{e}_2
      \pm
      \mathrm{e}_3 \wedge \mathrm{e}_1
    \big)
    \\
    & = \;
    \big(
    \underset{
      = \mathrm{e}_1 \wedge \mathrm{e}_2
    }{
      \underbrace{
        \mathrm{e}_1 \wedge \mathrm{e}_2
      }
    }
    +
    \underset{
      = \mathrm{e}_0 \wedge \mathrm{e}_3
    }{
      \underbrace{
        \mathrm{e}_0 \wedge \mathrm{e}_3
      }
    }
    \big)
    \pm
    \big(
    \underset{
      = \mathrm{e}_1 \wedge \mathrm{e}_2
    }{
      \underbrace{
        (- \mathrm{e}_2) \wedge \mathrm{e}_1
      }
    }
    +
    \underset{
      = \mathrm{e}_0 \wedge \mathrm{e}_3
    }{
      \underbrace{
        \mathrm{e}_3 \wedge (- \mathrm{e}_0)
      }
    }
    \big)
   =\;
    \left\{
      \begin{array}{l}
        2 a_3^L
        \\
        0
      \end{array}
    \right.
\\
    \mathrm{e}_1 \cdot a_3^{L/R}
    & =\;
    \mathrm{e}_1
      \cdot
    \big(
      \mathrm{e}_0 \wedge \mathrm{e}_3
      \pm
      \mathrm{e}_1 \wedge \mathrm{e}_2
    \big)
    \\
    & = \;
    \big(
    \underset{
      = - \mathrm{e}_3 \wedge \mathrm{e}_1
    }{
      \underbrace{
        \mathrm{e}_1 \wedge \mathrm{e}_3
      }
    }
    +
    \underset{
      = - \mathrm{e}_0 \wedge \mathrm{e}_2
    }{
      \underbrace{
        \mathrm{e}_0 \wedge (-\mathrm{e}_2)
      }
    }
    \big)
    \pm
    \big(
    \underset{
      = - \mathrm{e}_0 \wedge \mathrm{e}_2
    }{
      \underbrace{
        (- \mathrm{e}_0) \wedge \mathrm{e}_2
      }
    }
    +
    \underset{
      = - \mathrm{e}_3 \wedge \mathrm{e}_1
    }{
      \underbrace{
        \mathrm{e}_1 \wedge \mathrm{e}_3
      }
    }
    \big)
 =\;
    \left\{
      \begin{array}{l}
        -2 a_2^L
        \\
        0 \;.
      \end{array}
    \right.
    \end{aligned}
      $$
  Since everything here is invariant under cyclic permutation of
  the three non-zero indices, it follows generally that
  $$
    (\tfrac{1}{2} \mathrm{e}_i) \cdot a_j^L
      \;=\;
    \underset{k}{\sum} \epsilon_{i j k} a_k^L
    \,,
    \phantom{AA}
    (\tfrac{1}{2} \mathrm{e}_i) \cdot a_j^R \;=\; 0
    \phantom{AA}
    \mbox{for all $i,j \in \{1,2,3\}$}.
  $$
    This identifies $\big\langle \{a^L_1, a^L_2, a^L_3\} \big\rangle$ and
  $\big\langle \{a^R_i\} \big\rangle$ as $\mathbf{3}$ and $\mathbf{1}$,
  respectively, as representations of $\mathfrak{sp}_1$ and hence as
  representations of $\mathrm{Sp}(1)$:
  \begin{equation}
    \label{BasesFor3and1}
    \big\langle \{a^L_1, a^L_2, a^L_3\} \big\rangle
    \;\simeq\;
    \mathbf{3}
    \,,
    \phantom{aa}
    \big\langle \{a^R_i\} \big\rangle
    \;\simeq\;
    \mathbf{1}
    \;\;\;\in
    \mathrm{Rep}_{\mathbb{R}}(\mathrm{Sp}(1))
    \,.
  \end{equation}

\vspace{-2mm}
\hfill \end{proof}

\begin{lemma}[Left-right exchange of $\mathbf{4} \wedge \mathbf{4}$]
  \label{LeftRightExchangeOf4Wedge4}
  Under the exceptional isomorphism
  $\mathrm{Sp}(1)_R \times \mathrm{Sp}(1)_L
  \overset{\cong}{\to} \mathrm{Spin}(4)$ (from \eqref{ExceptionalIsomorphismFromSpSpToSpin4}),
  the second exterior power $\wedge^2 \mathbf{4}$ of the
  vector representation $\mathbf{4}$ of $\mathrm{Spin}(4)$
  splits as the direct sum of the outer tensor products (from \eqref{OuterTensorProductOfRepresentations})
  of the
  $\mathbf{3}$ (from \eqref{QuaternionIncarnationOf3dVectorRep})
  of one of the $\mathrm{Sp}(1)$ factors with the $\mathbf{1}$ of the other factor:
  $$
    \wedge^2 \mathbf{4}
    \;\;\simeq\;\;
    \mathbf{3} \boxtimes \mathbf{1}
    \;\oplus\;
    \mathbf{1} \boxtimes \mathbf{3}
    \;\;\;\;
    \in
    \;
    \mathrm{Rep}_{\mathbb{R}}( \mathrm{Sp}(1)_L \times \mathrm{Sp}(1)_R )
    \,.
$$
\end{lemma}
\begin{proof}
  We have to show that we have an $\mathbb{R}$-vector space splitting
 \[
 \wedge^2 \mathbf{4}=V\oplus W
 \]
with both $V$ and $W$ representations of  $\mathrm{Sp}(1)_{A/V}$ via the
   the restrictions along the two inclusions
  $
    \mathrm{Sp}(1)_{L/R}
  \;\;  \overset{\iota_{L/R}}{\longhookrightarrow}
\;    \mathrm{Sp}(1)_L \times \mathrm{Sp}(1)_R
  $
 and with
  $V\cong \mathbf{3}$ and $W\cong 3\cdot \mathbf{1}$ in $\mathrm{Rep}_{\mathbb{R}}( \mathrm{Sp}(1)_L )$, and $V\cong 3\cdot \mathbf{1}$ and $W\cong \mathbf{3}$ in $\mathrm{Rep}_{\mathbb{R}}( \mathrm{Sp}(1)_R )$.
  Set, in the notation of Lemma \ref{IrrepDecompositionOf4Wedge4},
 \[
   V
   \;:=\;
   \big\langle
     \{a^L_1, a^L_2, a^L_3\}
   \big\rangle;
   \qquad
   W
   \;:=\;
   \big\langle \{a^R_1, a^R_2, a^R_3\}
   \big\rangle.
 \]
  Then the identification
  \[
  V\oplus W = \mathbf{3}\,\,\oplus\,\, 3\cdot \mathbf{1} \in \mathrm{Rep}_{\mathbb{R}}( \mathrm{Sp}(1)_L )
  \]
  is precisely the content of Lemma \ref{IrrepDecompositionOf4Wedge4}, and the identification
   \[
  V\oplus W = 3\cdot \mathbf{1}\,\,\oplus\,\, \mathbf{3}  \in \mathrm{Rep}_{\mathbb{R}}( \mathrm{Sp}(1)_R )
  \]
  is proved analogously, by considering the $\mathfrak{sp}_1$-action on $\mathbb{H}$ induced by the $\mathrm{Sp}(1)$-action on the right (see equation (\ref{443})).
\hfill \end{proof}

\begin{lemma}[The third wedge power]
  \label{ThirdWedgePowerOf4}
  The fourth exterior power $\wedge^4 \mathbf{4}$  of $\mathbf{4} \in \mathrm{Rep}_{\mathbb{R}}(\mathrm{Sp}(1))$ is the trivial representation $\mathbf{1}$, and
  the third exterior power $\wedge^3 \mathbf{4}$  of $\mathbf{4} \in \mathrm{Rep}_{\mathbb{R}}(\mathrm{Sp}(1))$
  is isomorphic to $\mathbf{4}$ itself:
  $$
    \wedge^3 \mathbf{4}
    \;\simeq\;
    \mathbf{4}
    \;\;\;\;
    \in \;
    \mathrm{Rep}_{\mathbb{R}}(\mathrm{Sp}(1))
    \,.
  $$
\end{lemma}
\begin{proof}
For any $i\in\{0,1,2,3\}$ let $b_i\in \wedge^3 \mathbf{4}$ be the element
\[
b_i:=(-1)^i \mathrm{e}_0\wedge\cdots \wedge \widehat{\mathrm{e}_i} \wedge\cdots \mathrm{e}_3,
\]
where the factor $\mathrm{e}_i$ is omitted. Then $\{b_0,b_1,b_2,b_3\}$ is an $\mathbb{R}$-basis of $\wedge^3 \mathbf{4}$
and we have an isomorphism of $\mathbb{R}$-vector spaces $\wedge^3 \mathbf{4}\to \mathbf{4}$ given by $\phi\colon b_i \mapsto \mathrm{e}_i$. Direct inspection shows that $\phi$ is actually an isomorphism of $\mathfrak{sp}_1$ representations, and so of $\mathrm{Sp}(1)$ representations.
  This follows by direct inspection. For instance, for the
  Lie action of $\mathrm{e}_1$ we find:
  $$
  \hspace{-5mm}
    \begin{aligned}
      \phi(\mathrm{e}_1 \cdot b_0)
      & \; = \;
      \phi\big(\mathrm{e}_1 \cdot ( \mathrm{e}_1 \wedge \mathrm{e}_2 \wedge \mathrm{e}_3 )\big)
      \\
      & \; = \; \phi\Big(
      \underset{
        = b_1
      }{
        \underbrace{
          (- \mathrm{e}_0) \wedge \mathrm{e}_2 \wedge \mathrm{e}_3
        }
      }
      \;+\;
      \underset{
        = 0
      }{
        \underbrace{
          \mathrm{e}_1 \wedge \mathrm{e}_3 \wedge \mathrm{e}_3
        }
      }
      \;+\;
      \underset{
        = 0
      }{
        \underbrace{
          \mathrm{e}_1 \wedge \mathrm{e}_2 \wedge (- \mathrm{e}_2)
        }
      }\Big)
      \;=\;
      \phi(b_1)=\mathrm{e}_1=\mathrm{e}_1\cdot \mathrm{e}_0=\mathrm{e}_1\cdot \phi(b_0)\;,
\\
    \phi( \mathrm{e}_1 \cdot b_1)
      & \;=\;
    \phi \big(  \mathrm{e}_1 \cdot ( - \mathrm{e}_0 \wedge \mathrm{e}_2 \wedge \mathrm{e}_3 )\big)
      \\
      & \;=\;
     \phi\Big(
        -
      \underset{
        b_0
      }{
        \underbrace{
          \mathrm{e}_1 \wedge \mathrm{e}_2 \wedge \mathrm{e}_3
        }
      }
    -
      \underset{
        = 0
      }{
      \underbrace{
        \mathrm{e}_0 \wedge \mathrm{e}_3 \wedge \mathrm{e}_3
      }
      }
      -
      \underset{
        = 0
      }{
        \underbrace{
          \mathrm{e}_0 \wedge \mathrm{e}_2 \wedge (-\mathrm{e}_2)
        }
      }\Big)
 \;=\;
      \phi(- b_0)=-\mathrm{e}_0=\mathrm{e}_1\cdot \mathrm{e}_1=\mathrm{e}_1\cdot \phi(b_1)\;,
    \end{aligned}
$$
and similarly for the other cases.
\hfill \end{proof}

In summary, we have proven the first
statement of Theorem \ref{FieldsAtSingularity}:

\begin{prop}[The $\mathbf{3}$ inside $\mathbf{4} \wedge \mathbf{3}$]
  \label{3Inside4Wedge3}
  The $\mathrm{SU}(2)_L$-fixed locus inside the
  $\wedge^2 \mathbf{4}$ of
  $\mathrm{Spin}(4) \simeq \mathrm{SU}(2)_L \times \mathrm{SU}(2)_R$
  is the $\mathbf{3}$ of $\mathrm{SU}(2)_R$,
  while the $\mathrm{SU}(2)_L$-fixed locus in
  $\wedge^3 \mathbf{4}$ (and in $\mathbf{4}$) is trivial
\end{prop}
\begin{proof}
  Using the fact that fixed loci are the direct summands of the corresponding
  trivial representations
  \eqref{H2FixedPointsOfOuterTensorProductWithH2Irrep},
  this follows from
  Lemma \ref{LeftRightExchangeOf4Wedge4}
  and
  Lemma \ref{ThirdWedgePowerOf4}.
\hfill \end{proof}

\subsection{11d Spinors at the $\mathrm{A}_1$-singularity }
\label{11dSpinorsAtTheA1TypeSingularity}

We now combine \cref{Octonionic2ComponentSpinors}
with \cref{TheA1TypeSingularity} to discuss the representation theory of 11d spinors restricted to an $\mathrm{A}_1$-singularity.

\begin{defn}[Identifying $\Z_2$ subgroup]
With respect to the inclusion
$\mathrm{Spin}(5,1)\times \mathrm{Spin}(4)\subset \mathrm{Spin}(9,1)$
given by \eqref{ExOctopnionicPauliMatices},
consider now
the corresponding inclusion of the subgroup from Lemma \ref{Z2inSU2}:
\begin{equation}
  \label{Z2A}
  \mathbb{Z}^A_2
    \;:=\;
    \big\{
      1,\,
      \exp(2\pi J^1_L))
    \big\}
   \subset
   \mathrm{SU}(2)_L
   \subset
   \mathrm{Spin}(4)
   \subset
   \mathrm{Spin}(10,1)\;.
\end{equation}
\end{defn}

The following lemma is essentially the content of \cite[Lemma 4.13]{HSS18}:
\begin{lemma}[The fermionic $\mathbb{Z}_2^A$-fixed locus]
\label{FixedSubspaceOfGamma6789}
\label{lemma:fixed-locus}
The  fixed locus \eqref{FormingFixedPoints} of $\mathbb{Z}_2^{A}$ \eqref{Z2A}
in the real $\mathrm{Spin}(9,1)$ representations
$\mathbf{16}$ and $\overline{\mathbf{16}}$ \eqref{HeteroticSpinRepsAsOctonionicModules}
is, as a residual $\mathrm{Spin}(5,1)$ representations
the $\mathbf{8}$ and $\overline{\mathbf{8}}$
from \eqref{HeteroticSpinRepsAsOctonionicModules} respectively,
given under Lemma \ref{SpinRepsViaSL2O} by the inclusion
$\mathbb{H}^2 \subset \mathbb{O}^2$ \eqref{LinearSubspacesOfTheOctonions}
\begin{equation}
  \label{FixedLociOfZ2AIn16}
  \mathbf{16}^{\mathbb{Z}_2^A}
  \;=\;
  \mathbf{8}
  \,,
  \phantom{AA}
  \overline{\mathbf{16}}^{\mathbb{Z}_2^A}
  \;=\;
  \overline{\mathbf{\mathbf{8}}}
  \phantom{AAAA}
  \in
  \mathrm{Rep}_{\mathbb{R}}
  \big(
    \mathrm{Spin}(5,1)
  \big)
  \,.
\end{equation}
\end{lemma}
\begin{proof}
By Prop. \ref{SpinRepsViaSL2O}
we need to prove that, under the identifications of the $\mathrm{SO}(9,1)$ representations $\mathbf{16}$ and $\overline{\mathbf{16}}$ with $\mathbb{O}^2$,
the $\mathbb{Z}_2^{A}$ fixed locus in
$\mathbf{16}$ and $\overline{\mathbf{16}}$
is identified with the real subspace $\mathbb{H}^2$ of $\mathbb{O}^2$
\eqref{LinearSubspacesOfTheOctonions}.
According to Lemma \ref{Z2inSU2} and to Proposition \ref{SpinRepsViaSL2O},
  we are equivalently asking for the fixed locus of
  the consecutive left action of the octonionic Pauli matrices
  $\sigma_6$, $\sigma_7$, $\sigma_8$, $\sigma_9$ from
  Lemma \ref{OctonionicPauliMatrices} on the space $\mathbb{O}^2$:
  \begin{equation}
    \exp(2\pi {\pmb\tau}^1_L)\psi
    \;=\;
    {\pmb\Gamma}_6
    {\pmb\Gamma}_7
    {\pmb\Gamma}_8
    {\pmb\Gamma}_9
    \psi
    \;=\;
    \sigma_6(\sigma_7(\sigma_8(\sigma_9\psi)))
    \;=\;
    \mathrm{e}_4(\mathrm{e}_5(\mathrm{e}_6(\mathrm{e}_7 \psi))
    \,.
  \end{equation}
  Now writing
  $\psi = \eta_1+\eta_2\ell$ with $\eta_1,\eta_2\in \mathbb{H}^2$,
  Lemma \ref{ConsecutiveLeftMultiplicationBye4e5e6e7} gives
  \[
     \exp(2\pi {\pmb\tau}^1_L)(\eta_1+\eta_2\ell)
     \;=\;
     \eta_1-\eta_2\ell
     \,.
  \]
  Therefore
  $\exp(2\pi {\pmb\tau}^L)\psi = \psi$ is equivalent to $\psi = \eta_1$.
\hfill \end{proof}

\begin{lemma}[Clifford action of $\mathfrak{su}(2)_R$ is by quaternion right action]
  \label{CliffordActionOfSU2LIsByQuaternionRightAction}
  Under the identification $\mathbf{16}^{\mathbb{Z}^A_2}\cong \mathbf{8}\cong\mathbb{H}^2$  from Lemma \ref{lemma:fixed-locus},
  the action of ${\pmb\tau}^i_R \in \mathfrak{su}(2)_R$ \eqref{su2GeneratorsInTermsOfGammaMatrices}
  is given by right multiplication with half the quaternion unit $\mathrm{e}_i$:
  \begin{equation}
    {\pmb\tau}^i_L \psi = \psi
    \;\;\;\;\;\;
      \Rightarrow
    \;\;\;\;\;\;
    {\pmb\tau}^i_R \psi
    \;=\;
    \tfrac{1}{2}
    \psi \mathrm{e}_i
    \phantom{AAA}
    \in
    \mathbb{H}^2.
  \end{equation}
\end{lemma}
\begin{proof}
Using \eqref{HalfSpin3Rotation}, \eqref{SpinRepsViaSL2O}
and \eqref{InducedRelations}
we compute as follows:
\begin{align*}
  {\pmb\tau}^{1}_R\psi
  &=
  \big(
    \tfrac{1}{4}
    {\pmb\Gamma}_6{\pmb\Gamma}_7
    -
    \tfrac{1}{4}
    {\pmb\Gamma}_8{\pmb\Gamma}_9
  \big)
  \psi
  \;=\;
  \tfrac{1}{4}{\pmb\Gamma}_6 {\pmb\Gamma} _{7} \psi
  -
  \tfrac{1}{4}{\pmb\Gamma}_{8}{\pmb\Gamma}_9
  {\pmb\Gamma}_6{\pmb\Gamma}_7{\pmb\Gamma}_8{\pmb\Gamma}_9
  \;=\;
  \tfrac{1}{4}{\pmb\Gamma}_6{\pmb\Gamma}_7
  +
  \tfrac{1}{4}{\pmb\Gamma}_6{\pmb\Gamma}_7
  \\
  & =
  \tfrac{1}{2}\sigma_6(\sigma_7\psi)
  \;=\;
  \tfrac{1}{2}\mathrm{e}_4(\mathrm{e}_5\psi)
  \;=\;
  \tfrac{1}{2}\ell(({\rm i}\ell)\psi)
  \;=\;
  \tfrac{1}{2}\ell(({\rm i}\psi^\ast)\ell)
  \;=\;
  \tfrac{1}{2}\psi {\rm i}
  \;=\;
  \tfrac{1}{2} \psi {\rm e}_1\;.
\end{align*}
A directly analogous computation shows the statement in the other
cases.
\hfill \end{proof}

\begin{lemma}[The ${\bf 8}$ of ${\rm Spin}(5, 1)$ as representation
 of ${\rm Spin}(3, 1) \times {\rm SU}(2)$]
  \label{The8OfSpin51As2C2CofSpin31TimesSU2}
  When regarded as a
  $\mathrm{Spin}(3,1)\times \mathrm{SU}(2)_R$ representation,
  along \eqref{RepresentationTheoreticQuestion},
  the $\mathbf{8}$ and $\overline{\mathbf{8}}$ of
  \eqref{FixedLociOfZ2AIn16}
  are isomorphic, as real representations,
  to the outer tensor product
  \eqref{OuterTensorProductOfRepresentations}
   of the left and right complex 2-dimensional Weyl representation
   $\mathbf{2}_{\mathbb{C}}, \overline{\mathbf{2}}_{\mathbb{C}}$ of $\mathrm{Spin}(3,1)$ \eqref{ComplexWeylRepresentations}
   with
   the fundamental complex representation
   $\mathbf{2}_{\mathbb{C}}$ of $\mathrm{SU}(2)$
   \eqref{FundamentalRepresentationofSU2}:
  \begin{equation}
   \mathbf{8}
   \;=\;
   \mathbf{2}_{\mathbb{C}}
   \underset{\mathbb{C}}{\boxtimes}
   \mathbf{2}_{\mathbb{C}}
   \,,
   \phantom{AAA}
   \overline{\mathbf{8}}
   \;=\;
   \overline{\mathbf{2}}_{\mathbb{C}}
   \underset{\mathbb{C}}{\boxtimes}
   \mathbf{2}_{\mathbb{C}}
   \;\;\;\;\;\;\;
   \in
   \mathrm{Rep}_{\mathbb{R}}
   \big(
     \mathrm{Spin}(3,1)
     \times
     \mathrm{SU}(2)_R
   \big).
 \end{equation}
\end{lemma}
\begin{proof}
    From Remark \ref{ComplexWeylSpinorRepresentations},
    and in view of the chain of inclusions
    \eqref{LinearSubspacesOfTheOctonions}, one manifestly sees that
    the restricted representations without the $\mathrm{SU}(2)_R$-action
    considered are
  \begin{equation}
    \label{BranchingOf8OfSpin51ToSpin31}
    \xymatrix@R=1pt{
      \mathrm{Spin}(5,1)
      \ar@{<-^{)}}[rr]^-{i}
      &&
      \; \mathrm{Spin}(3,1)
      \\
      \mathbf{8}
      \ar@{|->}[rr]
      \ar@{}[dd]|-{
        \scalebox{.9}{
          \begin{rotate}{-90}
            $\mathclap{\simeq_{{\scalebox{.6}{$\mathbb{R}$}}}}$
          \end{rotate}
        }
      }
      &&
      \mathbf{2}_{\mathbb{C}}
      \oplus
      \mathbf{2}_{\mathbb{C}}
      \ar@{}[dd]|-{
        \scalebox{.9}{
          \begin{rotate}{-90}
            $\mathclap{\simeq_{{\scalebox{.6}{$\mathbb{R}$}}}}$
          \end{rotate}
        }
      }
      \\
      \\
      \mathbb{H}^2
      &&
      \mathbb{C}^2
      \oplus
      \mathbb{C}^2 \mathrm{j}
    }
    \;
    \phantom{AAAAA}
    \;
    \xymatrix@R=1pt{
      \mathrm{Spin}(5,1)
      \ar@{<-^{)}}[rr]^-{i}
      &&
      \; \mathrm{Spin}(3,1)\;.
      \\
      \overline{\mathbf{8}}
      \ar@{|->}[rr]
      \ar@{}[dd]|-{
        \scalebox{.9}{
          \begin{rotate}{-90}
            $\mathclap{\simeq_{{\scalebox{.6}{$\mathbb{R}$}}}}$
          \end{rotate}
        }
      }
      &&
      \overline{\mathbf{2}}_{\mathbb{C}}
      \oplus
      \overline{\mathbf{2}}_{\mathbb{C}}
      \ar@{}[dd]|-{
        \scalebox{.9}{
          \begin{rotate}{-90}
            $\mathclap{\simeq_{{\scalebox{.6}{$\mathbb{R}$}}}}$
          \end{rotate}
        }
      }
      \\
      \\
      \mathbb{H}^2
      &&
      \mathbb{C}^2
      \oplus
      \mathbb{C}^2 \mathrm{j}
    }
  \end{equation}
  So it remains to show that,
  with respect to the $\mathrm{SU}(2)_R$-action we have
  $\mathbf{2}_{\mathbb{C}}\cdot 2 \simeq^{}_{\mathbb{R}}
  \mathbf{2}_{\mathbb{C}} \underset{\mathbb{C}}{\boxtimes} \mathbf{2}_{\mathbb{C}}$.
By Lemma \ref{CliffordActionOfSU2LIsByQuaternionRightAction},
the action of $\mathrm{SU}(2)_R$ on spinors
$
   \mathbb{C}^2\oplus \mathbb{C}^2\ni
   [\begin{matrix} \psi^+ & \psi^- \end{matrix}]
   \;\;\longleftrightarrow\;\;
   \psi^+
     +
   \psi^- \mathrm{j} \in \mathbb{C}^2\oplus \mathbb{C}^2\mathrm{j}
$
  is by right multiplication by the unit quaternions, under the
  isomorphism $\mathrm{SU}(2) \simeq \mathrm{Sp}(1)$.
  Therefore, using the quaternion algebra,
  we check explicitly that this right multiplication
  gives the fundamental complex representation
  $\mathbf{2}_{\mathbb{C}} \in \mathrm{Rep}_{\mathbb{C}}(\mathrm{SU}(2))$
  via the Pauli matrix representation \eqref{PauliMatrices}:
  \begin{equation}
    \label{ExplicitSU2RActionOn4dSpinors}
  \hspace{-3mm}
    \begin{aligned}
        2 {\pmb\tau}^1_R
    [
      \!\!\!
      \begin{array}{cc}
        \psi & \phi
      \end{array}
      \!\!\!
    ]
       & =
    [
      \!\!\!
      \begin{array}{cc}
        \psi & \phi
      \end{array}
      \!\!\!
    ]
    \mathrm{e}_1
    \\
    & =
    \big(
      \psi + \phi \mathrm{j}
    \big)
    \mathrm{i}
    \\
    & =
    \psi \mathrm{i} - (\phi \mathrm{i}) \mathrm{j}
    \\
    &=
    [
      \!\!\!
      \begin{array}{cc}
        \psi \mathrm{i} & - \phi \mathrm{i}
      \end{array}
      \!\!\!
    ]
    \\
    & =
    [
      \!\!\!
      \begin{array}{cc}
        \psi & \phi
      \end{array}
      \!\!\!
    ]
    \cdot
    \left[
    \!\!\!
    \begin{array}{cc}
      \mathrm{i} & 0
      \\
      0 & -\mathrm{i}
    \end{array}
    \right],
    \end{aligned}
    \phantom{A}
    \begin{aligned}
     2{\pmb\tau}^2_R
    [
      \!\!\!
      \begin{array}{cc}
        \psi & \phi
      \end{array}
      \!\!\!
    ]
        & =
    [
      \!\!\!
      \begin{array}{cc}
        \psi & \phi
      \end{array}
      \!\!\!
    ]
    \mathrm{e}_2
    \\
    & =
    \big(
      \psi + \phi \mathrm{j}
    \big)
    \mathrm{j}
    \\
    & =
    \psi \mathrm{j} - \phi
    \\
    & =
    [
      \!\!\!
      \begin{array}{cc}
        - \phi & \psi
      \end{array}
      \!\!\!
    ]
    \\
    & =
    [
      \!\!\!
      \begin{array}{cc}
        \psi & \phi
      \end{array}
      \!\!\!
    ]
    \cdot
    \left[
    \!\!\!
    \begin{array}{cc}
      0 & 1
      \\
      -1 & 0
    \end{array}
    \right],
    \end{aligned}
        \phantom{A}
    \begin{aligned}
     2 {\pmb\tau}^3_R
    [
      \!\!\!
      \begin{array}{cc}
        \psi & \phi
      \end{array}
      \!\!\!
     & =
    [
      \!\!\!
      \begin{array}{cc}
        \psi & \phi
      \end{array}
      \!\!\!
    ]
    \mathrm{e}_3
    \\
    & =
    \big(
      \psi + \phi \mathrm{j}
    \big)
    \mathrm{k}
    \\
    & =
    \psi
      \underset{
        = \mathrm{i} \mathrm{j}
      }{
        \underbrace{
          \mathrm{k}
        }
      }
      +
    \phi \underset{= \mathrm{i}}{\underbrace{\mathrm{j} \mathrm{k}}}
    \\
    & =
    [
      \!\!\!
      \begin{array}{cc}
        \phi \mathrm{i} & \psi \mathrm{i}
      \end{array}
      \!\!\!
    ]
    \\
    & =
    [
      \!\!\!
      \begin{array}{cc}
        \psi & \phi
      \end{array}
      \!\!\!
    ]
    \cdot
    \left[
    \!\!\!
    \begin{array}{cc}
      0 & \mathrm{i}
      \\
      \mathrm{i} & 0
    \end{array}
    \right].
    \end{aligned}
  \end{equation}

  \vspace{-6mm}
  \hfill
\end{proof}

\begin{lemma}[Residual $U(1)_V$-action on spinors]
 \label{ResidualU1ActionOnSpinors}
  When regarded as a
  $\mathrm{Spin}(3,1)\times \mathrm{U}(1)_V$ representation,
  along \eqref{RepresentationTheoreticQuestion},
  the $\mathbf{8}$ and $\overline{\mathbf{8}}$ of
  \eqref{FixedLociOfZ2AIn16}
  are isomorphic, as real representations,
  to the outer tensor product
  \eqref{OuterTensorProductOfRepresentations}
  of the left and right complex 2-dimensional Weyl representation
  $\mathbf{2}_{\mathbb{C}}, \overline{\mathbf{2}}_{\mathbb{C}}$ of $\mathrm{Spin}(3,1)$ (see \eqref{ComplexWeylRepresentations})
  with the complex irrep $\mathbf{1}_{\mathbb{C}}^{\!-}$ of
  $\mathrm{U}(1)$ (see \eqref{RepresentationsOfU1}).
\end{lemma}

\begin{proof}
  We explicitly compute the action of
  $R_{4 5}$, via its Clifford representation
  $\tfrac{1}{2}{\pmb\Gamma}_4{\pmb\Gamma}_5$
  \eqref{LieSpinRepresentation}
  on a spinor
$
   \mathbb{C}^2\oplus \mathbb{C}^2\ni
   [\begin{matrix} \psi & \phi\end{matrix}]
   \;\;\longleftrightarrow\;\;
   \psi+\phi\mathrm{j} \in \mathbb{C}^2\oplus \mathbb{C}^2\mathrm{j}
$
using \eqref{KPauliActions}:
\begin{equation}
  \begin{aligned}
    2\tfrac{1}{2}{\pmb\Gamma}_4 {\pmb\Gamma}_5
    [
      \!\!\!
      \begin{array}{cc}
        \psi & \phi
      \end{array}
      \!\!\!
    ]
    &=
    \sigma_4(\sigma_5(\psi + \phi \mathrm{j}))
    \;=\;
    -
    \mathrm{e}_2(\mathrm{e}_3(\psi + \phi \mathrm{j})))
    \\
    & =
    -
    \mathrm{j}\mathrm{k}(\psi + \phi \mathrm{j}))
    \;=\;
    -
    \mathrm{i}(\psi + \phi \mathrm{j}))
    \;=\;
    (-
     \psi \mathrm{i})
    +
    (-
    \phi \mathrm{i}) \mathrm{j}
    \\
    & =
    [
      \!\!\!
      \begin{array}{cc}
        \psi & \phi
      \end{array}
      \!\!\!
    ]
    \cdot
    \left[
      \begin{array}{cc}
        - \mathrm{i} & 0
        \\
        0 & - \mathrm{i}
      \end{array}
    \right]
    \,,
  \end{aligned}
\end{equation}
where we also used that quaternionic multiplication is associative
and then that complex multiplication, furthermore, is  commutative.
\hfill \end{proof}


\medskip

In summary, we have proven the second statement of
Theorem \ref{FieldsAtSingularity}:

\begin{prop}[The reduced Spin representation]
  \label{TheReducedSpinRep}
  The fixed locus \eqref{FormingFixedPoints}
  under the action of $\mathbb{Z}_2^A \subset \mathrm{SU}(2)_L$  \eqref{Z2A}
  in the $\mathrm{Spin}(10,1)$ representation $\mathbf{32}$
  (from \eqref{RealSpinRepsInDimPlus1}),
  regarded with its residual action
  of $\mathrm{Spin}(3,1) \times \mathrm{U}(1) \times \mathrm{SU}(2)_R$ is:
  \begin{equation}
    \label{4dSpinRepFrom11d}
    \xymatrix@R=-2pt{
      \mathrm{Spin}(10,1)
      \ar@{<-^{)}}[rr]
      &&
     \; \mathrm{Spin}(3,1) \times \mathrm{U}(1)_V \times \mathrm{SU}(2)_R\;.
      \\
      \mathbf{32}
      \ar@{|->}[rr]
      &&
      \mathbf{32}^{\mathbb{Z}_2^A}
      =
      \big(
        \mathbf{2}_{\mathbb{Z}}
        \oplus
        \overline{\mathbf{2}}_{\mathbb{Z}}
      \big)
      \underset{\mathbb{C}}{\boxtimes}
      \mathbf{1}^{\!-}_{\mathbb{C}}
      \underset{\mathbb{C}}{\boxtimes}
      \mathbf{2}_{\mathbb{C}}
    }
  \end{equation}
\end{prop}
\begin{proof}
  This follows immediately from combining
  Lemma \ref{The8OfSpin51As2C2CofSpin31TimesSU2}
  and
  Lemma \ref{ResidualU1ActionOnSpinors}.
\hfill \end{proof}

For the record, and as this is somewhat subtle, we highlight the following equivalent and non-equivalent versions
of the statement in \eqref{4dSpinRepFrom11d}:

\begin{lemma}[Real isomorphism classes of chiral charged 4d spinor reps]
  \label{IsoclassesOfTheFullReps}
   Consider the real representations of
  $\mathrm{Spin}(3,1) \times \mathrm{U}(1)_V \times \mathrm{SU}(2)_R$.

  \vspace{-1mm}
 \item  {\bf (i)} There are isomorphisms of real representations as follows:
  \begin{equation}
    \label{RealIsomorphismOfChiralChargedSpinorReps}
    \begin{aligned}
      \mathbf{2}_{\mathbb{C}}
        \underset{\mathbb{C}}{\boxtimes}
      \mathbf{1}^\mp_{\mathbb{C}}
        \underset{\mathbb{C}}{\boxtimes}
      \mathbf{2}_{\mathbb{C}}
      \;\simeq_{{}_{\mathbb{R}}}\;
      \overline{\mathbf{2}}_{\mathbb{C}}
        \underset{\mathbb{C}}{\boxtimes}
      \mathbf{1}^\pm_{\mathbb{C}}
        \underset{\mathbb{C}}{\boxtimes}
      \mathbf{2}_{\mathbb{C}}\;.
    \end{aligned}
  \end{equation}

   \vspace{-4mm}
 \item {\bf (ii)}  There are, however, \emph{no} isomorphisms,
  not even as real representations, between
  $
      \mathbf{2}_{\mathbb{C}}
        \underset{\mathbb{C}}{\boxtimes}
      \mathbf{1}^-_{\mathbb{C}}
        \underset{\mathbb{C}}{\boxtimes}
      \mathbf{2}_{\mathbb{C}}
  $
  and
  $
      \mathbf{2}_{\mathbb{C}}
        \underset{\mathbb{C}}{\boxtimes}
      \mathbf{1}^+_{\mathbb{C}}
        \underset{\mathbb{C}}{\boxtimes}
      \mathbf{2}_{\mathbb{C}}
  $
  nor between
  $
      \overline{\mathbf{2}}_{\mathbb{C}}
        \underset{\mathbb{C}}{\boxtimes}
      \mathbf{1}^-_{\mathbb{C}}
        \underset{\mathbb{C}}{\boxtimes}
      \mathbf{2}_{\mathbb{C}}
  $
  and
  $
      \overline{\mathbf{2}}_{\mathbb{C}}
        \underset{\mathbb{C}}{\boxtimes}
      \mathbf{1}^+_{\mathbb{C}}
        \underset{\mathbb{C}}{\boxtimes}
      \mathbf{2}_{\mathbb{C}}
  $.
\end{lemma}
\begin{proof}
  Notice that under the canonical embedding of $\mathrm{SU}(2)$
  as the spatial $\mathrm{Spin}(3)$-group inside the Lorentzian
  $\mathrm{Spin}(3,1)$, the two complex Weyl spinor representation
  \eqref{ComplexWeylRepresentations}
  both restrict to the fundamental representation of $\mathrm{SU}(2)$
from  \eqref{FundamentalRepresentationofSU2}:
  $$
    \xymatrix@R=-2pt{
      \mathrm{Spin}(3,1)
       \ar@{<-^{)}}[r]
      &
      \;
      \mathrm{Spin}(3)
       \ar@{}[r]|-{\simeq}
      &
      \mathrm{SU}(2)\;.
      \\
      \mathbf{2}_{\mathbb{C}}
      \ar@{|->}[rr]
      &&
      \mathbf{2}_{\mathbb{C}}
      \\
      \overline{\mathbf{2}}_{\mathbb{C}}
      \ar@{|->}[rr]
      &&
      \mathbf{2}_{\mathbb{C}}
    }
  $$
  Now consider the identification of complex representations
  $$
    \mathbf{2}_{\mathbb{C}}
      \underset{\mathbb{C}}{\boxtimes}
    \mathbf{1}_{\mathbb{C}}
      \underset{\mathbb{C}}{\boxtimes}
    \mathbf{2}_{\mathbb{C}}
    \;\simeq_{{}_{\mathbb{C}}}\;
    \mathrm{Mat}(2 \times 2, \mathbb{C})
    \,,
  $$
  where on the right the $\mathrm{Spin}(3,1)$-action
  is by left matrix multiplication with $\mathbb{C}$-Pauli matrices
  \eqref{HeteroticSpinRepsAsOctonionicModules}
  $\tfrac{1}{2}\sigma_{a_1} \sigma_{a_2}$ for $a_i \in \{0,1,2,3\}$,
  while the $\mathrm{SU}(2)_R$-action is by right matrix multiplication
  with transposed Pauli matrices
  ${}^{t}\big(\tfrac{1}{2}\sigma_{a_1} \sigma^{a_2}\big)$
  for $a_i \in \{1,2,3\}$.
  Therefore, by the proof of Lemma \ref{RealReps4AndBar4OfSpin31AreIsomorphic},
  we have an isomorphism of real representations given by
  \begin{equation}
    \label{AnIsomorphism}
    \xymatrix@R=-1pt{
      \mathbf{2}_{\mathbb{C}}
        \underset{\mathbb{C}}{\boxtimes}
      \mathbf{1}_{\mathbb{C}}
        \underset{\mathbb{C}}{\boxtimes}
      \mathbf{2}_{\mathbb{C}}
      \ar[rr]
      &&
      \overline{\mathbf{2}}_{\mathbb{C}}
        \underset{\mathbb{C}}{\boxtimes}
      \mathbf{1}_{\mathbb{C}}
        \underset{\mathbb{C}}{\boxtimes}
      \mathbf{2}_{\mathbb{C}}\;.
      \\
      \mathrm{Mat}(2 \times 2, \mathbb{C})
      \ar[rr]
      &&
      \mathrm{Mat}(2 \times 2, \mathbb{C})
      \\
      A \ar@{|->}[rr]
      &&
      \epsilon\cdot A^\ast \cdot {}^{t}\epsilon
    }
  \end{equation}
  Since the real-linear map in the last row is complex anti-linear,
  this same real-linear map gives the claimed isomorphisms
  \eqref{RealIsomorphismOfChiralChargedSpinorReps}.
  Moreover, since $\mathbf{2}_{\mathbb{C}}$ is not isomorphic
  to $\overline{\mathbf{2}}_C$ as a complex representation,
  no real isomorphism
  as in the first line of \eqref{AnIsomorphism}
  can be complex-linear, hence all such isomorphisms must be
  complex anti-linear. This gives the second claim.
\hfill \end{proof}

\medskip

In closing this discussion of octonionic spinor algebra,
we notice the following identifications, which will be useful
in \cref{TheWZWTermsInTheLagrangian} below:

\begin{notation}[4d chirality opearator]
  \label{4dChiralityOperator}
  If we introduce for the 4d sector of the
  octonionic Dirac matrices \eqref{OctonionicDiracMatrices}
  the following standard notation for complex Dirac matrices
  \begin{equation}
    \label{4dgamma5}
    \begin{aligned}
      \gamma_\mu
      &
      :=
      \;
      {\pmb\Gamma}_\mu
      \phantom{AAAA}
      \mu \in \{0,1,2,3\}
      \\
      \gamma_5
      & :=
      \;
      \mathrm{i} \gamma_0 \gamma_1 \gamma_2 \gamma_3
      \\
      \mathllap{
        \mbox{then the chirality operator $\gamma_5$ is}
        \phantom{AAAAAAAAA}
      }
      \\
      \gamma_5
      &
      =
      \;
      {\pmb\Gamma}_{5'}
      \\
      \mathllap{\mathrm{i}}
      \gamma_5
      & =
      \;
      {\pmb\Gamma}_{4}
      {\pmb\Gamma}_{5}
      {\pmb\Gamma}_{5'}\;.
    \end{aligned}
  \end{equation}
\end{notation}

\medskip
\begin{lemma}[Branching of spinor pairing]
\label{BranchingOfSpinorPairing}
The pairing in Prop. \ref{SpinorPairingViaKDiracMatrices}
applied to 6d spinors and expressed in terms of their
branching, from Prop. \ref{The8OfSpin51As2C2CofSpin31TimesSU2},
into iso-doublets of 4d spinors
$$
  \xymatrix@R=5pt@C=2pt{
    \big[
      \!\!\!
      {\begin{array}{cc}
        \psi^+ & \psi^-
      \end{array}}
      \!\!\!
    \big]
    \ar@{=}[d]
    &
    \in
    &
    (\mathbf{2}_{\mathbb{C}} \oplus \overline{\mathbf{2}}_{\mathbb{C}})
    \underset{
      \mathbb{C}
    }{\boxtimes}
    \mathbf{2}_{\mathbb{C}}
    &\in&
    \mathrm{Rep}_{\mathbb{R}}
    \big(
      \mathrm{Spin}(3,1) \times \mathrm{SU}(2)_R
    \big)
    \\
    \mathllap{
      \psi =
      \;
    }
    \psi^+ + \psi^- \mathrm{j}
    &\in&
    \mathbf{8} \oplus \overline{\mathbf{8}}
    &\in&
    \mathrm{Rep}_{\mathbb{R}}
    \big(
      \mathrm{Spin}(5,1)
    \big)
  }
$$
equals the sum over iso-doublet components of the pairing
in Prop. \ref{SpinorPairingViaKDiracMatrices} applied to 4d spinors:
$$
  \begin{aligned}
  \big\langle
  \left[
    \!\!\!
    \begin{array}{cc}
      \psi^+ & \psi^-
    \end{array}
    \!\!\!
  \right]
  ,
  \left[
    \!\!\!
    \begin{array}{cc}
      \phi^+ & \phi^-
    \end{array}
    \!\!\!
  \right]
  \big\rangle
  &
  :=
  \big\langle
    \psi, \phi
  \big\rangle
  \\
  & =
  \big\langle
    \psi^+, \phi^+
  \big\rangle
  +
  \big\langle
    \psi^-, \phi^-
  \big\rangle\;.
  \end{aligned}
$$
\end{lemma}
\begin{proof}
Noticing that
$$
  \big(
    c \in \mathbb{C} \hookrightarrow \mathbb{O}
  \big)
  \;\;\Rightarrow\;\;
  \left\{
  \begin{array}{l}
    (c \mathrm{j})^\ast
    \;=\;
    - c \mathrm{j}
    \\
    - \mathrm{j} c d^\ast \mathrm{j}
    \;=\;
    c^\ast d
  \end{array}
  \right.
  \;\;
    \Rightarrow
  \;\;
  \left\{
  \begin{array}{l}
    \mathrm{Re}
    \big(
      c \mathrm{j}
    \big)
    \;=\;
    0
    \\
    \mathrm{Re}
    \big(
      -
      \mathrm{j}
      c d^\ast
      \mathrm{j}
    \big)
    \;=\;
    \mathrm{Re}
    \big(
      c d^\ast
    \big)
  \end{array}
  \right.
$$
we compute as follows:
$$
  \begin{aligned}
    \big\langle
      \psi, \phi
    \big\rangle
    &
    =
    \big\langle
      \psi^+ + \phi^- \mathrm{j},
      \psi^+ + \phi^- \mathrm{j}
    \big\rangle
    \\
    & =
    \mathrm{Re}
    \Big(
    \big(
      \psi^+ + \psi^- \mathrm{j}
    \big)^\dagger
    \cdot
    {\pmb\Gamma}_0
    \cdot
    \big(
      \phi^+ + \phi^- \mathrm{j}
    \big)
    \Big)
    \\
    & =
    \mathrm{Re}
    \Big(
    \big(
      (\psi^+)^\dagger - \mathrm{j} ((\psi^-)^\dagger)
    \big)
    \cdot
    {\pmb\Gamma}_0
    \cdot
    \big(
      \phi^+ + \phi^- \mathrm{j}
    \big)
    \Big)
    \\
    & =
    \mathrm{Re}
    \big(
      (\psi^+)^\dagger
      \cdot
      {\pmb\Gamma}_0
      \cdot
      \phi^+
    \big)
      \;+\;
    \mathrm{Re}
    \big(
      (\psi^-)^\dagger
      \cdot
      {\pmb\Gamma}_0
      \cdot
      \phi^-
    \big)
    \\
    & =
    \big\langle
      \psi^+, \phi^+
    \big\rangle
    +
    \big\langle
      \psi^-, \phi^-
    \big\rangle\;.
  \end{aligned}
$$

\vspace{-5mm}
\hfill
\end{proof}

\section{Analysis of the WZW-type terms}
\label{TheWZWTermsInTheLagrangian}

With the full field content \eqref{TheFieldContent} in hand,
via the proof in \cref{FieldContent},
here we spell out two
noteworthy contributions to the super-exceptional PS-Lagrangian
\eqref{SuperExceptionalPSLagrangianSummands} that both originate
from the super-exceptional WZW term \eqref{SuperExceptionalWZWTerm},
and we comment on the relation of both to quantum hadrodynamics.
A more comprehensive
discussion of the emergent $\mathrm{SU}(2)$-flavor gauge theory
encoded by the Lagrangian \eqref{SuperExceptionalPSLagrangianSummands}
is beyond the scope of this article, but see \cref{ConcludingRemarks}
for some outlook.

\medskip

We need the following:

\begin{notation}[Radial holographic Kaluza-Klein mode expansion]
  In the following, we consider
  a choice of sequence of smooth functions of the coordinate $x^4$
  on the M5-brane worldvolume \eqref{SuperExceptionalFieldConfiguration},
  \begin{equation}
    \label{TheKKModes}
      h_{(n)} \;:\; x^4 \;\mapsto\; h_{(n)}(x^4)  \;\in\; \mathbb{R}
      \,,
      \phantom{AAAA}
      n \in \mathbb{N}
  \end{equation}
  to be regarded as KK-modes for flat space holography.
  These functions \eqref{TheKKModes}
  may, for instance, be any of the following:
  \begin{itemize}
  \vspace{-2mm}
  \item {\bf (Flat hard-wall holography)}: Trigonometric functions
  of $x^4$
  defined on a closed interval
  and vanishing at the boundary points, as considered in
  \cite[5.1]{SonStephanov03}.

  \vspace{-2mm}
  \item {\bf (Atiyah-Manton holography)}: Hermite functions
  defined on all $x^4 \in \mathbb{R}$
  as considered in \cite[3]{Sutcliffe10}\cite[3]{Sutcliffe15}.

  \vspace{-2mm}
  \item {\bf (Circle-compactification of M5)}: Periodic functions
    of $x^4$ for some given radius of periodicity, as
    for the KK-modes in the model considered in \cite{IvanovaLechtenfeldPopov18}.
  \end{itemize}

    \vspace{-2mm}
\noindent  With such a choice understood,
  let  $V$ be a function or differential form
  on the M5-worldvolume \eqref{SuperExceptionalFieldConfiguration} that is
  pulled back from ${\mathbb{R}^{3,1}}$. Then
  we have
  \begin{equation}
    \label{KKExpansion}
    V(\{x^\mu\}, x^4)
      \;=\;
    V^{(n)}(\{x^\mu\}) h_{(n)}(x^4)
    \;\in\;
    \Omega^\bullet(\mathbb{R}^{3,1})
    \longhookrightarrow
    \Omega^\bullet(\mathbb{R}^5,1)
    \,,
  \end{equation}
  where a sum over $n \in \mathbb{N}$ left notationally implicit,
  for the corresponding mode expansion (assumed to exist),
  with the $V^{(n)}$ being functions or differential forms
  just on $\mathbb{R}^{3,1}$.
The details of the super-exceptional M5-brane model will depend
  on the choice of this mode expansion \eqref{KKExpansion},
  its assumed completeness
  relations and boundary conditions, etc. However, for the present purpose
  of identifying just the general form of a few main terms
  in the super-exceptional
  Lagrangian \eqref{SuperExceptionalPSLagrangianSummands},
  we only need to assume that any expansion as in \eqref{KKExpansion}
  has been chosen.

\medskip
  The point of this, in the following,
  is that acting upon the expansion \eqref{KKExpansion} by
  the de Rham differential in 5d decomposes the exceptional
  vielbein field values $d (V_{\cdots})$
  $$
    \begin{aligned}
      d_{{}_{5}} V
      & =
      d_{{}_{5}} ( V^{(n)} h_{(n)} )
      \\
      & =
        h_{(n)}
        \,
        d_{{}_{4}} V^{(n)}
        \;+\;
        h'_{(n)}
        \,
        V^{(n)} \wedge d x^4
    \end{aligned}
  $$
  with $d_{{}_4} V^{(n)}$ being the de Rham differential form
  on $\mathbb{R}^{3,1}$, and $h'_{(n)}$ denoting the derivative
  of the mode function $h_{(n)}$
  (with respect to its only variable $x^4$).
  As a result, every summand in the super-exceptional Lagrangian
  \eqref{EvaluationOfTheLagrangian} is a wedge product of fields
  with precisely one of the factor fields not differentiated itself,
  but instead contributing with the differential $h'_{(n)} d x^4$
  of its mode function.
\end{notation}

\begin{remark}[Role of KK-mode expansion]
 \label{RoleOfKKModeExpansion}
 We highlight the following:

 \vspace{-1mm}
\item {\bf (i)}  A mode expansion as in \eqref{KKExpansion}
  also governs the field content
  in the Witten-Sakai-Sugimoto model
  \cite[(3.20)]{SakaiSugimoto04}\cite[(2.10)]{SakaiSugimoto05},
  reviewed in \cite{Sutcliffe10}\cite{BolognesiSutcliffe13}\cite{Sutcliffe15}.

 \vspace{-1mm}
\item {\bf (ii)}   There it is crucial,
  for the details of the meson mass spectrum (highlighted in \cite[(53)]{Sutcliffe15}),
  that the geometry is curved (specifically: asymptotically AdS).
  Here we are concerned with
  flat geometry
  (or at least: parallelizable, see Remark \ref{SkyrmeBaryonCurrent} below),
  since the whole complexity of
  the super-exceptional M5-brane model
  in \cref{SuperExceptionalMGeometry} emerges \cite{FSS19d} from
  the super M-brane cocycle in rational Cohomotopy theory
  on flat $D = 11$ super Minkowski spacetime \cite{FSS19a}.

 \vspace{-1mm}
\item {\bf (iii)}    Heuristically, this matches the idea that our model
  should pertain to a \emph{single} heterotic M5-brane
  ($N = 1$, deep M-theory regime),
  in contrast to a stack of a large number
  of coincident branes backreacting on their ambient spacetime
  geometry ($N \gg 1$, supergravity regime) as considered in
  the Witten-Sakai-Sugimoto model.
  Alternatively the flat super-exceptional spacetime
  ought to be regarded as the local tangent frame of a
  super-exceptional curved Cartan geometry \cite[p. 7s]{HSS18},
  with only the lowest modes relevant locally.

 \vspace{-1mm}
 \item {\bf (iv)}   In any case, mode expansions \eqref{KKExpansion}
  in flat space holography have been considered in
  \cite[5.1]{SonStephanov03}\cite[3]{Sutcliffe10}\cite[3]{Sutcliffe15}
  and suffice for the following purpose of analyzing the
  form of the interaction terms that appear.
\end{remark}

\subsection{Vector meson coupling to Skyrme baryon current}
\label{VectorMesonCouplingToSkyrmeBaryonCurrent}

We show here that the contribution to the super-exceptional WZW term
\eqref{SuperExceptionalWZWTerm}
by those bosonic fields
that transforms as $\mathbf{3}$ of $\mathrm{SU}(2)_R$ in \eqref{TheFieldContent}
has the form  which, in quantum hadrodynamics,
is characteristic of the coupling of the neutral vector mesons
to the baryon current, with the baryons appearing as Skyrme solitons
in the pion field (see \cite{RhoEtAl16}).

\begin{remark}[Linear basis for $\mathbf{4} \wedge \mathbf{4}$]
By \eqref{LinearBasisFor4Wedge4} and Prop. \ref{3Inside4Wedge3}
we have on the $\mathbb{Z}_2^A$-fixed locus the relations
\begin{equation}
  \label{RelationBetweenInternalIndicesOnFixedLocus}
  \left.
  \begin{aligned}
    e^{I 6}
    & = \epsilon^{I J K} e_{J K}
    \\
    e_{a_1 a_2 a_3}{}^{I 6}
    & = \epsilon^{I J K}
    e_{a_1 a_2 a_2 J K}
  \end{aligned}
  \right\}
  \;\;
  \in
  \mathrm{CE}
  \big(
    \mathbb{R}^{10,1\vert \mathbf{32}}_{\mathrm{ex}_s}
  \big)^{\mathbb{Z}_2^A}.
\end{equation}
\end{remark}

\begin{defn}[The $\pi$ field]
  \label{PiField}
For $\sigma$ a super-exceptional field configuration \eqref{SuperExceptionalFieldConfiguration},
write
\begin{equation}
  \label{PionFieldComponent}
  d \pi^I
  \;:=\;
  \sigma^\ast
  \big(
    e^{I 6}
  \big)
  =
  \sigma^\ast
  \Big(
    \tfrac{1}{2}
    \big(
      e^{I 6}
      +
      \epsilon^{I J K}
      e_{J K}
    \big)
  \Big)
\end{equation}
for its component corresponding to the M2-brane wrapping modes
on the vanishing 2-cycle in the $\mathrm{A}_1$-singularity.
On the right in \eqref{PionFieldComponent}
we have highlighted the equivalent
expressions, using \eqref{RelationBetweenInternalIndicesOnFixedLocus}.
\end{defn}

\begin{prop}[Emergence of isospin WZW term]
  \label{EmergenceOfTraceOverInternalIndices}
  Consider a super-exceptional sigma-model field \eqref{SuperExceptionalFieldConfiguration}
  with $\pi$-field component
  according to Def. \ref{PiField}.
  Then the value of the second factor in the super-exceptional WZW-term
  $\mathbf{L}_{\mathrm{ex}_s}^{\!\mathrm{WZW}}$ \eqref{SuperExceptionalWZWTerm},
  for indices ranging in $a_i \in \{6,7,8,9\}$, is:
  \begin{equation}
    \label{TheFirstDBICorrectionTerm}
    \sigma^\ast
    \big(
      e_{ \color{orangeii}  a_1 a_2}
      \wedge
      e^{ \color{orangeii} a_2 a_3 }
      \wedge
      e_{ \color{orangeii} a_3 }{}^{ \color{orangeii} a_1 }
    \big)
    \;\;=\;\;
    4
    \,
    \epsilon_{I J K}
    \;
      d \pi^I
      \wedge
      d \pi^J
      \wedge
      d \pi^K.
  \end{equation}
\end{prop}
\begin{proof}
We directly compute as follows:
\vspace{-2mm}
  $$
    \begin{aligned}
    &
    \sigma^\ast
    \big(
      e_{ \color{orangeii}  a_1 a_2}
      \wedge
      e^{ \color{orangeii} a_2 a_3 }
      \wedge
      e_{ \color{orangeii} a_3 }{}^{ \color{orangeii} a_1 }
    \big)
    \\
    & =
    \sigma^\ast
    \big(
      e_{ \color{orangeii}  I J}
      \wedge
      e^{ \color{orangeii} J K }
      \wedge
      e_{ \color{orangeii} K }{}^{ \color{orangeii} I }
      +
      3
      e_{ \color{orangeii}  6 I}
      \wedge
      e^{ \color{orangeii} I J }
      \wedge
      e_{ \color{orangeii} J }{}^{ \color{orangeii} 6 }
    \big)
    \\
    & =
    \epsilon_{ { \color{orangeii} I J } J'}
    \epsilon^{ { \color{orangeii} J K } K'}
    \epsilon_{ \color{orangeii} K}{}^{ { \color{orangeii} I }}{}_{ I' }
    d \pi^{J'}
    \wedge
    d \pi_{K'}
    \wedge
    d \pi^{I'}
    -
    3
    \,
    d \pi^{ \color{orangeii} I }
    \wedge
    \epsilon_{ { \color{orangeii} I J } J'}
    d \pi^{J'}
    \wedge
    d \pi^{ \color{orangeii} J }
    \\
    & =
    \epsilon_{ J'}{}^{ { K' }}{}_{ I' }
    d \pi^{J'}
    \wedge
    d \pi_{K'}
    \wedge
    d \pi^{I'}
    +
    3
    \,
    \epsilon_{ I J J' }
    d \pi^{ I }
    \wedge
    d \pi^{ J }
    \wedge
    d \pi^{ J' }
    \\
    & =
    4 \, \epsilon_{I J K }
    \,
    d \pi^{ I }
    \wedge
    d \pi^{ J }
    \wedge
    d \pi^{ K }
    ,
    \end{aligned}
  $$
  where in the second step we identified the $\pi$-field
  factors according to \eqref{PionFieldComponent}.
\hfill
\end{proof}

\begin{defn}[The $A$-field]
  \label{AFieldExpansion}
Let $A$ be a 1-form on $\mathbb{R}^{4,1}$
with mode expansion
$$
  A(\{x^\mu\}, x^4)
  \;=\;
  A^{(n)}(\{x^\mu\}) h_{(n)}(x^4)
$$
as in \eqref{KKExpansion}, and consider a super-exceptional
$\sigma$-model field \eqref{SuperExceptionalFieldConfiguration} with
\begin{equation}
  \sigma^\ast
  \big(
    e_{\mu 5'}
  \big)
  \;=\;
  d(A_\mu)
  \;=\;
  d
  \big(
    A^{(n)}_\mu
    h_{(n)}
  \big)
  \,.
\end{equation}
\end{defn}

\begin{example}[Photon-pion coupling term]
  \label{gammapipipiterm}
  Consider a super-exceptional sigma-model field
  \eqref{SuperExceptionalFieldConfiguration}
  with $\pi$-field component according to Def. \ref{PiField} and
  with $A$-field component according to Def. \ref{AFieldExpansion}.
Then the super-exceptional WZW-term \eqref{SuperExceptionalWZWTerm}
has, by Prop. \ref{EmergenceOfTraceOverInternalIndices}, the
contribution:
\begin{equation}
  \label{gammapipipi}
    \sigma^\ast
    \big(
      e_{
        { \color{darkblue} a }
        5'
      }
      \wedge
      e^{
        \color{darkblue} a
      }
      \;\wedge\;
      e_{ \color{orangeii}  a_1 a_2}
      \wedge
      e^{ \color{orangeii} a_2 a_3 }
      \wedge
      e_{ \color{orangeii} a_3 }{}^{ \color{orangeii} a_1 }
    \big)
    \;=\;
    4
    \,
    \underset{
      \mathclap{
      \mbox{
        \tiny
        \color{darkblue}
        \bf
        $\gamma \pi \pi \pi$ -- coupling
      }
      }
    }{
    \underbrace{
      A^{(0)}
        \wedge
      \epsilon_{I J K}
      d \pi^I \wedge d \pi^J \wedge d \pi^K
    }
    }
    \;\wedge\;
    h'_{(0)}
    d x^4 \wedge d x^{5'}
    \;+\;
    \mathcal{O}\big( h_{(n> 0)} \big).
\end{equation}
If  we interpret $A^{(0)}$ as the photon field
(in accordance with the super-exceptional Perry-Schwarz mechanism
\eqref{TheMaxwellLagrangian}) and $\pi^I$ with the pion field,
then expression \eqref{gammapipipi} below is the characteristic form
of the $\gamma \pi \pi \pi$-coupling term
\cite[(2.2)]{AvivZee72}\cite[(22)]{Witten83}\cite[(4)]{BDDL09}.
\end{example}

\begin{remark}[Skyrme baryon current]
 \label{SkyrmeBaryonCurrent}
We may consider also a curved but parallelizable geometry,
where the super-exceptional vielbein components
$e^{I 6}$ become a basis of three left-invariant 1-forms on
the group manifold $\mathrm{SU}(2)_R \simeq S^3$.
Then the field identification \eqref{PionFieldComponent} becomes
\begin{equation}
  \sigma^\ast
  \big(
    e^{I 6}
  \big)
  \;=\;
  e^{ -\vec \pi }
  d
  e^{ \vec \pi }
  \;=\;
  \big(
    e^{ \vec \pi }
  \big)^\ast
  \theta
\end{equation}
for
\begin{equation}
  e^{\vec \pi}
  \;:\;
  \mathbb{R}^{3,1}
  \longrightarrow
  \mathrm{SU}(2)_R
\end{equation}
an $\mathrm{SU}(2)_R$-valued function
and $\theta \in \Omega^1_{\mathrm{LI}}\big( \mathrm{SU}(2), \mathfrak{su}(2) \big)$
the Maurer-Cartan form on the group manifold, characterized by the
condition
$
  d\theta + [\theta \wedge \theta] = 0
 $.
In this situation, the proof of Prop. \ref{EmergenceOfTraceOverInternalIndices}
gives
\begin{equation}
  \label{SkyrmeBaryonCurrentExpression}
  \begin{aligned}
  \tfrac{1}{4}
  \sigma^\ast
  \big(
    e_{ \color{orangeii}  a_1 a_2}
    \wedge
    e^{ \color{orangeii} a_2 a_3 }
    \wedge
    e_{ \color{orangeii} a_3 }{}^{ \color{orangeii} a_1 }
  \big)
  & =
  \overset{
    \mathclap{
    \mbox{
      \tiny
      \color{darkblue} \bf
      Skyrme baryon current
    }
    }
  }{
  \overbrace{
  \mathrm{Tr}
  \big(
    e^{ -\vec \pi }
    d
    e^{ \vec \pi }
    \wedge
    e^{ -\vec \pi }
    d
    e^{ \vec \pi }
    \wedge
    e^{ -\vec \pi }
    d
    e^{ \vec \pi }
  \big)
  }
  }
  \\
  & =
  \big(
    e^{\vec \pi}
  \big)^\ast
  \mathrm{Tr}
  \big(
    \theta \wedge \theta \wedge \theta
  \big)
  \,,
  \end{aligned}
\end{equation}
where $\mathrm{Tr}(\theta \wedge \theta \wedge \theta)$
is proportional to the canonical volume form on the group
manifold $\mathrm{SU}(2)_R \simeq S^3$.
If we still interpret $\pi$ as the pion field, as in
Example \ref{gammapipipiterm}, then this is the characteristic form of the
baryon current
\cite[(6)]{GoldstoneWilczek81}\cite[(29)]{Witten83}\cite[(2)]{Witten83b}\cite[(11)]{AdkinsNappiWitten83}
in the Skyrme model of quantum hadrodynamics (review in \cite{RhoEtAl16}).
\end{remark}

\begin{example}[Photon-Skyrmion coupling]
  Accordingly, in the global case of Remark \ref{SkyrmeBaryonCurrent}
  the coupling in Example \ref{gammapipipiterm} becomes
\begin{equation}
  \label{gammapipipi}
  \hspace{-3mm}
    \sigma^\ast
    \big(
      e_{
        { \color{darkblue} a }
        5'
      }
      \wedge
      e^{
        \color{darkblue} a
      }
      \;\wedge\;
      e_{ \color{orangeii}  a_1 a_2}
      \wedge
      e^{ \color{orangeii} a_2 a_3 }
      \wedge
      e_{ \color{orangeii} a_3 }{}^{ \color{orangeii} a_1 }
    \big)
    \! =
    4
    \,
    \underset{
      \mathclap{
      \mbox{
        \tiny
        \color{darkblue}
        \bf
        $\gamma$ -- Skyrmion coupling
      }
      }
    }{
    \underbrace{
      A^{(0)}
        \wedge
  \mathrm{Tr}
  \big(
    e^{ -\vec \pi }
    d
    e^{ \vec \pi }
    \wedge
    e^{ -\vec \pi }
    d
    e^{ \vec \pi }
    \wedge
    e^{ -\vec \pi }
    d
    e^{ \vec \pi }
  \big)
    }
    }
    \wedge
    h'_{(0)}
    d x^4 \wedge d x^{5'}
    +
    \mathcal{O}\big( h_{(n> 0)} \big).
\end{equation}
This is the form of the photon-Skyrmion coupling
\cite[(19)]{Witten83}\cite[(13)]{KHKX85}.
\end{example}

\begin{defn}[Dual graviton]
  \label{DualGraviton}
  We shall say that a super-exceptional sigma-model field
  \eqref{RepresentationTheoreticQuestion} satisfies
  \emph{higher self-duality} if it takes the same value on the
  two copies of $\mathbf{7} \boxtimes \mathbf{1}$ in \eqref{TheFieldContent}.
  By the super-exceptional embedding condition
  \eqref{SuperExceptionalEmbeddingCondition}
  this means, in particular, that
  \begin{equation}
    \label{ValueOfemu6789}
    \sigma^\ast
    \big(
      e^\mu{}_{6 7 8 9}
    \big)
    \;=\;
    \sigma^\ast
    \big(
      e^\mu{}
    \big)
    \;=\;
    d x^\mu
    .
  \end{equation}
\end{defn}

\begin{defn}[The $\omega$-field]
  \label{omegaFieldExpansion}
Let $\omega$ be a 1-form on $\mathbb{R}^{4,1}$
with mode expansion
$$
  \omega(\{x^\mu\}, x^4)
  \;=\;
  \omega^{(n)}(\{x^\mu\}) h_{(n)}(x^4)
$$
as in \eqref{KKExpansion}, and consider a super-exceptional
$\sigma$-model field \eqref{SuperExceptionalFieldConfiguration} with
\begin{equation}
  \sigma^\ast
  \big(
    e_{ \mu_1 \mu_2 \mu_3}{}^{4 5}
  \big)
  \;=\;
  d
  \big(
    (\star_{{}_4} \omega)_{\mu_1 \mu_2 \mu_3}
  \big)
  \;=\;
  d
  \big(
    (\star_{{}_{4}}\omega^{(n)})_{\mu_1 \mu_2 \mu_3}
    h_{(n)}
  \big)
  \,.
\end{equation}
\end{defn}

\begin{example}[The $\omega$-meson couplings to pions and Skyrmions]
  \label{OmegaMesonCouplingToPionAndSkyrmion}
  Consider a super-exceptional sigma-model field
  \eqref{SuperExceptionalFieldConfiguration}
  satisfying higher self-duality in the sense of Def. \ref{DualGraviton}
  and
  with $\pi$-field component as in Def. \ref{PiField}
  and $\omega$-field component as in Def. \ref{omegaFieldExpansion}.
  Then the super-exceptional WZW-term \eqref{SuperExceptionalWZWTerm}
  has the contribution
  $$
    \begin{aligned}
    &
    \sigma^\ast
    \big(
      \epsilon^{
        { \color{orangeii} \mu_1 \mu_2 \mu_3 4 5  }
        5'
        { \color{greenii} \mu_4 6 7 8 9 }
      }
      e_{
        { \color{orangeii} \mu_1 \mu_2 \mu_3 4 5 }
      }
      \wedge
      e_{ \color{greenii}  \mu_4 6 7 8 9 }
      \;\wedge\;
      e_{ \color{orangeii}  a_1 a_2}
      \wedge
      e^{ \color{orangeii} a_2 a_3 }
      \wedge
      e_{ \color{orangeii} a_3 }{}^{ \color{orangeii} a_1 }
      \,\wedge\,
      e^{5'}
    \big)
    \\
    & =
      \epsilon^{
        { \color{orangeii} \mu_1 \mu_2 \mu_3 4 5  }
        { \color{greenii} \mu_4  }
      }
      (\star_{{}_4} \omega^{(0)})_{
        \color{orangeii} \mu_1 \mu_2 \mu_3
      }
      \wedge
      d x_{
        \color{greenii} \mu_4
      }
      \;\wedge\;
      \epsilon_{I J K}
      d \pi^I
      \wedge
      d \pi^J
      \wedge
      d \pi^K
      \;\wedge\;
      4
      h'_{{0}}
      d x^4 \wedge d x^{5'}
      +
      \mathcal{O}(h_{(n \gt 0)})
      \\
      & =
      \underset{
        \mathclap{
        \mbox{
          \tiny
          \color{darkblue} \bf
          $\omega \pi \pi \pi$ -- coupling
        }
        }
      }{
      \underbrace{
      \omega^{(0)}
      \,\wedge\,
      \epsilon_{I J K}
      \,
      d \pi^I
      \wedge
      d \pi^J
      \wedge
      d \pi^K
      }
      }
      \;\wedge\;
      h'_{{0}}
      d x^4 \wedge d x^{5'}
      +
      \mathcal{O}(h_{(n \gt 0)})
      \,.
    \end{aligned}
  $$
  If we interpret $\pi$ as the pion field,
  following Example \ref{gammapipipiterm},
  and in addition interpret $\omega^{(0)}$ as the
  $\omega$-meson field, then this is the characteristic form
  of the $\omega \pi \pi \pi$-coupling term
  \cite[(2)]{Rudaz84}\cite[(1)]{GT12}.
  More generally, in the global situation of
  Remark \ref{SkyrmeBaryonCurrent} this contribution becomes
  \begin{equation}
    \label{OmegaSkyrmionCoupling}
    \begin{aligned}
    &
    \sigma^\ast
    \big(
      \epsilon^{
        { \color{orangeii} \mu_1 \mu_2 \mu_3 4 5  }
        5'
        { \color{greenii} \mu_4 6 7 8 9 }
      }
      e_{
        { \color{orangeii} \mu_1 \mu_2 \mu_3 4 5 }
      }
      \wedge
      e_{ \color{greenii}  \mu_4 6 7 8 9 }
      \;\wedge\;
      e_{ \color{orangeii}  a_1 a_2}
      \wedge
      e^{ \color{orangeii} a_2 a_3 }
      \wedge
      e_{ \color{orangeii} a_3 }{}^{ \color{orangeii} a_1 }
      \,\wedge\,
      e^{5'}
    \big)
    \\
      & =
      \underset{
        \mathclap{
        \mbox{
          \tiny
          \color{darkblue} \bf
          $\omega$ -- Skyrmion coupling
        }
        }
      }{
      \underbrace{
      \omega^{(0)}
      \,\wedge\,
      \mathrm{Tr}
      \big(
        e^{- \vec \pi}
        d
        e^{\vec \pi}
        \wedge
        e^{- \vec \pi}
        d
        e^{\vec \pi}
        \wedge
        e^{- \vec \pi}
        d
        e^{\vec \pi}
      \big)
      }
      }
      \;\wedge\;
      4
      h'_{{0}}
      d x^4 \wedge d x^{5'}
      +
      \mathcal{O}(h_{(n \gt 0)})
      \,.
    \end{aligned}
  \end{equation}
This is the characteristic form of
$\omega$-meson coupling to the
Skyrme baryon current \eqref{SkyrmeBaryonCurrentExpression}
due to \cite[(2)]{AdkinsNappi84},
see also
\cite[(12)]{Kaiser01}\cite[(2.1)]{Hozwarth05}\cite[(2.1)]{GudnasonSpeight20}.
\end{example}

\begin{remark}[Chiral partner of the $\omega$-meson]
The difference in interpretation between the
photon-pion coupling in Example \ref{gammapipipiterm}
and the omega-pion coupling in Example \ref{OmegaMesonCouplingToPionAndSkyrmion}
is that, by the classification result of \eqref{TheFieldContent},
the vector field in the latter case is part of an
$\mathfrak{u}(2) \simeq \mathfrak{u}(1) \oplus \mathfrak{su}(2)$-multiplet
of vector fields whose $\mathfrak{su}(2)_R$-partner comes
from the super-exceptional vielbein component $e_{\mu_1 \mu_2 \mu_3}{}^{6 I}$.
\end{remark}

Therefore, we are led to interpreting this field component
as in the following definition and subsequent example.

\begin{defn}[The $\rho$-field]
  \label{rhoFieldExpansion}
Let $\rho$ be a $\mathfrak{su}(2)$-valued 1-form on $\mathbb{R}^{4,1}$
with mode expansion
$$
  \rho_I(\{x^\mu\}, x^4)
  \;=\;
  \rho^{(n)}_I(\{x^\mu\}) h_{(n)}(x^4)
$$
as in \eqref{KKExpansion}, and consider a super-exceptional
$\sigma$-model field \eqref{SuperExceptionalFieldConfiguration} with
\begin{equation}
  \sigma^\ast
  \big(
    e_{ \mu_1 \mu_2 \mu_3}{}^{6 I}
  \big)
  \;=\;
  d
  \big(
    (\star_{{}_4} \rho)^I_{\mu_1 \mu_2 \mu_3}
  \big)
  \;=\;
  d
  \big(
    (\star_{{}_{4}}(\rho^I)^{(n)})_{\mu_1 \mu_2 \mu_3}
    h_{(n)}
  \big)
  \,.
\end{equation}
\end{defn}
\begin{example}[The $\omega$/$\rho$-meson coupling to iso-doublet fermions]
  \label{OmegaRhoFermionicCoupling}
  Consider a super-exceptional sigma-model field
  \eqref{SuperExceptionalFieldConfiguration}
  with $\omega$-field component according to Def. \ref{PiField}
  and with $\rho$-field component according to Def. \ref{rhoFieldExpansion}.
  Then the fermionic coupling terms of these fields appearing in
  \eqref{SuperExceptionalBianchiIdentities} are of the form
\begin{equation}
  \label{omegaeRhoFermionicCouplingTerms}
  \raisebox{54pt}{
  \xymatrix@R=5pt{
    e_{a_1 \cdots a_5}
    {\pmb\Gamma}^{a_1 \cdots a_5}
        =\;
    \epsilon^{\mu_1 \mu_2 \mu_3 \mu_4}
    \mathrm{i}\gamma_5\gamma_{\mu_4}
    \big(
      e_{\mu_1 \mu_2 \mu_3}{}^{4 5}
        (-\mathrm{i})
        \;+\;
      e_{\mu_1 \mu_2 \mu_3}{}^{6 I}
        {\pmb\tau}_I
    \big)
    \ar@{|->}[d]^-{\sigma^\ast}
     \quad   +
    \cdots
    \\
        \mathrm{i}\gamma_5
    \underset{
      \mathclap{
      \mbox{
        \tiny
        \color{darkblue} \bf
        \begin{tabular}{c}
          $\omega$ / $\rho$ -- fermion coupling
        \end{tabular}
      }
      }
    }{
    \underbrace{
    \big(
      -
      \omega^{\mu} \mathrm{i}
      +
      \rho^\mu_I {\pmb\tau}^I
    \big)
    \gamma_\mu
    }
    }
    \quad +
    \cdots
  }
  }
\end{equation}

\vspace{-3mm}
\noindent where we identified the 4d spacetime Dirac matrices
$\gamma_\mu$, $\gamma_5$ by \eqref{4dgamma5}
and the isospin Pauli matrices ${\pmb\tau}^I$ according to
\eqref{su2GeneratorsInTermsOfGammaMatrices}.
This combination \eqref{omegaeRhoFermionicCouplingTerms} of couplings
to isodoublet fermions \eqref{4dSpinRepFrom11d}
is characteristic of that of the $\omega$- and $\rho$-mesons
in quantum hadrodynamics (\cite[(3.12)]{SerotWalecka92},
see also \cite[(24)]{KVW01}).
\end{example}

Finally we show one example of couplings involving fermions,
to illustrate how the left-invariance of the super-exceptional
vielbein, via the last line of \eqref{VielbeinInCanonicalCoordinateBasis},
makes the fermionic pairing terms in the super-exceptional Lagrangian
\eqref{SuperExceptionalPSLagrangianSummands} expand out to the
usual fermionic currents:

\begin{defn}[The $N$-field]
  \label{NFieldExpansion}
Let $N$ be a smooth function on $\mathbb{R}^{4,1}$
with values in Dirac spinor iso-doublets, and
with mode expansion
$$
  N(\{x^\mu\}, x^4)
  \;=\;
  N^{(n)}(\{x^\mu\}) h_{(n)}(x^4)
$$
as in \eqref{KKExpansion}, and consider a super-exceptional
$\sigma$-model field \eqref{SuperExceptionalFieldConfiguration} with
\begin{equation}
  \sigma^\ast
  \big(
    \psi
  \big)
  \;=\;
  d
  \big(
    N
  \big)
  \;=\;
  d
  \big(
    N^{(n)}
    h_{(n)}
  \big)
  \,,
\end{equation}
where the first equality uses the identification of
Theorem \ref{FieldsAtSingularity}.
\end{defn}

\begin{example}[fermion-pion coupling]
  \label{OneFermionPionCouplingTerm}
  Consider a super-exceptional sigma-model field
  \eqref{SuperExceptionalFieldConfiguration}
  with $\pi$-field component according to Def. \ref{PiField} and
  with $N$-field component according to \eqref{NFieldExpansion}.
  Then the super-exceptional WZW-term \eqref{SuperExceptionalWZWTerm}
  has, by Prop. \ref{EmergenceOfTraceOverInternalIndices}
  and using the last line of \eqref{VielbeinInCanonicalCoordinateBasis},
  the contribution:
 \begin{equation}
   \label{OneFermionPionCoupling}
  \begin{aligned}
   &
    \sigma^\ast
    \Big(
      \langle
        \eta
        \wedge
        {\pmb\Gamma}_{5'}
        \!\cdot\!
        \psi
      \rangle
      \wedge
      e_{
        \color{orangeii}
        a_1 a_2
      }
      \wedge
      e^{
        \color{orangeii}
        a_2 a_3
      }
      \wedge
      e_{
        \color{orangeii}
        a_3
      }
      {}^{
        \color{orangeii}
        a_1
      }
      \wedge e^{5'}
    \Big)
    \mathrlap{
      \phantom{AAAAAA}
      a_i \in \{6,7,8,9\}
    }
    \\
    & =
    (s + 1)
    \underset{
      \mathclap{
      \mbox{
        \tiny
        \color{darkblue}
        \bf
        fermion-pion coupling
      }
      }
    }{
    \underbrace{
    \langle
      N^{(0)},
      \gamma_{\mu} \gamma_{5}
      N^{(0)}
    \rangle
    d x^\mu
    \wedge
    \epsilon_{I J K}
    d \pi^I \wedge d \pi^J \wedge d \pi^K
    }
    }
    \;
    h'_{(0)}
    d x^4
    \wedge
    d x^{5'}
    \;+\;
    \cdots
  \end{aligned}
\end{equation}
Noticing that the spinor pairing over the brace
is the sum over iso-doublet components of the 4d spinor pairing
(by Lemma \ref{BranchingOfSpinorPairing}) this is the form
of spinor-pion coupling seen in
\cite[Table 3, item 30]{OllerVerbeniPrades06}\cite[Table 1, item 8]{FrinkMeissner06}
-- except for the appearance of the chirality operator $\gamma_5$,
via \eqref{4dChiralityOperator}.

Notice that in passing from the super-exceptional $\mathrm{MK}6$-locus
to the actual $\tfrac{1}{2}\mathrm{M5}$-locus
(Def. \ref{SuperExceptionalBraneLocus})
the right-handed spinors
get projected out anyway, so that the value of the chirality operator
becomes the identity.
\end{example}

\subsection{Non-linear electromagnetic coupling}
\label{DBICorrection}

We show here that the contribution to the super-exceptional WZW term
\eqref{SuperExceptionalWZWTerm} by the purely electromagnetic terms
\eqref{PhotonFieldComponent} and \eqref{DualPhotonField}
has the form of the first DBI-correction of non-linear electromagnetism,
and we discuss how this is compatible with the appearance of
hadrodynamics as above.

\begin{prop}[First Born-Infeld correction from super-exceptional WZW term]
\label{FirstBICorrection}
Consider a super-exceptional sigma-model field \eqref{SuperExceptionalFieldConfiguration}
with electromagnetic field component according to \eqref{PhotonFieldComponent}
and \eqref{DualPhotonField}. Then the contribution of the
super-exceptional WZW term \eqref{SuperExceptionalWZWTerm}
for indices ranging as $\alpha_i \in \{0,1,2,3, 5'\}$ is:
\begin{equation}
  \label{FirstBICorrectionEmerges}
  \sigma^\ast
  \big(
    e_{ { \color{darkblue} \alpha } 5' }
    \wedge
    e^{ \color{darkblue} \alpha }
    \wedge
    e_{ \color{orangeii} \alpha_1 \alpha_2 }
    \wedge
    e^{ \color{orangeii} \alpha_2 \alpha_3 }
    \wedge
    e_{ \color{orangeii} \alpha_3 }{}^{ \color{orangeii} \alpha_1 }
    \wedge
    e^{5'}
  \big)
  =
  \tfrac{3}{2}
  \left(
    \frac{
      F \wedge F
    }{
      \mathrm{dvol}_4
    }
  \right)^2
  \mathrm{dvol}_{6'}\;.
\end{equation}
\end{prop}
\begin{proof}
We calculate as follows:
\begin{equation}
  \begin{aligned}
  &
  \sigma^\ast
  \big(
    e_{ { \color{darkblue} \mu } 5' }
    \wedge
    e^{ \color{darkblue} \mu }
    \wedge
    e_{ \color{orangeii} \alpha_1 \alpha_2 }
    \wedge
    e^{ \color{orangeii} \alpha_2 \alpha_3 }
    \wedge
    e_{ \color{orangeii} \alpha_3 }{}^{ \color{orangeii} \alpha_1 }
    \wedge
    e^{5'}
  \big)
  \\
  & =
  \tfrac{3}{2}
  F_{
    \color{darkblue} \mu_1 \mu_2
  }
  d x^{ \color{darkblue} \mu_1 }
  \wedge
  d x^{ \color{darkblue} \mu_2 }
  \wedge
  F_{
    { \color{orangeii} \nu_1 }
    { \color{darkblue} \mu_3 }
  }
  d x^{ \color{darkblue} \mu_3 }
  \wedge
  F_{
    {\color{orangeii} \nu_2 }
    { \color{darkblue} \mu_4 }
  }
  d x^{ \color{darkblue} \mu_4 }
  \wedge
  \epsilon^{
    \color{orangeii}
    \nu_1 \nu_2 \nu _3 \nu_4
  }
  F_{
    \color{orangeii} \nu_3 \nu_4
  }
  \wedge
  d x^4
  \wedge
  d x^{5'}
  \\
  & =
  \tfrac{3}{2}
  \epsilon^{
    \color{darkblue}
    \mu_1 \mu_2 \mu_3 \mu_4
  }
  F_{
    \color{darkblue} \mu_1 \mu_2
  }
  F_{
    { \color{orangeii} \nu_1 }
    { \color{darkblue} \mu_3 }
  }
  F_{
    {\color{orangeii} \nu_2 }
    { \color{darkblue} \mu_4 }
  }
  \epsilon^{
    \color{orangeii}
    \nu_1 \nu_2 \nu _3 \nu_4
  }
  F_{
    \color{orangeii} \nu_3 \nu_4
  }
  \wedge
  \mathrm{dvol}_6
  \\
  & =
  \tfrac{3}{2}
  \left(
    \frac{
      F \wedge F
    }{
      \mathrm{dvol}_4
    }
  \right)^2
  \mathrm{dvol}_{6'}\;.
  \end{aligned}
\end{equation}
Here in the first step we used that
$
  \sigma^\ast\big(e_{\mu 5'} \wedge e^\mu\big) = d (A_\mu) \wedge d x^\mu
  = F
$
from \eqref{PhotonFieldComponent}
has no factor of $d x^4$, while
$
  \sigma^\ast\big(
    e_{\mu_1 \mu_2}
  \big)
  \;=\;
  \tfrac{1}{2} (\star_{{}_4})_{\mu_1 \mu_2} d x^4
$
from \eqref{DualPhotonField}
does have this factor, so that in the triple wedge product
(with the orange indices) the latter term contributes only linearly,
and as such once in each of the three factors.
The last step is the following computation:
\begin{equation}
  \label{AnAlternativeFormOfFirstDBICorrection}
  \hspace{-3mm}
  \begin{aligned}
  &
  F_{ \color{darkblue} \mu_1 \mu_2 }
  \epsilon^{ \color{darkblue} \mu_1 \mu_2 \mu_3 \mu_4 }
  F_{ {\color{darkblue} \mu_3 } {\color{greenii} \nu_1 } }
  F_{ {\color{darkblue} \mu_4 } {\color{greenii} \nu_2 } }
  \epsilon^{ \color{greenii} \nu_1 \nu_2 \nu_3 \nu_4 }
  F_{ \color{greenii} \nu_3 \nu_4}
  \\
   & =
  \phantom{+}
  4
  \,
  E_{ \color{darkblue} i }
  \epsilon^{ \color{darkblue} i j k }
  \epsilon_{ {\color{darkblue} j } {\color{greenii} j' } j'' }
  B^{j''}
  \epsilon_{ {\color{darkblue} k } {\color{greenii} k' } k'' }
  B^{k''}
  \epsilon^{ \color{greenii} i' j' k' }
  E_{ \color{greenii} i' }
  \;+\;
  4\,
  B^{ l }
  \epsilon_{ l \color{darkblue} i j }
  \epsilon^{ \color{darkblue} i j k }
  \underset{
    = 0
  }{
  \underbrace{
  E_{ \color{greenii} i' }
  B^{ l' } \epsilon_{ l' { \color{darkblue} k } {\color{greenii} j' } }
  \epsilon^{ \color{greenii} i' j' k' }
  E_{ \color{greenii} k'}
  }
  }
  \;+\;
  4\,
  B^{ \color{darkblue} i}
  E_{ \color{greenii} i'}
  E_{ \color{darkblue} i }
  B^{ \color{greenii} i' }
  \\
  & = \phantom{+}\;
  4\,
  E_{ \color{darkblue} i }
  B^{i}
  \underset{
    2\, E_{i'} B^{i'}
  }{
  \underbrace{
  \epsilon_{ {\color{darkblue} k } {\color{greenii} k' } k'' }
  B^{k''}
  \epsilon^{ \color{greenii} i' k k' }
  E_{ \color{greenii} i' }
  }
  }
  -
  E_{ \color{darkblue} i }
  \underset{
    = 0
  }{
  \underbrace{
  B^{k}
  \epsilon_{ {\color{darkblue} k } {\color{greenii} k' } k'' }
  B^{k''}
  }
  }
  \epsilon^{ \color{greenii} i' i k' }
  E_{ \color{greenii} i' }
  \;-\;
  4\,
  B^{ \color{darkblue} i}
  E_{ \color{greenii} i'}
  E_{ \color{darkblue} i }
  B^{ \color{greenii} i' }
  \\
  & = \phantom{+}
  4 (E \cdot B)^2
  \;=\;
  \left(
    \frac{
      F \wedge F
    }{
      \mathrm{dvol}_4
    }
  \right)^2 \!.
  \end{aligned}
\end{equation}
Here in the first step the three summands shown are the
contributions from two, one or no occurrences, respectively,
of the index ``0'' in the two outer factors of $F$.
\hfill \end{proof}

We close by highlighting how
such non-linear corrections to electromagnetism are to be expected in
a theory of hadrodynamics.

\medskip
\noindent {\bf Born-Infeld correction.}
The term \eqref{FirstBICorrectionEmerges} is of the form of the
first correction in Born-Infeld electromagnetism \cite[p. 437]{BornInfeld34}
(as reviewed in \cite[(22)]{Savvidy99}\cite[9.4]{Nastase15})
\begin{equation}
  \label{DBIAction}
  \mathbf{L}^{\mathrm{BI}}
  \;:=\;
  \sqrt{
    - \mathrm{det}(\eta + 2\pi \ell_s^2  F)
  }
  \,
  \mathrm{dvol}_4\;.
\end{equation}
\begin{equation}
  \label{DBICorrection}
  \hspace{-2cm}
  \begin{aligned}
 \mathrm{with} \;\;\; \mathrm{det}
  \big(
    (\eta_{\mu \nu})
    +
    (F_{\mu \nu})
  \big)
  & \; = \phantom{+}
  -1
    -
    \tfrac{1}{2}
    \underset{
      \mathclap{
      \mbox{
        \tiny
        \color{darkblue} \bf
        \begin{tabular}{c}
          ordinary
          Maxwell
          term
        \end{tabular}
      }
      }
    }{
      \underbrace{
        (F \wedge \star_{{}_4} F)
        / \mathrm{dvol}_4
      }
    }
    +
    \underset{
      \mathclap{
      \mbox{
        \tiny
        \color{darkblue} \bf
        \begin{tabular}{c}
          DBI
          correction
          term
        \end{tabular}
      }
      }
    }{
      \underbrace{
        \big(
          \tfrac{1}{2}
          (F\wedge F) / \mathrm{dvol}_4
        \big)^2
      }
    }
  \\
  & \; = \phantom{+}
  -1
  \,+\;\;\;\;\;\;\,
  E \cdot E - B \cdot B
  \;\;\;\;+\;\;\;\;\;\;\;\;\;\;\;
  (B \cdot E)^2.
  \end{aligned}
\end{equation}
which string perturbation theory suggests
appears in the low-energy effective action
on D-branes \cite{FradkinTseytlin85}\cite{ACNY87} \cite{Leigh89}
(reviewed in
\cite{Savvidy99}\cite{Tseytlin00}\cite{Schwarz01}\cite[9.4]{Nastase15}).
Notice that DBI-action \eqref{DBIAction} encodes a critical value
of the electric field strength
\begin{equation}
  \label{CriticalDBIFieldStrength}
  E_{\mathrm{BI}}
  \;:=\;
  T
  \sqrt{
    \frac{
      T^2 + B^2
    }
    {
      T^2 + B_\parallel^2
    }
  }
\end{equation}
(as in \cite[(2.6)]{HashimotoOkaSonoda15}) since
\begin{equation}
  \begin{aligned}
       -
    \mathrm{det}
    \big(
      (\eta_{\mu \nu})
      +
      \tfrac{1}{T}
      (F_{\mu \nu})
    \big)
    \geq  0 \qquad
  & \Leftrightarrow
  \qquad
  E
  \;\leq\;
  T
  \sqrt{
    \frac{
      T^2 + B^2
    }
    {
      T^2 + B_\parallel^2
    }
  }
  \end{aligned}
\end{equation}
where
$B_\parallel := \tfrac{1}{\sqrt{E \cdot E}} B \cdot E$
for the component of the magnetic field parallel to the electric field.

\medskip
We may observe that the same critical field strength is implied
by electromagnetic vacuum polarization
if the electromagnetic field is assumed to fill a leptonic/hadronic
vacuum:

\medskip
\noindent {\bf The Schwinger effect.}
Consider a constant electromagnetic field
$(\vec E, \vec B)$ on 4d Minkowski spacetime,
away from the measure-0 subset of configurations with
$\vec E \cdot \vec B = 0$, hence assuming that $\vec E \cdot \vec B \neq 0$.
Then there exists a Lorentz transformation
$(\vec E, \vec B) \mapsto (\vec E', \vec B')$ such that
$\vec E' \parallel \vec B'$.
The \emph{Schwinger effect} of vacuum polarization,
in this case, predicts that
in a theory with electrically charged particles of
charge $e$ and mass $m$
(reviewed for electron/positron pair creation in \cite[(1.28)]{Dunne05},
and for quark/anti-quark pair creation in \cite[(2)]{HIS11})
an electric field strength around the \emph{Schwinger limit} scale
\begin{equation}
  \label{CriticalFieldStrength}
  E'_{\mathrm{crit}}
  \;:=\;
  \frac{
    m^2 \, c^3
  }{
    e\, \hbar
  }
\end{equation}
(reviewed in \cite[(1.3)]{Dunne05}\cite[(40)]{Martin08})
causes a sizable density of particle/anti-particle pairs
to be created out of the vacuum,\footnote{
Even without the detailed formula for the Schwinger mechanism
decay rate, in any of its variants,
one may understand the critical electric field strength \eqref{CriticalFieldStrength}
as that whose work done by its Lorentz force $e E'_{\mathrm{crit}}$
over the Compton wavelength $\lambda_{m} := \hbar / m c$
equals the rest mass $m c^2$ of the given particles:
$
  E'_{\mathrm{crit}}
  \;=\;
  \frac{m c^2}{e \lambda_m}
  \,.
$
}
causing non-linear corrections to the electromagnetic
dynamics or even decay of the vacuum, and hence in any case
a breakdown of ordinary electromagnetism.

\medskip
Now observe that a Lorentz-invariant expression
for the field strength $E'$ above,
as a function of the electromagnetic field $(\vec E, \vec B)$ in any frame,
is
given by the following expression (see also \cite[(1.6)]{Dunne05}):
$$
  E'(\vec E, \vec B)
  \;=\;
  \sqrt{
    \sqrt{
      \Big(
      \tfrac{1}{2}
      \big(
        \vec E \cdot \vec E - \vec B \cdot \vec B
      \big)
      \Big)^2
      +
      \big(
        \vec B \cdot \vec E
      \big)^2
    }
    +
    \tfrac{1}{2}
    \big(
      \vec E \cdot \vec E
      -
      \vec B \cdot \vec B
    \big)
  }\;.
$$
This follows immediately from the fact that this
expression is a Lorentz invariant (being a function of the
basic Lorentz invariants $\vec E\cdot \vec E - \vec B \cdot \vec B$
and $\vec E \cdot \vec B$, e.g. \cite{EscobarUrrutia14})
and that it evidently reduces to the absolute value
$E' = \sqrt{ \vec E' \cdot \vec E' }$ in any
Lorentz frame where $\vec E' \parallel \vec B'$.
But then basic algebra reveals (\cite[(2.6)]{HashimotoOkaSonoda15})
that
$$
  E'(\vec E, \vec B)
  \;=\;
  E'_{\mathrm{crit}}
  \qquad
  \Leftrightarrow
\qquad
  E \;=\;
  E'_{\mathrm{crit}}
  \;
  \sqrt{
    \frac{
      {E'}_{\!\mathrm{crit}}^2 + B^2
    }{
      {E'}_{\!\mathrm{crit}}^2 + B_{\parallel}^2
    }
  }\;.
$$
This is exactly the formula for the critical field strength
\eqref{CriticalDBIFieldStrength}
in Born-Infeld -theory,  if we identify
the string tension with the Schwinger limit
field strength $T = E'_{\mathrm{crit}}$,
as befits the picture that it is the flux tube strings
connecting a quark/anti-quark which counteracts their
separation.

\medskip

\section{Conclusion}
\label{ConcludingRemarks}

 Irrespective of the physics interpretation of the full field content
in the big table \eqref{TheFieldContent},
found to emerge in the super-exceptional M5-brane model \cite{FSS19d},
we have identified elusive field theoretic structure
in the previously somewhat mysterious expression \eqref{SuperExceptionalHFlux}
for the super-exceptional 3-flux form
(the ``hidden superalgebra of 11d supergravity''
\cite{DAuriaFre82}\cite{BAIPV04}\cite{ADR16})
including:
\begin{enumerate}[{\bf (i)}]
\vspace{-2mm}
\item a Skyrme current-type term (Lemma \ref{EmergenceOfTraceOverInternalIndices})
and
\vspace{-2mm}
\item the first DBI-correction term (Prop. \ref{FirstBICorrection}).
\end{enumerate}
\vspace{-.25cm}

\noindent Specifically the DBI-term \eqref{TheFirstDBICorrectionTerm} 
is the expected first-order interaction correction
to the free Perry-Schwarz Lagrangian \eqref{TheMaxwellLagrangian}
in the full interacting M5-brane model \cite[(63)]{PerrySchwarz97}
as seen explicitly after double dimensional reduction to the D4-brane model
\cite[6]{APPS97a}\cite[6]{APPS97b}. 

\vspace{.1cm}

\noindent In the same manner one may now identify further 
interaction terms in the full super-exceptinal PS-Lagrangian 
\eqref{SuperExceptionalPSLagrangianSummands}.
For instance, following Example \ref{OneFermionPionCouplingTerm}
there are evident 
vector-, axial- and isospin-current terms coupling the two fermion iso-doublets
to the various vector fields. This is to be discussed elsewhere.

\medskip

Thereby one possible interpretation of the field content
in \cref{FieldContent} certainly suggests itself:
We seem to be seeing the emergence of a variant/sector of
$\mathrm{SU}(2)$-flavor chiral hadrodynamics
(e.g. \cite{SerotWalecka92}\cite{Scherer03}\cite{BM07}),
on the single M5-brane at an $\mathrm{A}_1$-singularity,
in a way that is different from but akin to
(see also Remark \ref{RoleOfKKModeExpansion})
the hadrodynamics seen
on flavor D-branes \cite{KarchKatz02} in holographic QCD models
such as the Witten-Sakai-Sugimoto (WSS) model \cite{Rebhan14}\cite{RhoEtAl16}.
\footnote{Recently is has been pointed out \cite{IvanovaLechtenfeldPopov18}
that the core mechanism of the hadron sector in the
WSS model, which is KK-reduction of 5d Yang-Mills theory,
also appears on the M5-brane.}

\medskip
This seems remarkable, since the
only input for the super-exceptional M5-brane model \cref{SuperExceptionalMGeometry}
is the pair of
M-brane super-cocycles \eqref{H3exTrivializesPullbackOf4Flux} and \eqref{M5BraneCocycle},
which jointly form a single cocycle in super rational Cohomotopy theory
\cite[2.5]{Sati13}\cite{FSS15}\cite[3.2]{HSS18} (review in \cite[(57)]{FSS19a}).
Elsewhere we have shown that taking seriously this cohomotopical
nature of the M-brane charges (``Hypothesis H'')
implies a fair bit of topological/homotopical structure expected in M-theory
in general \cite{FSS19b}\cite{SS19a} and for the M5-brane specifically
\cite{FSS19c}\cite{SS19a} (reviewed in \cite{Sc20}). In particular,
it implies \cite{FSS20a}\cite{Roberts20}
the emergence of a topological sector of
an $\mathrm{SU}(2)_R$-gauge field on the heterotic M5-brane,
similar to that considered in \cite{SaemannSchmidt17}, whose
origin matches that of the $\mathrm{SU}(2)_R$-valued local differential form
data found here (see Remark \ref{SkyrmeBaryonCurrent}).
All taken together, we seem to be seeing a coherent non-abelian
M5-brane model specifically for single heterotic M5-branes with
$\mathrm{SU}(2)$-flavor gauge fields on their worldvolume.

\medskip

This raises the question of the extent to which
the $\mathrm{SU}(2)$-flavor gauge theory emerging in
the super-exceptional cohomotopical M5-brane model
might be related to actual real-world hadrodynamics,
possibly as a version in the deep M-theoretic regime (on $N =1$ brane!)
of the D-brane models for holographic QCD that
have to work in the (unrealistic) $N \gg 1$ supergravity regime.

\vspace{-2mm}
\begin{enumerate}[{\bf (i)}]
\item On the one hand, the super-exceptional
PS-Lagrangian \eqref{SuperExceptionalPSLagrangianSummands} contains
a wealth of interaction terms, all of which being
Spin-, Isospin- and Hypercharge-invariant
(due to Theorem \ref{FieldsAtSingularity},
by $\mathrm{Spin}(10,1)$-invariance of the unreduced model),
hence identifiable
among the list
\cite[5.2]{OllerVerbeniPrades06}\cite[Tab. 1]{FrinkMeissner06}\cite{Geng13}
of possible interactions in the effective field theory of hadrons.
Interestingly, the super-exceptional M5-brane model organizes these
hadron interaction terms in
a supersymmetric \eqref{SuperExceptionalBianchiIdentities}
KK-tower \eqref{KKExpansion}.
Notice here that while both

\begin{enumerate}[{\bf (a)}]
\vspace{-3mm}
\item  supersymmetry and
\vspace{-1mm}
\item KK-modes
\end{enumerate}

\vspace{-3mm}
\noindent remain notoriously hypothetical
in the unconfined color-charged sector of the standard model,
both are in fact observed, to a reasonable degree
of accuracy, in the confined flavor-sector of nature.
This is, respectively,
the phenomenon of

\begin{enumerate}[{\bf (a)}]
\vspace{-3mm}
\item  {\it hadron supersymmetry}
\cite{Miyazawa66}\cite{Miyazawa68}\cite{CG85}
\cite{CG88}
(reviewed in \cite{Lichtenberg99}\cite{deTeramond17}\cite{Brodsky18}\cite{Brodsky19})
and
\vspace{-1mm}
\item  the success of the {\it dimensional deconstruction}
of vector meson models \cite{SonStephanov03}\cite[5.2]{SakaiSugimoto04}\cite[3.3]{SakaiSugimoto04}
in holographic QCD (see \cite[Fig. 15.2]{Sugimoto16}).
\end{enumerate}

\vspace{-4mm}
 \item On the other hand, various standard Lagrangian terms
considered in the literature
on chiral perturbation theory and quantum hadrodynamics
are clearly missing in the super-exceptional PS-Lagrangian \eqref{SuperExceptionalPSLagrangianSummands};
not the least the kinetic terms for the mesons.
But since the one kinetic term that does appear, the one for the photon
\eqref{MaxwellLagrangianAppears},
crucially appears from the self-duality constraint \eqref{DualPhotonField}
on the ordinary 3-form flux, this might just indicate that the super-exceptional
version of the 3-flux self-duality remains to be understood and
to be implemented.
\footnote{Note that in the presence of $H_3$ one can consider a noncommutative worldvolume
for the M5-brane -- see \cite{MS}.}
\end{enumerate}

\vspace{-2mm}
It is worth recalling that none of the existing models of hadrodynamics,
be it
chiral perturbation theory,
quark-bag models,
vector meson dominance,
Walecka-type QHD models etc.,
have been derived from fundamental
principles -- notably not from QCD.
This is the open ``Confinement Problem'' \cite{Greensite11},
one of the ``Millennium Problems'' \cite{CMI}\cite{JaffeWitten},
the ``Holy Grail'' of nuclear physics
\cite{Holstein99}\cite[13.1.9]{Guidry08}. Instead
they are all models in effective field theory 
using clues from experiment.
Therefore, even a partial hint for a possible emergence of
confined hadrodynamics from first (cohomotopical) principles of
M5-brane theory seems profound and interesting.

\paragraph{\large Acknowledgements.}
The authors thank Laura Andrianopoli and Dmitri Sorokin 
for useful comments on an earlier version of this article.
We also thank the participants of the NYUAD
workshop {\it M-theory and Mathematics} in January 2020
for useful discussions, especially Neil Lambert, Christian Saemann
and Dmitri Sorokin.

\vspace{.5cm}
\noindent Domenico Fiorenza,  {\it Dipartimento di Matematica, La Sapienza Universita di Roma, Piazzale Aldo Moro 2, 00185 Rome, Italy.}

 \medskip
\noindent Hisham Sati, {\it Mathematics, Division of Science, New York University Abu Dhabi, UAE.}

 \medskip
\noindent Urs Schreiber,  {\it Mathematics, Division of Science, New York University Abu Dhabi, UAE; on leave from Czech Academy of Science, Prague.}

\end{document}